\documentclass[article]{imsart}

\input{settings.tex}
\startlocaldefs
\theoremstyle{plain}

\newtheorem{theorem}{Theorem}[section]
\newtheorem{lemma}[theorem]{Lemma}
\theoremstyle{definition}
\newtheorem{definition}[theorem]{Definition}

\theoremstyle{remark}

\endlocaldefs

\newcommand{\E}{\mathbb{E}}             
\newcommand{\R}{\mathbb{R}}             
\newcommand{\I}{\mathbb{I}}             
\renewcommand{\P}{\mathbb{P}}           
\newcommand{\D}{\mathcal{D}}            
\newcommand{\Q}{\mathcal{Q}}            
\renewcommand{\S}{\mathcal{S}}          
\renewcommand{\L}{\mathcal{L}}          
\newcommand{\F}{\mathcal{F}}          
\newcommand{\Xc}{\mathcal{X}}           
\newcommand{\KL}{D_\text{KL}}             

\newcommand{\sigmoid}{s}
\newcommand{\IG}{\Gamma^{-1}}

\newcommand{\argmin}{{\arg\!\min}}      
\newcommand{\argmax}{{\arg\!\max}}      

\newcommand{\red}[1]{#1}
\newcommand{\blue}[1]{#1}

\DeclareMathOperator{\corr}{corr}

\DeclareMathOperator{\diag}{diag}

\DeclareMathOperator{\tr}{tr}

\newcommand{\eps}{\varepsilon}
\newcommand{\mm}{m_{\text{max}}}
\newcommand{\Tr}{\text{Tr}}
\newcommand{\T}{\mathcal{T}}

\newtheorem{remark}{Remark}

\newtheorem{innercustomthm}{Assumption}
\newenvironment{customassumption}[1]
  {\renewcommand\theinnercustomthm{#1}\innercustomthm}
  {\endinnercustomthm}

\arxiv{2309.10378}

\begin{document}

\begin{frontmatter}
\title{Group Spike-and-Slab Variational Bayes}
\runtitle{Group Spike-and-Slab Variational Bayes}

\begin{aug}
\author[A]{\fnms{Michael}~\snm{Komodromos}\ead[label=e1]{mk1019@ic.ac.uk}},
\author[A]{\fnms{Marina}~\snm{Evangelou}\ead[label=e2]{m.evangelou@ic.ac.uk}},
\author[A]{\fnms{Sarah}~\snm{Filippi}\ead[label=e3]{s.filippi@ic.ac.uk}}
\and
\author[A]{\fnms{Kolyan}~\snm{Ray}\ead[label=e4]{kolyan.ray@ic.ac.uk}}

\address[A]{Depart of Mathematics,
Imperial College London\printead[presep={,\ }]{e1}}
\runauthor{M. Komodromos et al.}
\end{aug}

\begin{abstract}
We introduce Group Spike-and-Slab Variational Bayes (GSVB), a scalable method for group sparse regression. A fast co-ordinate ascent variational inference (CAVI) algorithm is developed for several common model families including Gaussian, Binomial and Poisson. Theoretical guarantees for our proposed approach are provided by deriving contraction rates for the variational posterior in grouped linear regression. 
Through extensive numerical studies, we demonstrate that GSVB provides state-of-the-art performance, offering a computationally inexpensive substitute to MCMC, whilst performing comparably or better than existing MAP methods. Additionally, we analyze three real world datasets wherein we highlight the practical utility of our method, demonstrating that GSVB provides parsimonious models with excellent predictive performance, variable selection and uncertainty quantification. 
\end{abstract}


\end{frontmatter}

 \newcommand{\blacklinethick}{\raisebox{3pt}{%
    \protect\tikz{ \protect\draw[-,black,solid,line width = 1.2pt](0,0) -- (8mm,0);}}}
 \newcommand{\thickline}[1]{\raisebox{2pt}{%
    \protect\tikz{ \protect\draw[-,#1,solid,line width = 1.2pt](0,0) -- (5mm,0);}}}
 \newcommand{\dashedline}[1]{\raisebox{2pt}{%
    \protect\tikz{ \protect\draw[-,#1,dashed,line width = 0.8pt](0,0) -- (5mm,0);}}}
\newcommand{\pointFilled}[1]{\raisebox{1.2pt}{%
   \protect\tikz{ \protect\draw[-,#1,fill=#1](0,0) circle (1pt);}}}

\definecolor{tabgreen}{HTML}{2ca02c}
\definecolor{tabblue}{HTML}{1f77b4}
\definecolor{taborange}{HTML}{ff7f0e}
\definecolor{magenta}{HTML}{ff00ff}

\definecolor{setting1}{HTML}{FFFFFF} 
\definecolor{setting2}{HTML}{AACBAA} 
\definecolor{setting3}{HTML}{559755} 
\definecolor{setting4}{HTML}{006400}
\definecolor{gsvb-j}{HTML}{9932CC} 
\definecolor{gsvb-d}{HTML}{4DAF4A} 
\definecolor{gsvb-b}{HTML}{E41A1C} 
\definecolor{ssgl}{HTML}{377EB8} 
\definecolor{mcmc}{HTML}{FF7F00} 
\definecolor{dnf}{HTML}{FF0000} 
\newcommand{\colorSquareB}[1]{\textcolor{black}{\setlength{\fboxsep}{0pt}\fbox{\textcolor{#1}{\rule{7pt}{7pt}}}}}
\newcommand{\colorSquare}[1]{\textcolor{#1}{\rule{7pt}{7pt}}}

\newcommand{\blacklinethin}{\raisebox{3pt}{%
    \protect\tikz{\protect\draw[-,black,solid,line width = 0.4pt](0,0) -- (8mm,0);}}}
\newcommand{\whitediamond}{{$\diamond$}}
\newcommand{\thickcolorline}[1]{\raisebox{3pt}{\protect\tikz{\protect\draw[-,#1,solid,line width = 1.5pt](0,0) -- (8mm,0);}}}

\newcommand{\thickredline}{\raisebox{3pt}{\protect\tikz{\protect\draw[-,red,solid,line width = 1.5pt](0,0) -- (8mm,0);}}}
\definecolor{p90}{HTML}{006400} 
\definecolor{p95}{HTML}{559755} 
\definecolor{p99}{HTML}{aacbaa}
\definecolor{p90B}{HTML}{00008b} 
\definecolor{p95B}{HTML}{5555b1} 
\definecolor{p99B}{HTML}{aaaad8}
\newcommand{\thickgreen}{\raisebox{3pt}{\protect\tikz{\protect\draw[-,p90!70,solid,line width = 1.5pt](0,0) -- (8mm,0);}}}
\newcommand{\thickblue}{\raisebox{3pt}{\protect\tikz{\protect\draw[-,p90B!70,solid,line width = 1.5pt](0,0) -- (8mm,0);}}}
\newcommand{\redline}{\raisebox{2pt}{\protect\tikz{\protect\draw[-,red,dotted,line width = 0.9pt](0,0) -- (8mm,0);}}}
\newcommand{\blackline}{\raisebox{2pt}{\protect\tikz{\protect\draw[-,black,solid,line width = 1.2pt](0,0) -- (8mm,0);}}}
\newcommand{\blackdot}{{\color{black} $\bullet$}}
\newcommand{\greydot}{{\color{lightgray} $\bullet$}}

\newpage
\section{Introduction}

Group structures arise in various applications, such as genetics \citep{Wang2007, Breheny2009}, imaging \citep{Lee2021}, multi-factor analysis of variance \citep{Meier2008}, non-parametric regression \citep{Huang2010a}, 
and multi-task learning, 
among others. In these settings, $p$-dimensional feature vectors, 
$x_i = (x_{i1}, \dots, x_{ip})^\top \in \mathbb{R}^p$, for $i=1,\dots,n$ observations
can be partitioned into groups of features. Formally this means that we can construct sub-vectors, $x_{G_k} = \{x_{j} : j \in G_k \}$ for $k=1,\dots, M$, 
where the groups $G_k = \{ G_{k,1}, \dots, G_{k, m_k} \}$ are disjoint sets of indices satisfying $\bigcup_{k=1}^M G_k = \{1, \dots, p \}$. Typically, these group structures are known beforehand, for example in genetics where biological pathways (gene sets) are known, or they are constructed artificially, for example, through a basis expansion in non-parametric additive models. 

In the regression setting incorporating these group structures is crucial, as disregarding them can result in sub-optimal models \citep{Huang2010, Lounici2011}. 
In this manuscript, we focus on the general linear regression model (GLM) where, for each observation $i=1,\dots, n$, the response, $Y_i$   can be modelled by
\begin{equation} \label{eq:model} 
    \E[Y_i | x_i, \beta] = f \left( \sum_{k=1}^M x_{i, G_k}^\top \beta_{G_k} \right), \quad i = 1, \dots, n
\end{equation}
where $f: \mathbb{R} \rightarrow \mathbb{R}$ represents the link function, $\beta = (\beta_1, \dots, \beta_p)^\top \in \mathbb{R}^p$ denotes the model coefficient vector with $\beta_{G_k} = \{\beta_{j} : j \in G_k \}$. 
Beyond incorporating the group structure,
it is often of practical importance to identify the groups of features that are associated with the response. This holds particularly true when there are a large number of them. To address this, various methods have been proposed over the years, with one of the most popular being the group LASSO \citep{Yuan2006}, which applies an $\ell_{2, 1}$ norm to groups of coefficients. Following \cite{Yuan2006} there have been numerous extensions, including the group SCAD \citep{Wang2007}, the group LASSO for logistic regression \citep{Meier2008}, the group bridge \citep{Huang2009}, group LASSO with overlapping groups \citep{Jacob2009a}, 
among others (see \cite{Huang2012} for a detailed review of frequentist methods).



In a similar vein, Bayesian group selection methods have arisen. The earliest of which being the Bayesian group LASSO \citep{Raman2009, Kyung2010}, 
which uses a multivariate double exponential distribution prior to impose shrinkage on groups of coefficients. Notably, the maximum a posteriori (MAP) estimate under this prior coincides with the estimate under the group LASSO. Other methods include Bayesian sparse group selection \citep{Xu2015, Chen2016} and the group spike-and-slab LASSO \citep{Bai2020} 
which approach the problem via stochastic search variable selection \citep{Mitchell1988, Chipman1996}. 
Formally, these methods utilize a group spike-and-slab prior, a mixture distribution of a multivariate Dirac mass on zero and a continuous distribution over $\R^{m_k}$, where $m_k$ is the size of the $k$th group. Such priors have been shown to work excellently for variable selection as they are able to set coefficients exactly to zero, avoiding the use of shrinkage to enforce sparsity. For a comprehensive review see \cite{Lai2021} and \cite{Jreich2022}.

However, a serious drawback of these methods is that most use Markov Chain Monte Carlo (MCMC) to sample from the posterior \citep{Raman2009, Xu2015, Chen2016}. In general, MCMC approaches are known to be computationally expensive when there are a large number of variables, and often exhibit poor mixing and slow convergence.
To circumvent these issues, some authors proposed computing MAP estimates \citep{Kyung2010, Bai2020}, by relaxing the form of the prior, replacing the multivariate Dirac mass with a continuous distribution concentrated at zero \citep{Rockova2018}. Although these algorithms are fast to compute, they come at the sacrifice of interpretability, as posterior inclusion probabilities no longer guarantee the coefficient is zero but rather concentrated at zero. Beyond this, these algorithms only return a point estimate for $\beta$ and therefore do not provide uncertainty quantification -- a task at the heart of Bayesian inference.

To bridge the gap between scalability and uncertainty quantification several authors have turned to variational inference (VI). An approach to inference wherein the posterior distribution is approximated by a tractable family of distributions known as the variational family (see \cite{Zhang2019} for a review). In the context of high-dimensional Bayesian inference, VI has proven particularly successful, and has been employed in linear regression \citep{Carbonetto2012, Ormerod2017, RS22},  logistic regression \citep{Ray2020} and survival analysis \citep{Komodromos2021} to name a few. Within the context of Bayesian group regression, to our knowledge only the Bayesian group LASSO has seen a variational counterpart \citep{Babacan2014}.

In this manuscript we examine the variational Bayes (VB) posterior arising from the group spike-and-slab prior with multivariate double exponential slab and Dirac spike, referring to our method as Group Spike-and-slab Variational Bayes (GSVB).
 We provide scalable variational approximations to three common classes of generalized linear models: the Gaussian with identity link function, Binomial with logistic link function, and Poisson with exponential link function. We outline a general scheme for computing the variational posterior via co-ordinate ascent variational inference. We further show that for specific cases, the variational family can be re-parameterized to allow for more efficient updates.

%

Through extensive numerical experiments we demonstrate that GSVB achieves state-of-the-art performance while significantly reducing computation time by several orders of magnitude compared to MCMC.
Moreover, through our comparison with MCMC, we highlight that the variational posterior provides excellent uncertainty quantification through marginal credible sets with impressive coverage. Additionally, the proposed method is compared against the spike-and-slab group LASSO \citep{Bai2020}, a state-of-the-art Bayesian group selection algorithm that returns MAP estimates. Within this comparison our method demonstrates comparable or better performance in terms of group selection and effect estimation, whilst carrying the added benefit of providing uncertainty quantification, a feature not available by other scalable methods in the literature. 

To highlight the practical utility of our method, we analyse three real datasets, demonstrating that our method provides excellent predictive accuracy, while also achieving parsimonious models. Additionally, we illustrate the usefulness of the VB posterior by showing its ability to provide posterior predictive intervals, a feature not available to methods that provide MAP estimates, and computationally prohibitive to compute via MCMC.

Theoretical guarantees for our method are provided in the form of posterior contraction rates, which quantify how far the posterior places most of its probability from the `ground truth' generating the data for large sample sizes. This approach builds on the work of \cite{CSV15} and \cite{RS22}, and extends previous Bayesian contraction rate results for true posteriors in the group sparse setting \citep{NJG20,Bai2020} to their variational approximation.

Concurrent to our work, \cite{lin2023} proposed a similar spike-and-slab model for group variable selection in linear regression, which appeared on arXiv shortly after our preprint.

\noindent \textbf{Notation.} Let $y = (y_1, \dots, y_n)^\top \in \R^n$ denote the realization of the random vector $Y = (Y_1, \dots, Y_n)$. Further, let $X = (x_1, \dots, x_n)^\top \in \R^{n \times p}$ denote the design matrix, where for a group $G_k$ let  $X_{G_k} = (x_{1, G_k}, \dots, x_{n, G_k})^\top \in \R^{n \times m_k}$. Similarly, for $G_k^c = \{1,\dots, p \} \setminus G_k$, denote $X_{G_k^c} = (x_{1, G_k^c}, \dots, x_{n, G_k^c})^\top \in \R^{n \times (p - m_k)}$ where $x_{i, G_k^c} = \{x_{ij} : j \in G_k^c \}$, and $\beta_{G_k^c} = \{ \beta_j : j \in G_k^c \}$. Wherein, without loss of generality we assume that the elements of the groups are ordered such that $1 = G_{1,1} < G_{1,2} < \dots < G_{M, m_M} = p$. Finally, the Kullback-Leibler divergence is defined as $\KL = \KL(Q \| P) = \int_{\Xc} \log \left( \frac{dQ}{dP} \right) dQ $,
where $Q$ and $P$ are probability measures on $\Xc$, such that $ Q $ is absolutely continuous with respect to $P$.

\section{Prior and Variational family}  \label{sec:prior}


To model the coefficients $\beta$, we consider a group spike-and-slab prior. For each group $G_k$, the prior over $\beta_{G_k}$ consists of a mixture distribution of a multivariate Dirac
mass on zero and a multivariate double exponential distribution,  
$\Psi(\beta_{G_k})$, whose density is given by, 
\begin{equation*}
 \psi(\beta_{G_k}; \lambda) = C_k \lambda^{m_k} \exp \left( - \lambda \| \beta_{G_k} \| \right)    
\end{equation*}
where $ C_k = [ 2^{m_k} \pi^{(m_k -1)/2} \Gamma ( (m_k + 1) /2 ) ]^{-1} $ and $ \| \cdot \| $ is the $\ell_2$-norm.
In the context of sparse Bayesian group regression the multivariate double exponential has been previously considered by \cite{Raman2009} and \cite{Kyung2010} as part of the Bayesian group LASSO
and by \cite{Xu2015} 
within a group spike-and-slab prior. 

Formally the prior, which we consider throughout, is given by $\Pi(\beta) = \bigotimes_{k=1}^M \Pi_k(\beta_{G_k})$, where each $\Pi_k(\beta_{G_k})$ has the hierarchical representation,
\begin{equation} \label{eq:prior}
\begin{aligned}
    \beta_{G_k} | z_k \overset{\text{ind}}{\sim} &\ z_k \Psi(\beta_{G_k}; \lambda) + (1-z_k) \delta_0(\beta_{G_k}) \\
    z_k | \theta_k \overset{\text{ind}}{\sim} &\ \text{Bernoulli}(\theta_k) \\
    \theta_k \overset{\text{iid}}{\sim} &\ \text{Beta}(a_0, b_0)
\end{aligned}
\end{equation}
for $k=1,\dots,M$, where $\delta_0(\beta_{G_k})$ is the multivariate Dirac mass on zero with dimension $m_k = \dim(\beta_{G_k})$. 
In conjunction with the log-likelihood, $\ell(\D; \beta)$ for a dataset $\D = \{(y_i, x_i) \}_{i=1}^n$, we write the posterior density under the given prior as,
\begin{equation} \label{eq:posterior} 
d\Pi(\beta | \D) = \Pi_D^{-1} e^{\ell(\D; \beta)} d\Pi(\beta)
\end{equation}
where $\Pi_\D = \int_{\R^{p}} e^{\ell(\D; \beta)} d\Pi(\beta)$ is a normalization constant known as the model evidence.


The posterior arising from the prior \eqref{eq:prior} and the dataset $\D$ assigns probability mass to all $2^M$ possible sub-models, i.e. each subset $\S \subseteq \{1, \dots, M\}$, such that $z_k = 1, k \in \S$ and $z_k = 0$ otherwise. 
As using MCMC procedures to sample from this complex posterior distribution is computationally prohibitive, even for a moderate number of groups, we resort to variational inference to approximate it.
This approximation is known as the variational posterior 
which is an element of the variational family $\Q$ given by
\begin{equation} \label{eq:optim} 
\tilde{\Pi} = \underset{Q \in \Q}{\argmin}\ \KL\left( Q \| \Pi(\cdot |\D) \right)\;.
\end{equation}

For our purposes the variational family we consider is a mean-field variational family,
\begin{equation} \label{eq:family_1}
    \Q =
    \left\{ Q(\mu, \sigma, \gamma) = 
	\bigotimes_{k=1}^M Q_k(\mu_{G_k}, \sigma_{G_k}, \gamma_k) :=
	\bigotimes_{k=1}^M 
	\left[ 
	    \gamma_k\ N_k\left(\mu_{G_k}, \diag(\sigma_{G_k}^2) \right) + 
	    (1-\gamma_k) \delta_0
	\right] 
    \right\} 
\end{equation}
where $\mu \in \R^p$ with $\mu_{G_k} = \{ \mu_j : j \in G_k \}$, $\sigma^2 \in \R_+^p$ with $\sigma^2_{G_k} = \{\sigma^2_j : j \in G_k \}$, $ \gamma = (\gamma_1, \dots, \gamma_M)^\top \in [0, 1]^M $. 
$N_k(\mu, \Sigma)$ denotes the multivariate Normal distribution with mean parameter $\mu$ and covariance $\Sigma$. Notably, under $Q \in \Q$, the vector of coefficients for each group $G_k$ is a spike-and-slab distribution where the slab consists of the product of independent Normal distributions,
\begin{equation*}
    \beta_{G_k} \overset{\text{ind}}{\sim} \gamma_k \big[ \bigotimes_{j \in G_k} N(\mu_j, \sigma_j^2) \big] + (1-\gamma_k) \delta_0,
\end{equation*}
meaning the structure (correlations) between elements within the same group are not captured.
To mitigate this, a second variational family is introduced, where the covariance within groups is unrestricted. Formally,
\begin{equation}\label{eq:family_2}
    \Q' =
    \left\{ Q'(\mu, \Sigma, \gamma) = 
	\bigotimes_{k=1}^M Q_k'(\mu_{G_k}, \Sigma_{G_k}, \gamma_k) :=
	\bigotimes_{k=1}^M 
	\left[ 
	    \gamma_k\ N\left(\mu_{G_k}, \Sigma_{G_k} \right) + 
	    (1-\gamma_k) \delta_0
	\right] 
    \right\} 
\end{equation}
where $\Sigma \in \R^{p \times p}$ is a covariance matrix for which $\Sigma_{ij} = 0$, for $i \in G_k, j \in G_l, k \neq l$ and $\Sigma_{G_k} = (\Sigma_{ij})_{i, j \in G_k} \in \R^{m_k \times m_k}$ denotes the covariance matrix of the $k$th group.  Notably, $\Q'$ should provide greater flexibility when approximating the posterior, because unlike $\Q \subset \Q'$, it is able to capture the dependence between coefficients in the same group. The importance of 
which is
highlighted empirically in \Cref{sec:simulations} wherein the two families are compared.

Note that the posterior does not take the form $Q$ or $Q'$, as the use of these variational families replaces the $2^M$ model weights by $M$ VB group inclusion probabilities, $\gamma_k$, thereby introducing substantial additional independence. For example, information as to whether two groups of variables are selected together or not is lost. However, the form of the VB approximation retains many of the interpretable features of the original posterior such as the inclusion probabilities of particular groups.


\section{Computing the variational posterior} \label{c5:sec:method}

\red{
Computing the variational posterior defined in \eqref{eq:optim} relies on optimizing 
a lower bound on the model evidence, $\Pi_\D$, known as the evidence lower bound (ELBO). The ELBO follows from the non-negativity of the Kullback-Leibler divergence and is given by
\begin{equation}
    \L_{Q'}(\D) := \E_{Q'} \left[ \ell(\D; \beta) - \log \frac{dQ'}{d\Pi} \right].
\end{equation} 
Intuitively, the first term evaluates the model's fit to the data, while the second term acts as a regularizer, ensuring the variational posterior is ``close'' to the prior distribution. 
}

Various strategies exist to tackle the maximization of the ELBO,
with respect to $Q' \in \Q'$ \citep{Zhang2019}. One popular approach is co-ordinate ascent variational inference (CAVI), where sets of parameters of the variational family are optimized in turn while the remainder are kept fixed. Although this strategy does not (in general) guarantee the global optimum, it is easy to implement and often leads to scalable inference algorithms. In the following subsections we detail how this algorithm is constructed, noting that derivations are presented under the variational family $\Q'$, as the equations for $\Q \subset \Q'$ follow by restricting $\Sigma_{G_k}$ to be a diagonal matrix.

\subsection{Computing the Evidence lower bound}

To compute the ELBO we begin by deriving the KL divergence between the variational distribution and the prior, $\KL(Q' \| \Pi) = \E_{Q'} \left[ \log(dQ' / d\Pi) \right]$, which we note does not change regardless of the form of the likelihood. To evaluate an expression for this quantity, the group independence structure between the variational distribution and prior is exploited, allowing for the log Radon-Nikodym derivative, $\log \left( dQ'/ d\Pi \right)$, to be expressed as, 
\begin{equation*}
     \log \frac{dQ'}{d\Pi}(\beta) = 
    \sum_{k=1}^M \log \frac{d Q_k'}{d\Pi_k} (\beta_{G_k}) 
    = \sum_{k=1}^M \mathbb{I}_{z_k = 1} \log \frac{\gamma_k dN_k}{\bar{w} d \Psi_k}(\beta_{G_k})
    + \mathbb{I}_{z_k = 0} \log \frac{(1-\gamma_k) d\delta_0}{(1 - \bar{w}) d \delta_0} (\beta_{G_k}).
\end{equation*}
where $\bar{w} = {\alpha_0} / ({\alpha_0 + b_0})$. Subsequently, it follows that, 
\begin{equation} \label{eq:penalty} 
\begin{aligned}
\KL(Q' \| \Pi) = &\
    \sum_{k=1}^M 
\bigg(
    \gamma_k \log \frac{\gamma_k}{\bar{w}}
-
    \frac{\gamma_k}{2} \log( \det (2 \pi \Sigma_{G_k}))
-
    \frac{\gamma_k m_k}{2}
- 
    \gamma_k \log(C_k) \\
- &\
    \gamma_k m_k \log (\lambda)
+
    \E_{Q'} \left[ \I_{z_k =1} \lambda \| \beta_{G_k} \| \right]
+
    (1-\gamma_k) \log \frac{1 - \gamma_k}{1 - \bar{w}}
\bigg)
\end{aligned}
\end{equation}
where we have used the fact that
$\E_{\beta_{G_k} \sim N(\mu_{G_k}, \Sigma_{G_k})} [ (\beta_{G_k} - \mu_{G_k})^\top \Sigma_{G_k}^{-1} (\beta_{G_k} - \mu_{G_k}) ] = m_k$. 

Notably, there is no closed form for $\E_{Q'} \left[ \I_{z_k =1} \lambda \| \beta_{G_k} \| \right]$, meaning optimization over this quantity would require costly Monte Carlo methods.

To circumvent this issue, a surrogate objective is constructed and used instead, which follows by applying Jensen's inequality to 
$\E_{Q'} \left[ \I_{z_k =1} \lambda \| \beta_{G_k} \| \right]$, giving,
\begin{equation}
    \E_{Q'} \left[ \I_{z_k =1} \lambda \| \beta_{G_k} \| \right] = \gamma_k \E_{N_k} \left[ \lambda \| \beta_{G_k} \| \right] \leq \gamma_k \lambda \left( \sum_{i \in G_k} \Sigma_{ii} + \mu_i^2 \right)^{1/2}.
\end{equation}
Thus, $\KL(Q' \| \Pi)$ can be upper-bounded by,
\begin{equation}
\begin{aligned}
  \varrho(& \mu, \Sigma, \gamma) :=
    \sum_{k=1}^M 
\bigg(
    \gamma_k \log \frac{\gamma_k}{\bar{w}}
-
    \frac{\gamma_k}{2} \log( \det (2 \pi \Sigma_{G_k}))
-
    \frac{\gamma_k m_k}{2} \\
- &\
    \gamma_k \log(C_k) 
- 
    \gamma_k m_k \log (\lambda)
+
     \gamma_k \lambda \left( \sum_{i \in G_k} \Sigma_{ii} + \mu_i^2 \right)^{1/2}
+
    (1-\gamma_k) \log \frac{1 - \gamma_k}{1 - \bar{w}}
\bigg)
\end{aligned}
\end{equation}
and in turn, \red{$\L(Q')(\D) \geq  \E_{Q'} \left[ \ell(\D; \beta) \right] - \varrho ( \mu, \Sigma, \gamma )$}. Via this \red{lower} bound, we are able to construct a tractable surrogate objective. Formally, this objective is denoted by,
$\F(\mu, \Sigma, \gamma, \theta)$, where we introduce $\vartheta$ to parameterize additional hyperparameters, specifically those introduced to model the variance term under the Gaussian family, and those introduced to bound the Binomial likelihood. Under this parameterization, we define,
\red{\begin{equation} \label{eq:objective} 
    \F(\mu, \Sigma, \gamma, \vartheta) := 
    \varLambda(\mu, \Sigma, \gamma, \vartheta)
    - \tilde{\varrho}(\vartheta) 
    - \varrho ( \mu, \Sigma, \gamma)
    \leq \L_{Q'}(\D),
\end{equation}}
where $\tilde{\varrho}(\vartheta) : \varTheta \rightarrow \R $ and \red{$\varLambda(\mu, \Sigma, \gamma, \vartheta) \leq \E_{Q'} \left[ \ell(\D; \beta) \right] $} (which is an inequality under non-tractable likelihoods). To this end, what remains to be computed is an expression for the expected log-likelihood and the term $\tilde{\varrho}(\vartheta)$ where appropriate. In the following three subsections, we provide derivations of the objective function $\F(\mu, \Sigma, \gamma, \theta)$ for the Gaussian, Binomial, and Poisson likelihoods taking the canonical link functions for each.

\subsubsection{Gaussian}

Under the Gaussian family with identity link function, $Y_i \overset{\text{iid}}{\sim} N(x_i^\top \beta, \tau^2)$ for all $i=1,\dots, n$, the log-likelihood is given by $\ell(\D; \beta, \tau^2) = - \frac{n}{2}\log(2\pi\tau^2) - \frac{1}{2\tau^2} \| y - X \beta \|^2$. To model $\tau^2$, an inverse Gamma prior is considered due to its popularity among practitioners, i.e., $\tau^2 \overset{\text{ind}}{\sim} \IG(a, b)$, which has density $\frac{b^a}{\Gamma(a)} x^{-a - 1} \exp\left(\frac{-b}{x}\right)$ where $a, b >0$. The prior is therefore given by $\Pi(\beta, \tau^2) = \Pi(\beta) \Gamma^{-1}(\tau^2)$, and the posterior density by $d\Pi(\beta, \tau^2 | \D) = \Pi_D^{-1} e^{\ell(\D; \beta, \tau^2)} d\Pi(\beta, \tau^2)$, where $\Pi_\D = \int_{\R^{p} \times \R_+} e^{\ell(\D; \beta, \tau^2)} d\Pi(\beta, \tau^2)$ and $\D = \{ (y_i, x_i) \}_{i=1}^n$ is the observed data. 

To include this term within the inference procedure, the variational families $\Q$ and $\Q'$ are extended by $\Q_\tau = \Q \times \{\IG(a', b') : a', b' > 0 \}$ and $\Q'_\tau = \Q' \times \{ \IG(a', b') : a', b' > 0\}$ respectively. Under the extended variational family $\Q'_\tau$, \red{the ELBO is given by 
$$
\E_{Q'_\tau} \bigg[ 
    \ell(\D; \beta, \tau^2) 
    - \log \frac{dQ'}{d\Pi}(\beta) 
    - \log \frac{d \IG(a', b')}{d \IG(a, b)}(\tau^2) 
    \bigg],
$$}
where the final term follows from the independence of $\tau^2$ and $\beta$ in the prior and variational family and is given by,
\begin{equation} \label{eq:gaus_rho}
    \tilde{\varrho}(\vartheta) 
     := 
     \E_{Q'_\tau} \left[ \log \left( \frac{d\IG(a', b')}{ d \IG(a, b)} \right) \right] 
    = 
    (a' - a)\kappa(a') 
    + a \log \frac{b'}{b}
    + \log \frac{\Gamma(a)}{\Gamma(a')}
    + \frac{(b - b')a'}{b'}
\end{equation}
where $\vartheta = \{ (a', b') \}$.
The expectation of the log-likelihood is given by,
\red{\begin{equation} \label{eq:gaus_lambda}
\begin{aligned}
    \varLambda(\mu, \Sigma, & \gamma, \vartheta) :=  
     \E_{Q'_\tau} 
     \left[ 
	 \ell(\D; \beta) 
     \right]  
 =
    - \frac{n}{2}(\log(2\pi) + \log(b') - \kappa(a'))
    \\
    - &\
    \frac{a'}{2b'}
    \Bigg(
	\| y \|^2 
	+
	\left(
	\sum_{i=1}^p \sum_{j=1}^{p} 
	    (X^\top X)_{ij} 
	    \E_{Q'_\tau} \left[ \beta_i \beta_j \right] 
	\right)
	 - 2 \sum_{k=1}^M \gamma_k \langle y, X_{G_k} \mu_{G_k} \rangle
    \Bigg)
\end{aligned}
\end{equation}}
where,
\begin{equation} \label{eq:exp_bibj}
    \E_{Q'_\tau} \left[ \beta_i \beta_j \right] = \begin{cases}
	\gamma_k \left( \Sigma_{ij} + \mu_{i} \mu_{j} \right) & \quad i,j \in G_k \\
	\gamma_k \gamma_h \mu_{i}\mu_{j}       & \quad i \in G_k, j \in G_h, h \neq k
    \end{cases}
\end{equation}
Substituting \eqref{eq:gaus_rho} and \eqref{eq:gaus_lambda} into \eqref{eq:objective} gives $\F(\mu, \Sigma, \gamma, \theta)$ under the Gaussian family.

\subsubsection{Binomial}

Under the Binomial family with logistic link, $Y_i \overset{\text{iid}}{\sim} \text{Bernoulli}(p_i)$ for all $i=1,\dots,n$ where $p_i = \P(Y_i = 1 | x_i) = \exp(x_i^\top \beta) / (1 + \exp(x_i^\top \beta))$.  The log-likelihood is given by
$ 
    \ell(\D, \beta) 
    = 
	\sum_{i=1}^n  y_i \left( x_i^\top \beta \right) 
	- \log \left(1 + \exp(x_i^\top \beta) \right)
$
where $\D = \{(y_i, x_i)\}_{i=1}^n$ with $y_i \in \{0, 1\}$.

Unlike the Gaussian family, variational inference in this setting is challenging because of the intractability of the expected log-likelihood under the variational family. To overcome this issue several authors have proposed bounds or approximations to maintain tractability (see \cite{Depraetere2017a} for a review). Here we employ a bound introduced by \cite{Jakkola1997}, given as
\begin{equation} \label{c4:eq:jj_bound}
    \sigmoid(x) \geq \sigmoid(t) \exp \left\{ \frac{x-t}{2} - \frac{a(t)}{2} (x^2 - t^2) \right\}
\end{equation}
where $\sigmoid(x) = (1 + \exp(-x))^{-1}$ and  $a(t) = \frac{\sigmoid (t) - 1/2}{t}$, $x \in \R$ and $t \in \R$ is an additional parameter that must be optimized to ensure tightness of the bound. 

\red{Using \eqref{c4:eq:jj_bound} allows for the log-likelihood to be bounded by, 
\begin{align}
    \ell(\D; &\ \beta)
    \geq
	\sum_{i=1}^n  
	    y_i x_i^\top \beta 
	    + \log \sigmoid (t_i)
	    - \frac{x_i^\top \beta + t_i}{2}
	    - \frac{a(t_i)}{2} ((x_i^\top \beta)^2 - t_i^2)
    \label{eq:jj_likelihood}
\end{align}
where $t_i\in\R$ is a hyper-parameter for each observation.
Taking the expectation of \eqref{eq:jj_likelihood} with respect to the variational family gives
\begin{equation} \label{eq:binom_lambda}
\begin{aligned}
    \E_{Q'} \left[ \ell (\D; \beta) \right] 
    \geq 
    \varLambda(\mu, \Sigma, \gamma, \vartheta) 
    := &
	\sum_{i=1}^n  
     \left( \sum_{k=1}^M \gamma_k (y_i - 1/2) x_{i, G_k}^\top \mu_{G_k} \right)
	    - \frac{t_i}{2}
	    + \log \sigmoid (t_i) \\
	    &\qquad
         -\frac{a(t_i)}{2} \left( \left( \sum_{j=1}^p \sum_{l=1}^p (x_{ij} x_{il} \E_{Q'} \left[ \beta_j \beta_l \right] \right) - t_i^2 \right)
\end{aligned}
\end{equation}}
where $\vartheta = \{ t_1, \dots, t_n \}$. In turn substituting \eqref{eq:binom_lambda} and $\tilde{\varrho}(\vartheta) = 0$ into \eqref{eq:objective} gives the objective under the Binomial family.

\subsubsection{Poisson}

Finally, under the Poisson family with an exponential link function, $Y_i \overset{\text{iid}}{\sim} \text{Poisson} (\lambda_i)$ for all $i=1,\dots, n$ with $\lambda_i = \exp(x_i^\top \beta) > 0$. The log-likelihood is given by $\ell(\D; \beta) = \sum_{i=1}^n y_i x_i^\top \beta - \exp(x_i^\top \beta) - \log(y!)$, whose \red{expectation is tractable and given by
\begin{equation}
    \E_{Q'} \left[ \ell (\D; \beta) \right] = 
    \varLambda(\mu, \Sigma, \gamma, \vartheta) 
    := 
    \sum_{i=1}^n 
    \left( \sum_{k=1}^M \gamma_k y_i x_{i, G_k}^\top \mu_{G_k} \right)
    -
    M_{Q'}(x_i)
    - 
    \log(y!)
\end{equation}}
where $ M_{Q'}(x_i) = \prod_{k=1}^M M_{Q_k}(x_{i, G_k}) $ is the moment generating function under the variational family, with $M_{Q_k}(x_{i, G_k}) := \gamma_k M_{N_k}(x_{i, G_k}) + (1- \gamma_k)$ being the moment generating function for the $k$th group and $M_{N_k}(x_{i, G_k}) = \exp \left\{ x_{i, G_k}^\top \mu_{G_k} + \frac{1}{2} x_{i, G_k}^\top \Sigma_{G_k} x_{i, G_k} \right\}$. Unlike the previous two families, the Poisson family does not require any additional variational parameters; therefore, $\vartheta = \{ \}$ and $\tilde{\varrho}(\vartheta) = 0$.

\subsection{Coordinate ascent algorithm} \label{sec:algorithm} 

Recall the aim is to approximate the posterior $\Pi(\cdot | \D)$ by a distribution from a given variational family. This approximation is obtained via the maximization of the objective $\F$ derived in the previous section. To achieve this, a CAVI algorithm is proposed as outlined in \Cref{alg:cavi_gsvb_gaussain}. 

In this context, the objective introduced in \eqref{eq:objective} is written as $\F(\mu, \Sigma, \gamma, \vartheta) = \F(\mu_{G_k},\\ \mu_{G_k^c}, \Sigma_{G_k}, \Sigma_{G_k^c}, \gamma_k, \gamma_{-k}, \vartheta)$, highlighting the fact that optimization over the variational parameters occurs group-wise. Further, for each group $k$, while the optimization of the objective function over the inclusion probability, $\gamma_k$, can be done analytically, we use the Limited-memory Broyden–Fletcher–Goldfarb–Shanno optimization algorithm (L-BFGS) to update $\mu_{G_k}$ at each iteration of the CAVI procedure \citep{Nocedal1980}. Details for the optimization with respect to $\Sigma_{G_k}$ are presented in the following subsection. The hyperparameters, $\vartheta$, are updated using L-BFGS for the Gaussian family and analytically for those under the Binomial family. Finally, to assess convergence, the total absolute change in the parameters is tracked, terminating when this quantity falls below a specified threshold, set to $10^{-3}$ in our implementation. Other methods involve monitoring the absolute change in the ELBO; however, we found this prohibitively expensive to compute for this purpose.

\vspace{1em}
\begin{algorithm}[htp]
    \caption{Group sparse coordinate ascent variational inference}
    \label{alg:cavi_gsvb_gaussain}
    \begin{algorithmic}[0]
	\vspace{.3em}
	\State Initialize $ \mu, \Sigma, \gamma, \vartheta $
	\NoDo \While {not converged} \NoDo
	    \NoDo \For {$k \in \text{descending\_order}(\| \mu_{G_k} \|, k = 1,\dots, M) $}
	    \State $ \mu_{G_k} \leftarrow {\argmax}_{\mu_{G_k} \in \R^{m_k}}  
	    \  
        \F(\mu_{G_k}, \mu_{G_k^c}, \Sigma_{G_k}, \Sigma_{G_k^c}, \gamma_k=1, \gamma_{-k}, \vartheta)$
	    \State $ \Sigma_{G_k} \leftarrow {\argmax}_{\Sigma_{G_k} \in \R^{m_k \times m_k}} 
	    \ 
        \F(\mu_{G_k}, \mu_{G_k^c}, \Sigma_{G_k}, \Sigma_{G_k^c}, \gamma_k=1, \gamma_{-k}, \vartheta)$
	    \State $ \gamma_k \leftarrow {\argmax}_{\gamma_k \in [0, 1]} 
	    \ 
        \F(\mu_{G_k}, \mu_{G_k^c}, \Sigma_{G_k}, \Sigma_{G_k^c}, \gamma_k, \gamma_{-k}, \vartheta)$
	    \EndFor
	\State $ \vartheta \  \leftarrow {\argmax}_{\vartheta \in \varTheta}\ 
        \F(\mu_{G_k}, \mu_{G_k^c}, \Sigma_{G_k}, \Sigma_{G_k^c}, \gamma_k, \gamma_{-k}, \vartheta)$
	\EndWhile
	\State \Return $ \mu, \Sigma, \gamma, \vartheta$.
    \end{algorithmic}
\end{algorithm}

\subsubsection[Reparameterization of variance]{Reparameterization of $\Sigma_{G_k}$} \label{sec:reparam_trick} 

\red{
Our focus now turns to the optimization of $\F(\mu_{G_k}, \mu_{G_k^c}, \Sigma_{G_k}, \Sigma_{G_k^c}, \gamma_k=1, \gamma_{-k}, \vartheta)$ with respect to $\Sigma_{G_k}$. By using similar ideas to those of \cite{Seeger1999} and \cite{Opper2009}, it can be shown that only one free parameter is needed to describe the optimum of $\Sigma_{G_k}$ under the Gaussian and Binomial families. For the Poisson family, however, this is not the case, and so $\Sigma_{G_k}$ is parameterized by $U_k^\top U_k$, where $U_k \in \R^{m_k \times m_k}$ is an upper triangular matrix. Optimization is then performed on the upper triangular elements of $U_k$. 

Formally, the objective under the Gaussian family with respect to $\Sigma_{G_k}$ is given by
\begin{equation} \label{eq:obj_sig_gaus}
- \frac{a'}{2b'} \tr (X^\top_{G_k} X_{G_k} \Sigma_{G_k}) + \frac{1}{2} \log \det \Sigma_{G_k} - \lambda (\sum_{i \in G_k} \Sigma_{ii} + \mu^2_i)^{1/2} + C
\end{equation}
where $C$ is a constant that does not depend on $\Sigma_{G_k}$. Differentiating \eqref{eq:obj_sig_gaus} with respect to $\Sigma_{G_k}$, setting to zero, and rearranging gives
\begin{equation} \label{eq:sigma_free}
    \Sigma_{G_k} = 
    \left( 
	\frac{a'}{b'} X_{G_k}^\top X_{G_k} + 2 \nu_k I_{m_k}
    \right)^{-1}
\end{equation}
where $\nu_k = \frac{1}{2} \lambda (\sum_{i \in G_k} \Sigma_{ii} + \mu_{i}^2)^{-1/2}$. Thus, substituting $\Sigma_{G_k} = \left( \frac{a'}{b'} X^\top_{G_k} X_{G_k} + w_kI_{m_k} \right)^{-1}$, where $w_k \in \R$, into the objective and optimizing over $w_k$ is equivalent to optimizing over $\Sigma_{G_k}$. A similar reparameterization can be performed for the Binomial family by following the same procedure as above, and is given by
\begin{equation}
    \Sigma_{G_k} = (X^\top_{G_k} A_t X_{G_k} + w_k I)^{-1}
\end{equation}
where $w_k \in \R$ is the free parameter to be optimized and $A_t = \diag(a(t_1), \dots, a(t_n))$. Notably, this result follows due to the quadratic nature of the bound employed. 

These reparameterizations carry the benefit of requiring one free parameter to optimize $\Sigma_{G_k}$ rather than $m_k(m_k - 1)/2$. However, under this reparameterization, the inversion of an $m_k \times m_k$ matrix is required, which can be a time-consuming operation for large $m_k$. 
}

\subsubsection{Initialization and parameter update ordering}
\red{
As with any gradient-descent-based approach, our CAVI algorithm is sensitive to the choice of initial values. As such, we suggest initializing $\mu$ using the group LASSO with a small regularization parameter, as this ensures many of the elements are non-zero. A comprehensive evaluation of various initialization strategies is presented in Section \ref{c5:sec:initialization} of the Supplementary materials.

Regarding the covariance matrix $\Sigma_{G_k}$, this can be initialized by using the reparameterization outlined in \Cref{sec:reparam_trick} with an initial value of $w_k = 1$ for $k = 1, \dots, M$ for both the Gaussian and the Binomial families. For the Poisson family, we propose the use of an initial covariance matrix $\Sigma_{G_k} = \diag(0.2,\dots, 0.2)$. The inclusion probabilities are initialized as $\gamma = (0.5, \dots, 0.5)^\top$. For the additional hyperparameters $\vartheta$, we use: $a' = b' = 10^{-3}$ and 
$ t_i = \left( 
	\sum_{k=1}^M \gamma_k \left[
	    \langle \mu_{G_k},  x_{i, G_k} \rangle^2 + x_{i, G_k}^\top \Sigma_{G_k} x_{i, G_k}
	\right]
    \right)^{1/2}
$ for all $i=1,\dots, n$. 

We suggest updating the parameters in descending order of $\| \mu_{G_k} \|$, as in \cite{Ray2020}, this scheme yielded models with better performance in practice in contrast to the other ordering schemes. A numerical study where other ordering schemes is presented in Section \ref{c5:sec:initialization} of the Supplementary materials.
}
\section{Theoretical results for grouped linear regression}

\subsection{Notation and assumptions}

This section establishes frequentist theoretical guarantees for the proposed VB approach in sparse high-dimensional linear regression with group structure. The full proofs of the following results are provided in the Supplementary Material. 

  To simplify technicalities, the variance parameter $\tau^2$ is taken as known and equal to 1, giving model $Y = X\beta + \eps$ with $Y \in \R^n$, $X \in \R^{n \times p}$ and $\eps \sim N_n(0, I_n)$. Under suitable conditions, contraction rates for the variational posterior are established, which quantify its spread around the `ground truth' parameter $\beta_0 \in \R^p$ generating the data as $n,p\to \infty$. 

Recall that the covariates are split into $M$ pre-specified disjoint groups $G_1,\dots,G_M$ of size $|G_k|=m_k$ with $\sum_{k=1}^M m_k = p$ and maximal group size $\mm = \max_{k=1,\dots,M} m_k$. The above model can then be written as
\begin{equation}\label{eq:model_lin}
Y = \sum_{k=1}^M X_{G_k} \beta_{G_k} + \eps,
\end{equation}
with $\beta_{G_k} \in \R^{m_k}$ and $X_{G_k} \in \R^{n \times m_k}$. Let $P_\beta$ denote the law of $Y$ under \eqref{eq:model_lin}, $S_\beta \subseteq \{G_1,\dots,G_M\}$ be the set containing the indices of the non-zero \textit{groups} in $\beta \in \R^p$. For a vector $\beta \in \R^p$ and set $S$, we also write $\beta_S = (\beta_i)_{i\in G_k: G_k \in S}\in \R^{\sum_{G_k \in S}m_k}$ for its vector restriction to $S$. We write $\beta_0$ for the ground truth generating the data, $S_0 = S_{\beta_0}$ and $s_0 = |S_0|$ for its \textit{group-sparsity}. For a matrix $A \in \R^{m\times n}$, let $\|A\|_F^2 = \sum_{i=1}^m \sum_{j=1}^n A_{ij}^2 = \Tr(A^T A)$ be the Frobenius norm and define the group matrix norm of $X \in \R^{n\times p}$ by
$$\|X\| = \max_{k=1,\dots,M} \|X_{G_k}\|_F.$$
If all the groups are singletons, $\|X\|$ reduces to the same norm as in \cite{CSV15}. We further define the $\ell_{2,1}$-norm of a vector by $\|\beta\|_{2,1} = \sum_{k=1}^M \|\beta_{G_k}\|_2$. We assume that the prior slab scale $\lambda$ satisfies
\begin{equation}\label{eq:lambda2}
\underline{\lambda} \leq \lambda \leq 2\bar{\lambda}, \qquad \qquad \underline{\lambda}=\frac{\|X\|}{M^{1/\mm}}, \qquad \qquad \bar{\lambda}= 3 \|X\| \sqrt{\log M},
\end{equation}
mirroring the situation without grouping \citep{CSV15,RS22}.

The parameter $\beta$ in \eqref{eq:model_lin} is not estimable without additional assumptions on the design matrix $X$, for instance that $X^T X$ is invertible  for sparse subspaces of $\R^p$. These notions of invertibility can be precise via the following definitions, which are the natural adaptations of compatibility conditions to the group sparse setting \citep{BvdG2011,CSV15}.
\begin{definition}
A model $S \subseteq \{1,\dots,M\}$ has \textit{compatibility number}
$$\phi(S) = \inf \left\{ \frac{\|X\beta\|_2 |S|^{1/2}}{\|X\| \|\beta_S\|_{2,1}} : \|\beta_{S^c}\|_{2,1} \leq 7\|\beta_S\|_{2,1}, \beta_S \neq 0 \right\}.$$
\end{definition}
Compatibility considers only vectors whose coordinates are small outside $S$, and hence is a (weaker) notion of approximate rather than exact sparsity. For all $\beta$ in the above set, it holds that $\|X\beta\|_2 |S|^{1/2} \geq \phi(S) \|X\| \|\beta_S\|_{2,1}$, which can be interpreted as a form of continuous invertibility of $X$ for approximately sparse vectors in the sense that changes in $\beta_S$ lead to sufficiently large changes in $X\beta$ that can be detected by the data.  The number 7 is not important and is taken in Definition 2.1 of \cite{CSV15} to provide a specific numerical value; since we use several of their techniques, we employ the same convention. We next consider two further notions of invertibility for exact group sparsity.
\begin{definition}
\label{def:unif_compat}
The compatibility number for vectors of dimension $s$ is
$$\bar{\phi}(s) = \inf\left\{ \frac{\|X\beta\|_2\|S_\beta|^{1/2}}{\|X\|\|\beta\|_{2,1}} : 0\neq |S_\beta| \leq s \right\}.$$
\end{definition}
\begin{definition}
The smallest scaled sparse singular value of dimension $s$ is
$$\widetilde\phi (s) = \inf \left\{ \frac{\|X\beta\|_2}{\|X\| \|\beta\|_2} : 0\neq |S_\beta|\leq s \right\}.$$
\end{definition}
While $X^TX$ is not generally invertible in the high-dimensional setting, these last two definitions weaken this requirement to sparse vectors. These are natural extensions of the definitions in \cite{CSV15} to the group setting, and similar interpretations and relations to the usual sparse setting apply also here, see \cite{BvdG2011} or Section 2.2 in \cite{CSV15} for further discussion. 

The interplay of the group structure and individual sparsity can lead to multiple regimes see for instance \cite{Lounici2011} and \cite{BvdG2011}. To make our results more interpretable, we restrict to the main case of practical interest where the group sizes are not too large, and hence the group sparsity drives the estimation rate. 
\begin{customassumption}{(K)}\label{ass:theory}
There exists $K>0$ such that $\mm \log \mm \leq K \log M.$
\end{customassumption}
While related works make similar assumptions \citep{Bai2020}, introducing an explicit constant $K>0$ above allows us to clarify the uniformity in our results.

\subsection{Asymptotic results}

We now state our main result on variational posterior contraction for both \textit{prediction loss} $\|X(\beta-\beta_0)\|_2$ and the usual $\ell_2$-loss. Our results are uniform over vectors in sets of the form
\begin{align}\label{eq:Bn}
\mathcal{B}_{\rho_n,s_n} = \mathcal{B}_{\rho_n,s_n}(c_0) :=\{\beta_0 \in \R^p:\phi(S_{\beta_0}) \geq c_0, ~ |S_{\beta_0}|\leq s_n, ~ \widetilde{\phi}(\rho_n |S_{\beta_0}|) \geq c_0 \},
\end{align}
where $s_n \geq 1$, $c>0$ and $\rho_n \to\infty$ (arbitrarily slowly). \red{The sequence $\rho_n$ plays the role of an `arbitrarily large' constant in the smallest scaled sparse singular value above.}

\begin{theorem}[Contraction]\label{thm:contraction}
Suppose that Assumption \ref{ass:theory} holds, the prior \eqref{eq:prior} satisfies \eqref{eq:lambda2} and $s_n$ satisfies $\mm \log s_n = O(\log M)$. Then the variational posterior $\tilde{\Pi}$ based on either the variational family $\mathcal{Q}$ in \eqref{eq:family_1} or $\mathcal{Q}'$ in \eqref{eq:family_2} satisfies, with $s_0 = |S_{\beta_0}|$,
$$\sup_{ \beta_0\in \mathcal{B}_{\rho_n,s_n} } E_{\beta_0} \tilde{\Pi} \left( \beta: \|X(\beta-\beta_0)\|_2 \geq  \frac{H_0 \rho_n^{1/2} \sqrt{s_0\log M}}{\bar{\phi}(\rho_n s_0)} \right) \to 0$$
$$\sup_{ \beta_0\in \mathcal{B}_{\rho_n,s_n} } E_{\beta_0} \widetilde{\Pi} \left( \beta: \|\beta-\beta_0\|_2 \geq  \frac{H_0\rho_n^{1/2} \sqrt{s_0\log M}}{ \|X\| \widetilde{\phi}(\rho_n s_0)^2}  \right) \to 0 , $$
for any $\rho_n \to \infty$ (arbitrarily slowly), $\mathcal{B}_{\rho_n,s_n}$ defined in \eqref{eq:Bn} and where $H_0$ depends only on the prior.
\end{theorem}

Contraction rates for the full posterior based on a group spike and slab prior were established in \cite{NJG20}, as well as for the spike and slab group LASSO in \cite{Bai2020}. Our proofs are instead based on those in \cite{CSV15} and \cite{RS22}. We have extended their theoretical results from the coordinate sparse setting to the group sparse setting. This approach permits more explicit proofs and allows us to consider somewhat different assumptions from previous group sparse works \citep{NJG20,Bai2020}, for instance removing the boundedness assumption for the parameter spaces \red{or certain sample size conditions. In particular, we can cover some settings far from i.i.d. responses, in which case $\|X\|^2 \widetilde{\phi}(\rho_n s_0)^4$ can be thought of as the effective sample size of the problem (corresponding to roughly $m_{\max}n$ in compatible i.i.d. models as in the above references).}

\red{\begin{remark}[Unknown variance]\label{rem:noise_variance}
If the noise variance $\tau^2$ is unknown, standard Bayesian approaches either replace the unknown parameter by an estimate or assign it a hyperprior, such as an inverse Gamma prior as we do in the methodology section above. The hierarchical approach has been theoretically studied in the present group sparse setting with posterior consistency established for $\tau^2$, see \cite{NJG20,Bai2020}. However, such proofs are based on abstract testing conditions and typically require additional assumptions, such as uniform boundedness of the parameter $\beta$, ruling out very strong signals. Since such a proof approach gives the exponential tail decay required to transfer results from the true posterior to its variational counterpart (see Lemma \ref{lem:post_to_VB} below), it should in principle be possible to extend our VB results to the hierarchical setting at the expense of such additional assumptions.
\end{remark}}

\begin{remark}
The optimization problem \eqref{eq:optim} is in general non-convex, and hence there is no guarantee that CAVI (or any other algorithm) will converge to the global minimizer $\tilde{\Pi}$. However, an inspection of the proofs shows that the conclusions of Theorems \ref{thm:contraction} and \ref{thm:dimension} apply also to any element $Q^*\in \mathcal{Q} \subset \mathcal{Q}'$ in the variational families for which
$$0 \leq \KL(Q^*\|\Pi(\cdot|Y) )- \KL (\tilde{\Pi}\|\Pi(\cdot|Y)) = \mathcal{L}_{\tilde{\Pi}}(\mathcal{D}) - \mathcal{L}_{Q^*}(\mathcal{D})= O(s_0 \log M),$$ 
where $s_0$ is the true group sparsity and $\mathcal{L}_Q(\mathcal{D})$ is the ELBO. Thus, as long as the ELBO is within $O(s_0\log M)$ of its maximum, the resulting variational approximation will satisfy the above conclusions, even if it is not the global optimum. \red{While this is not a condition our algorithm can check in practice, it nonetheless suggests some heuristic robustness of the algorithmic output to not terminating at the global optimum.}
\end{remark}

The next result shows that the variational posterior puts most of its mass on models of size at most a multiple of the true number of groups, meaning that it concentrates on sparse sets.

\begin{theorem}[Dimension]\label{thm:dimension}
Suppose that Assumption \ref{ass:theory} holds, the prior \eqref{eq:prior} satisfies \eqref{eq:lambda2} and $s_n$ satisfies $\mm \log s_n =O(\log M)$. Then the variational posterior $\tilde{\Pi}$ based on either the variational family $\mathcal{Q}$ in \eqref{eq:family_1} or $\mathcal{Q}'$ in \eqref{eq:family_2} satisfies
$$\sup_{ \beta_0\in \mathcal{B}_{\rho_n,s_n} } E_{\beta_0} \tilde{\Pi} \left( \beta: |S_\beta| \geq \rho_n |S_{\beta_0}|  \right) \to 0$$
for any $\rho_n \to \infty$ (arbitrarily slowly) and $\mathcal{B}_{\rho_n,s_n}$ defined in \eqref{eq:Bn}.
\end{theorem}

\red{Proofs of these results can be found in Section \ref{sec:prooofs}. Our results can be extended to the technically more involved case of group sparse binary logistic regression following similar ideas to \cite{Ray2020}, while they should also extend conceptually to sparse generalized linear models as in \cite{jeong2021} under additional assumptions such as uniform boundedness of the parameter, see Remark \ref{rem:noise_variance}.
}

\section{Numerical experiments} \label{sec:simulations}

In this section, a numerical evaluation of our method, referred to as Group Spike-and-slab Variational Bayes (GSVB), is presented. Where necessary, we distinguish between the two families, $\Q$ and $\Q'$ presented in \Cref{sec:prior}, by the suffixes `--D' and `--B' respectively, i.e., GSVB--D denotes the method under the variational family $\Q$. Notably, throughout all our numerical experiments, the prior parameters are set to $\lambda =1$, $\alpha_0=1$, $b_0=M$, and $a=b=10^{-3}$ for the inverse-Gamma prior on $\tau^2$ under the Gaussian family. Scripts to reproduce our results can be found at \url{https://github.com/mkomod/p3}. Furthermore, an R package implementing our methodology is available at \url{https://github.com/mkomod/gsvb}.

\red{
To evaluate the performance of GSVB, we compare it against MCMC and the Spike-and-Slab Group LASSO (SSGL) proposed by \cite{Bai2020}. 
The details of the MCMC sampler used in this study are outlined in \Cref{appendix:gibbs} and an implementation is available at \url{https://github.com/mkomod/spsl}.  Briefly, we ran the sampler for $500,\!000$ iterations with a burn-in period of $100,\!000$ iterations across four chains to report the results. Implementing MCMC with the spike-and-slab prior is known to be computationally challenging. To assess convergence, we applied the Gelman-Rubin diagnostic statistic \citep{Gelman1992}, specifically the multivariate potential scale reduction factor (PSRF) computed across the model coefficients. While most ($>95\%$) MCMC runs achieved convergence for the Gaussian and Binomial likelihood models (with PSRF values below 1.1), we were unable to achieve convergence across all settings in the Poisson model. 
Despite the lack of full convergence in the Poisson model, the simulation study demonstrates that the resulting posterior distribution effectively identifies the relevant groups and concentrates appropriately around the data-generating coefficients, making this sampler a relevant baseline for comparing our approach.

The Spike-and-Slab Group LASSO (SSGL) proposed by \cite{Bai2020} is a state-of-the-art Bayesian method that employs a similar prior as in \eqref{eq:prior}. However, unlike GSVB, the multivariate Dirac mass on zero is replaced with a multivariate double exponential distribution, giving a continuous mixture with one density acting as the spike and the other as the slab, parameterized by $\lambda_0$ and $\lambda_1$ respectively. Under this prior, \cite{Bai2020} derive an EM algorithm, which allows for fast updates; however, only MAP estimates are returned by default, meaning a posterior distribution for $\beta$ is not available. To mitigate this, debiasing methods can be used to obtain uncertainty quantification for MAP-based methods like SSGL. Notably, \cite{Bai2020} adopt ideas from a recent line of research based on debiasing estimates from high-dimensional regression~\citep{vanDeGeer2014}; these results yield an asymptotic distribution for the model coefficients, which in turn can be used to construct confidence intervals for the estimated coefficients.

Our study begins by first evaluating GSVB against MCMC and the debiased SSGL on smaller scale datasets. We then evaluate the performance against SSGL on large scale datasets as debiasing was too computationally expensive within this setting.}

\subsection{Simulation setup}

Data is simulated for $i=1,\dots,n$ observations, each having a response $y_i \in \R$ and $p$ continuous predictors $x_i \in \R^p$. The response is sampled independently from the respective family with mean given by $g(\beta_0^\top x_i)$, where $g$ is the link function (and variance applicable to the Gaussian family of $\tau^2=1$). The true coefficient vector $\beta_0 = (\beta_{0, G_1}, \dots, \beta_{0, G_M})^\top \in \R^p$ consists of $M$ groups, each of size $m$, i.e., $M \times m = p$. Of these groups, $s$ are chosen at random to be non-zero and have each of their element values sampled independently and uniformly from $[-\beta_{\max}, -0.2] \cup [0.2 , \beta_{\max}]$, where $\beta_{\max} = 1.5, 1.0$, and $0.45$ for the Gaussian, Binomial, and Poisson families, respectively. Finally, the predictors are generated from one of four settings:
\begin{description}
    \item[Setting 1:] $x_i \overset{\text{iid}}{\sim} N(0_p, I_p)$ 
    \item[Setting 2:] $x_i \overset{\text{iid}}{\sim} N(0_p, \Sigma)$, where $\Sigma_{ij} = 0.6^{|i - j|}$ for $i,j=1,\dots,p$.
    \item[Setting 3:] $x_i \overset{\text{iid}}{\sim} N(0, \Sigma)$, where $\Sigma$ is a block diagonal matrix where each block $A$ is a $50\times 50$ square matrix such that $A_{jl} = 0.6, j \neq l$ and $A_{jj} = 1$ otherwise. 
    \item[Setting 4:] $x_i \overset{\text{iid}}{\sim} N(0, \Sigma)$, where 
	$\Sigma = (1-\alpha) W^{-1} + \alpha V^{-1}$ with $W \sim \text{Wishart}(p+\nu, I_p)$ and $V$ is a block diagonal matrix of $M$ blocks, where each block $V_k$, for $k=1,\dots ,M$, is an $m_k \times m_k$ matrix given by $ V_k \sim \text{Wishart}(m_k + \nu, I_{m_k})$; we let $(\nu, \alpha) = (3, 0.9)$ so that predictors within groups are more correlated than variables between groups. 
\end{description}

To evaluate the performance of the methods, four different metrics are considered: 
    (i) the $\ell_2$-error, $\| \widehat{\beta} - \beta_0 \|_2$, between the true vector of coefficients and the estimated coefficient $\widehat{\beta}$ defined as the posterior mean where applicable, or the MAP estimate if this is returned,
    (ii) the area under the curve (AUC) of the receiver operator characteristic curve, which compares true positive and false positive rates for different thresholds of the group posterior inclusion probability, 
    (iii) the marginal coverage of the non-zero coefficients, which reports the proportion of times the true coefficient $\beta_{0, j}$ is contained in the marginal credible set 
    $\{ j : \beta_{0, j} \neq 0 \}$, 
    (iv) and the size of the marginal credible set for the non-zero coefficients, given by the Lebesgue measure of the set.
The last two metrics can only be computed when a distribution for $\beta$ is available, i.e., via MCMC, VB, \red{or debiasing}. For MCMC and GSVB, the $95\%$ 
marginal credible sets for each variable $j \in G_k$ for $k=1,\dots,M$ are given by: 
\begin{equation*}
    S_j = \begin{cases}
    \{0\} & \text{if } \gamma_k < \alpha \\
	\left[ \mu_{j} \pm \Sigma_{jj}^{1/2} \Phi^{-1}( \frac{\alpha_{\gamma_k}}{2})
	    \right]
	&\text{if } \gamma_k \geq \alpha  \text{ and }
	0 \notin 
	[\mu_{j} \pm \Sigma_{jj}^{1/2} \Phi^{-1}( \frac{\alpha_{\gamma_k}}{2})]\\
	\left[ 
	    \mu_{j} \pm \Sigma_{jj}^{1/2}
	    \Phi^{-1}( \frac{\alpha_{\gamma_k}}{2} + \frac{1 - \gamma_k}{2}) 
	\right] \cup \{ 0 \}
 &\text{otherwise}
	\\
    \end{cases}
\end{equation*}
where $\alpha_{\gamma_k} = 1 - \frac{\alpha}{\gamma_k}$ and $\Phi^{-1}$ is the quantile function of the standard Normal distribution.

\subsection{Small-scale simulations and comparison to MCMC} \label{sec:comp_mcmc} 

\red{
In this section, the performance of GSVB is evaluated for the Gaussian, Binomial, and Poisson families. Notably, we compare against MCMC, often considered the gold standard for Bayesian inference, to assess the quality of the variational posterior. In addition, we compare with the SSGL in the Gaussian setting as subsequent post-processing (via debiasing) can be performed to obtain uncertainty quantification.}

Within this comparison, we set $p = 1,\!000$, $m=5$, and vary the number of non-zero groups, $s$. As highlighted in Figure \ref{fig:gaus_mcmc_res}, GSVB-B performs excellently in nearly all settings, demonstrating comparable results to MCMC \red{and SSGL} in terms of $\ell_2$-error and AUC. This indicates that GSVB-B exhibits similar characteristics to \red{these methods}, both in terms of the selected groups and the posterior mean. As anticipated, all the methods exhibit better performance in simpler settings and show a decline in performance as the problem complexity increases.

Regarding coverage, while MCMC shows slightly better performance compared to GSVB, the proposed method still provides credible sets that capture a significant portion of the true non-zero coefficients (particularly GSVB-B). However, the credible sets of the variational posterior are sometimes not large enough to capture the true non-zero coefficients. This observation is further supported by the size of the marginal credible sets, with MCMC producing the largest sets, followed by GSVB-B and GSVB-D. These findings confirm the well-known fact that VB tends to underestimate the posterior variance \citep{Carbonetto2012, Blei2017, Zhang2019, Ray2020}. Interestingly, the set size is larger under $\Q'$, highlighting the fact that the mean field variational family ($\Q$) lacks the necessary flexibility to accurately capture the underlying structure in the data. Furthermore, this result indicates that the full marginal credible quantity improves through the consideration of the interactions within the group.

\red{
Finally, we notice that the coverage of the debiased SSGL is comparable to both MCMC and VB. This is expected as the debiasing procedure is designed to provide uncertainty quantification for the MAP estimate. However, the debiasing procedure and MCMC are computationally expensive. In particular, we found that these methods can be orders of magnitude slower than GSVB (\Cref{tab:runtime_mcmc_ssgl}), in turn, making the methods less practical for large-scale problems. 
}

\begin{figure}[htp]
    \centering
    \makebox[\textwidth][c]{
	\includegraphics[width=1.02\textwidth]{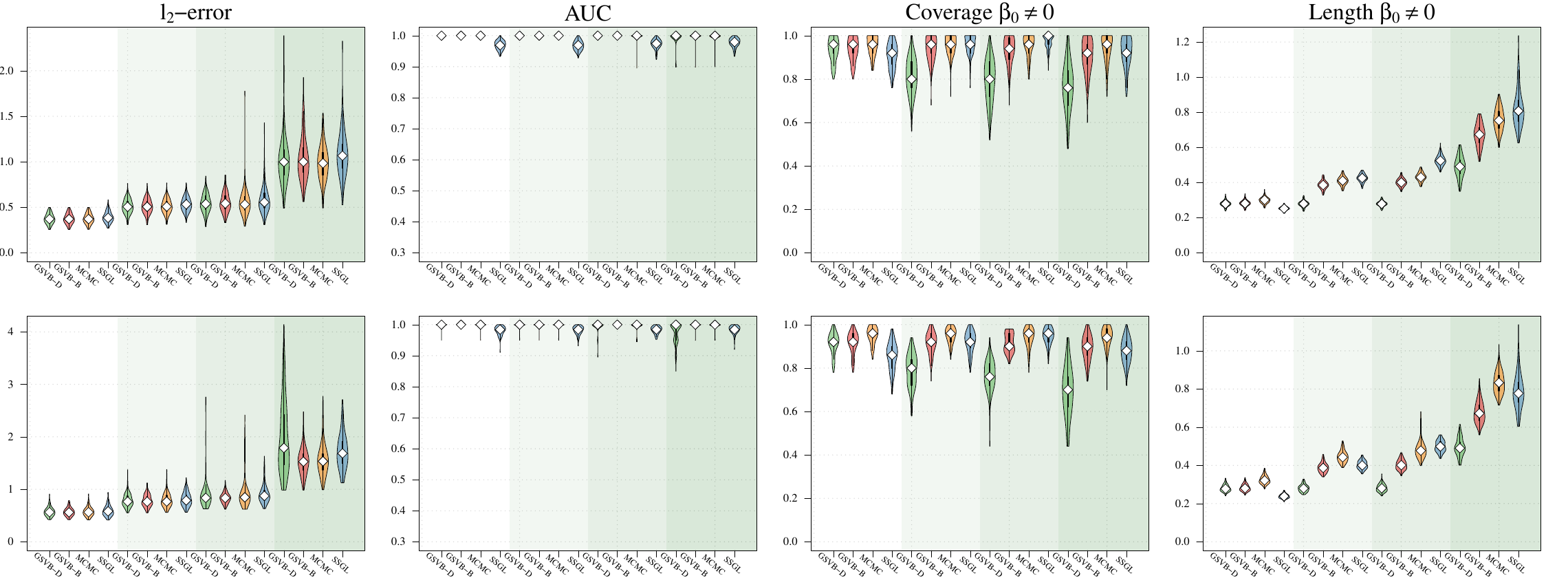}
    } %
    \makebox[\textwidth][c]{
	\includegraphics[width=1.02\textwidth]{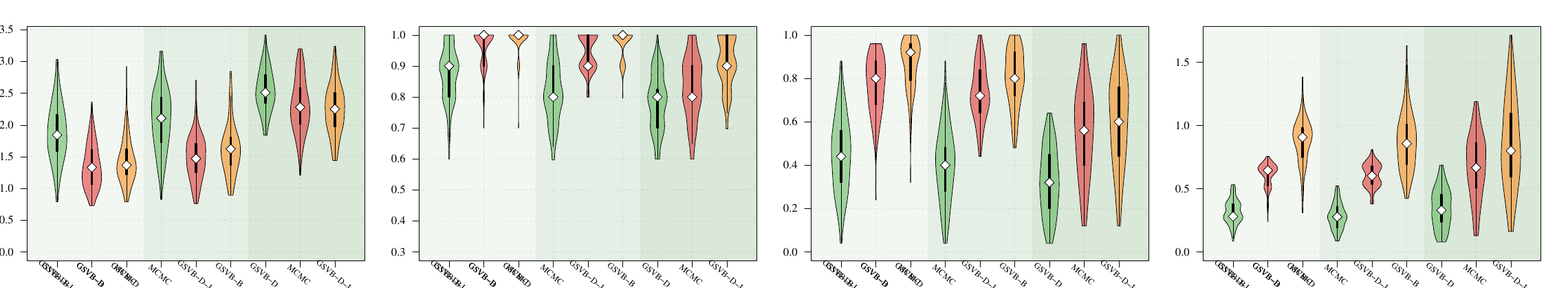}
    } %
    \makebox[\textwidth][c]{
	\includegraphics[width=1.02\textwidth]{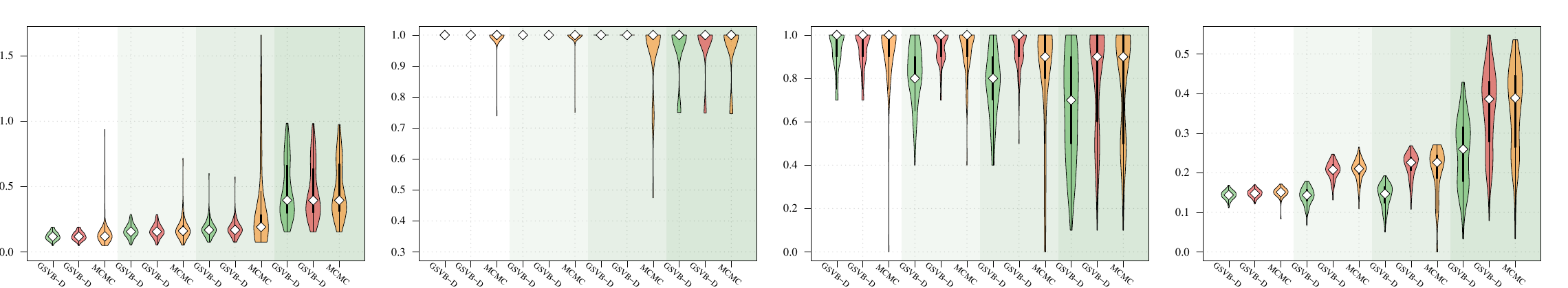}
    } %
    {\footnotesize
	\colorSquareB{setting1!15} Setting 1 \hspace{0.5em}
	\colorSquareB{setting2!15}  Setting 2 \hspace{0.5em}
	\colorSquareB{setting3!15}  Setting 3 \hspace{0.5em}
	\colorSquareB{setting4!15}  Setting 4 \hspace{0.5em}

    \vspace{0.0em}
	\colorSquareB{gsvb-d!50} GSVB-D \hspace{0.5em}
	\colorSquareB{gsvb-b!50} GSVB-B \hspace{0.5em}
	\colorSquareB{mcmc!50} MCMC \hspace{0.5em}
    \colorSquareB{ssgl!50} SSGL
    }
    \caption[Comparison of group sparse methods for small scale numerical experiments]
    {\red{Performance evaluation of GSVB, MCMC, and SSGL (with debiasing for the Gaussian family only) for Settings 1--4 with $p=1\!,000$ across 100 runs. For each method, the white diamond (\whitediamond) indicates the median of the metric, the thick black line (\blacklinethick) the interquartile range, and the black line (\blacklinethin) 1.5 times the interquartile range.  
    \textbf{Rows 1--2}: Gaussian family with $(n, m, s) = (200, 5, 5),
    (200, 5, 10)$. 
    \textbf{Row 3}: Binomial family with $(n, m, s) = (400, 5, 5)$.
    \textbf{Row 4}: Poisson family with  $(n, m, s) = (400, 5, 2)$. 
    Note that for the Binomial family, Setting 1 results in perfect separation of classes and is excluded from the study.
    }}
    \label{fig:gaus_mcmc_res}
\end{figure}

\red{
 \begin{table}[htp]
    \centering
    \makebox[1.0\textwidth][c]{
    \resizebox{0.90\textwidth}{!}{ %
    \renewcommand{\arraystretch}{1.0}
    \begin{tabular}{l l l l l l}
    \toprule
 & Method & Setting 1 & Setting 2 & Setting 3 & Setting 4 \\
    \midrule
\multirow{4}{1.5cm}{Gaus. s=5}
& GSVB--D & 
    5.1s (3.8s, 5.7s)            & 5.1s (3.7s, 5.8s)          & 5.5s (3.9s, 10.2s)          & 5.8s (4.2s, 9.3s) \\
& GSVB--B &
    1.4s (1.1s, 1.8s)            & 1.7s (1.3s, 2.5s)          & 1.8s (1.3s, 6.5s)           & 2.7s (1.7s, 5.8s) \\
& MCMC &
    37m 15s (36m 45s, 38m 45s) & 37m 15s (36m 45s, 39m 15s) & 37m (36m 30s, 37m 30s) & 37m 45s (36m 45s, 39m 15s) \\
& SSGL &
    6m 11s (6m 3s, 6m 16s)       & 6m 33s (6m 30s, 6m 38s)    & 6m 58s (6m 55s, 6m 2s)      & 6m 55s (6m 44s, 6m 9s) \\
\midrule
\multirow{4}{1.5cm}{Gaus. s=10}
& GSVB--D &
    10.2s (8.4s, 11.3s)          & 10.1s (8.2s, 10.9s)        & 16.9s (9.9s, 19.5s)         & 12.0s (8.6s, 18.9s) \\
& GSVB--B &
    2.9s (1.8s, 3.8s)            & 2.8s (2.3s, 3.5s)          & 8.8s (3.3s, 13.8s)          & 3.6s (2.2s, 12.5s) \\
& MCMC &
    51m 45s (50m 15s, 53m 45s) & 47m (46m 15s, 48m 30s) & 46m 45s (45m 45s, 48m) & 46m 45s (46m, 48m) \\
& SSGL &
    6m 8s (6m 6s, 6m 12s)        & 6m 37s (5m 26s, 6m 44s)    & 6m 2s (6m 51s, 6m 7s)       & 6m 57s (6m 41s, 6m 10s) \\
\midrule
\multirow{3}{1.5cm}{Binom. s=10}
& GSVB--D & 
    \multicolumn{1}{c}{--} & 5.0s (2.5s, 8.2s) & 4.5s (2.7s, 7.5s) & 3.7s (2.3s, 5.7s) \\
& GSVB--B & 
    \multicolumn{1}{c}{--} & 5.6s (3.8s, 7.4s) & 5.1s (3.8s, 7.7s) & 5.0s (3.2s, 7.3s) \\
& MCMC &
    \multicolumn{1}{c}{--} & 1h 11m (1h 8m, 1h 12m) & 1h 19m  (1h 15m, 1h 20m) & 1h 9m (1h 5m, 1h 12m) \\
\midrule
\multirow{3}{1.5cm}{Binom. s=10}
& GSVB--D & 
    3.6s (3.3s, 4.1s)          & 3.3s (2.9s, 3.7s)        & 3.1s (2.9s, 5.2s)        & 3.4s (2.4s, 4.1s)        \\
& GSVB--B & 
    12.5s (10.0s, 18.0s)       & 10.1s (7.9s, 14.5s)      & 15.5s (10.1s, 41.6s)     & 13.3s (10.5s, 24.0s) \\
& MCMC & 
    1h 9m (1h 6m, 1h 13m) & 1h 3m (1h 1m, 1h 4m) & 1h 19m (1h 13m, 1h 27m) & 1h 14m (1h 10m, 1h 18m)\\
    \bottomrule
    \end{tabular}
    }
    }
    \caption[Runtime comparison of GSVB, MCMC and the debiased SSGL]{Median (5\%, 95\% quartile) runtimes for numerical experiments presented in \Cref{fig:gaus_mcmc_res}. Here $n=200$ for the Gaussian family and $n=400$ for the Binomial and Poisson family. Note that under setting 1 for the Binomial family there is perfect separation of classes.}
    \label{tab:runtime_mcmc_ssgl}
\end{table}
}

\subsection{Large-scale simulations} \label{sec:large-sims}

Within this section, the performance of GSVB is evaluated on larger datasets. Here we compare against SSGL without debiasing, as the post-processing used to provide uncertainty quantification is not computationally scalable, requiring over 8 hours of compute for a single run, which is beyond our computational budget. However, taking only the MAP for SSGL means the method is scalable with $p$. In this section, $p$ is increased to $p=5,000$. Both the sample size and the number of active (non-zero) groups, $s$, are varied as illustrated in Figure \ref{fig:gaus_comp}.

Regarding the hyperparameters, we set $\lambda_1 = 1$ for SSGL, making the slab identical between the two methods. For the spike of SSGL, a value of $\lambda_0 = 100$ for the Gaussian and Poisson families and $\lambda_0 = 20$ under the Binomial family is used. These values were selected to ensure that sufficient mass is concentrated about zero without turning to cross-validation to select the value. Finally, we let $a_0 = 1$ and $b_0 = M$ for both methods.

Overall, GSVB performs comparatively or better than SSGL in most settings, obtaining a lower $\ell_2$-error and higher AUC (\Cref{fig:gaus_comp}). As expected, across the different methods, there is a decrease in performance as the difficulty of the setting increases, meaning all methods perform better when there is less correlation in the design matrix. We note that the runtime of SSGL is marginally faster than our method in Settings 1-3 and faster in Setting 4. This is explained by the fact that SSGL by default only provides a point estimate for $\beta$. We found that using debiasing to obtain uncertainty quantification was too computationally expensive within this setting, exceeding a runtime allowance of 8 hours.

Within this large-scale simulation, our method provides competitive uncertainty quantification. In particular, GSVB--B provides better coverage of the non-zero coefficients than GSVB--D, which can be justified by the set size. As in our comparison to MCMC, we notice that there is an increase in the posterior set size as the difficulty of the setting increases, i.e., when there is an increase in the correlation of the design matrix. Finally, the runtimes for each method are presented in \Cref{tab:runtime_comp}; generally, the runtime for GSVB is comparable or lower than SSGL. However, we notice that as the difficulty of the setting increases, the runtime for GSVB increases and can exceed the runtime of SSGL.

\vspace{1em}
\begin{figure}[htp]
    \centering
    \makebox[\textwidth][c]{
	\includegraphics[width=0.52\textwidth]{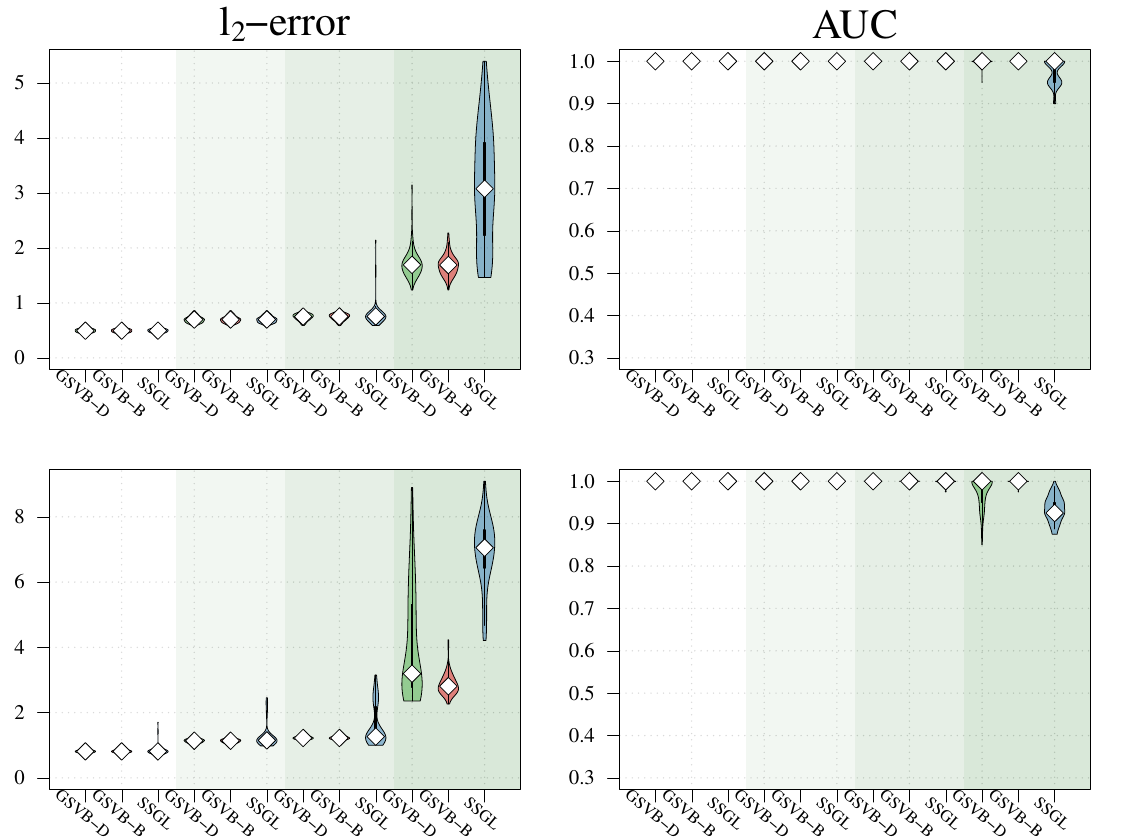}
	\includegraphics[width=0.52\textwidth]{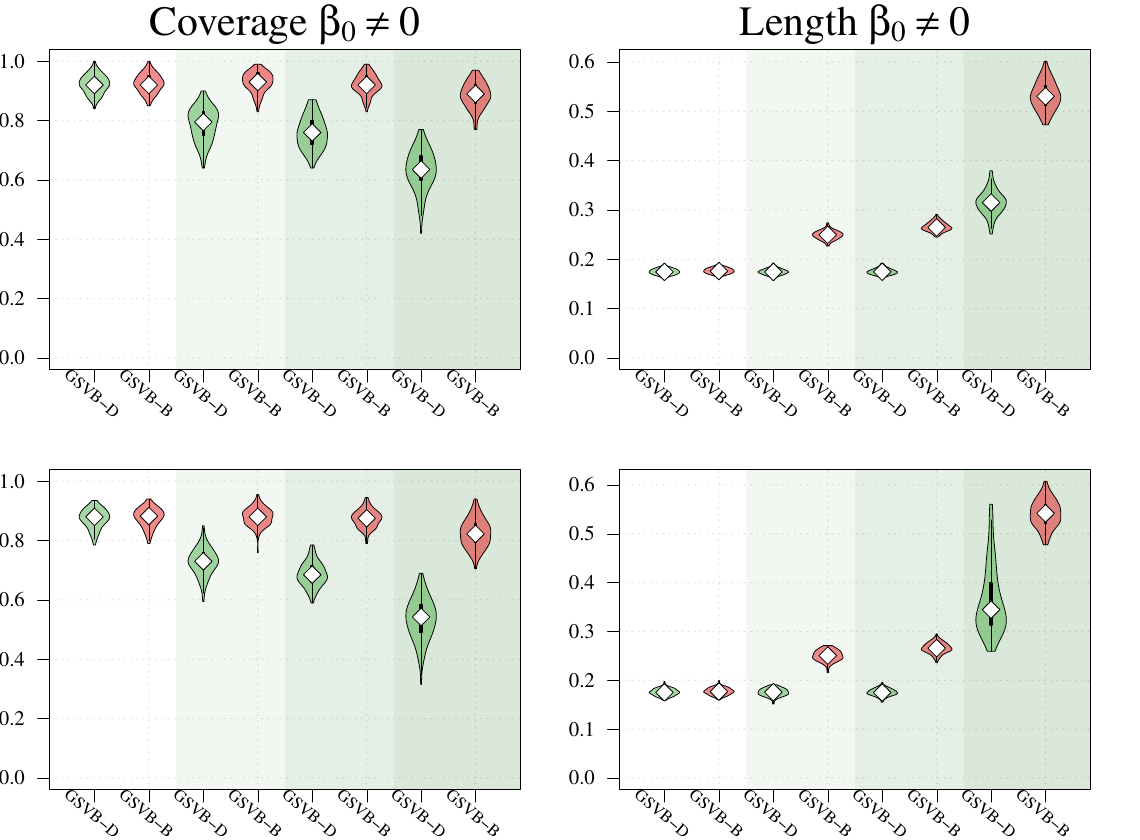}
    } %
    \makebox[\textwidth][c]{
	\includegraphics[width=0.53\textwidth]{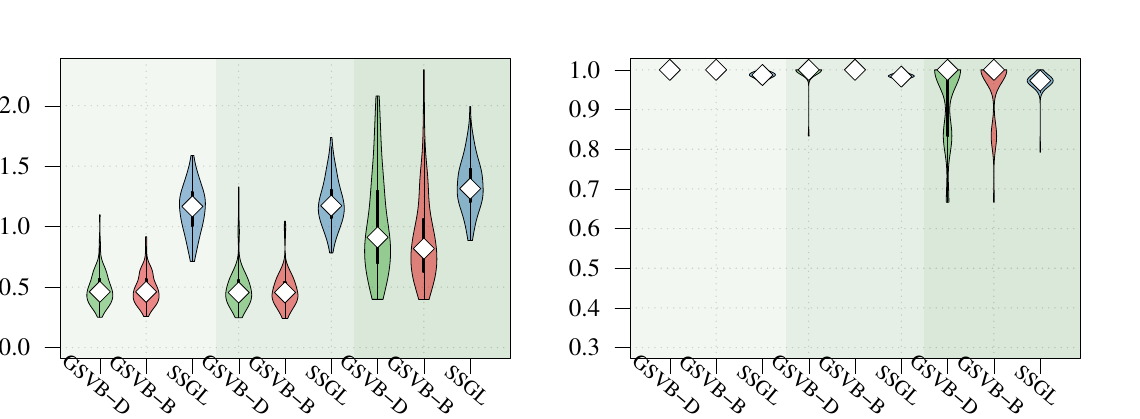}
	\includegraphics[width=0.53\textwidth]{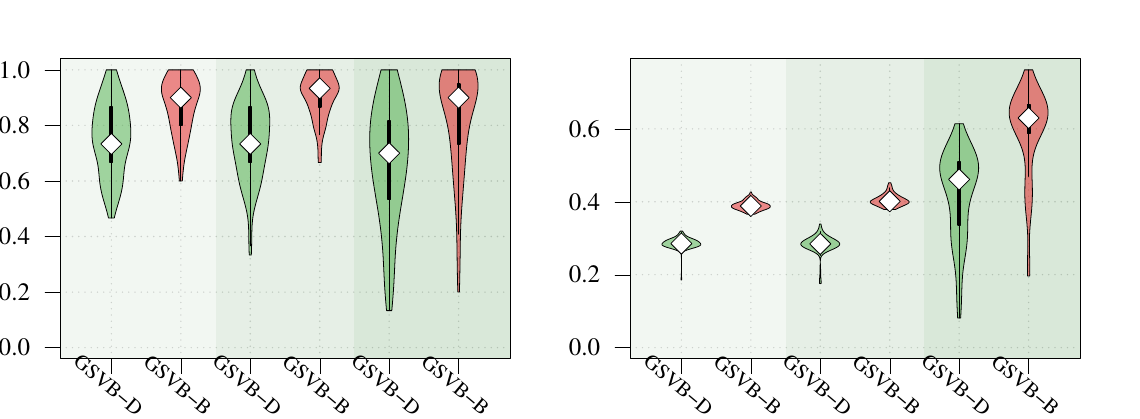}
    }
    \makebox[\textwidth][c]{
	\includegraphics[width=0.53\textwidth]{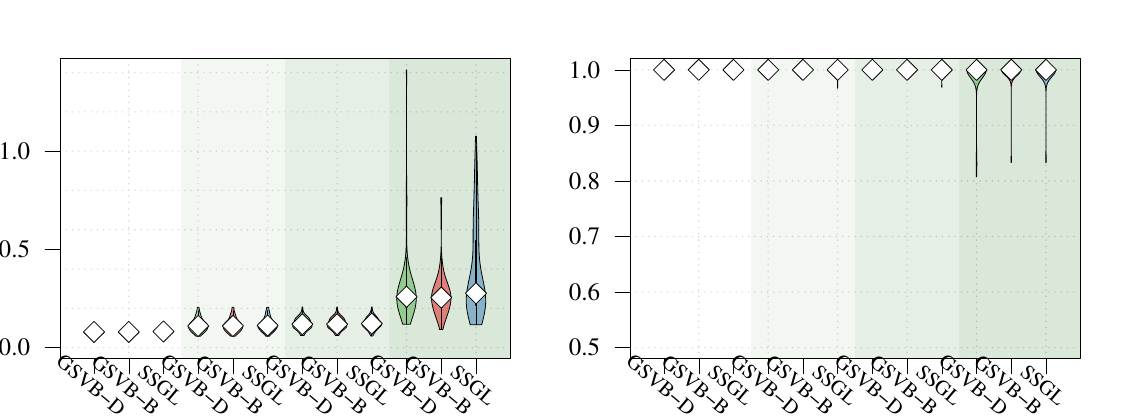}
	\includegraphics[width=0.53\textwidth]{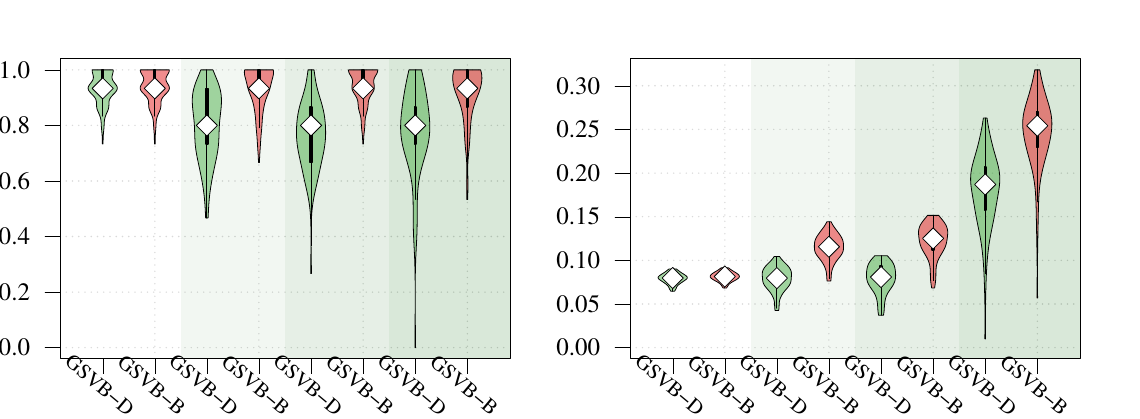}
    } %
    {\footnotesize
	\colorSquareB{setting1!15} Setting 1 \hspace{0.5em}
	\colorSquareB{setting2!15}  Setting 2 \hspace{0.5em}
	\colorSquareB{setting3!15}  Setting 3 \hspace{0.5em}
	\colorSquareB{setting4!15}  Setting 4 \hspace{0.5em}

	\colorSquareB{gsvb-d!50} GSVB-D \hspace{0.5em}
	\colorSquareB{gsvb-b!50} GSVB-B \hspace{0.5em}
	\colorSquareB{ssgl!50} SSGL \hspace{0.05em}
    }
    \caption
    [Comparison of group sparse methods for large scale numerical experiments]
    {\red{Performance evaluation of GSVB and SSGL for Settings 1--4 with $p=5\!,000$ across 100 runs. 
    For each method, the white diamond (\whitediamond) indicates the median of the metric, the thick black line (\blacklinethick) the interquartile range, and the black line (\blacklinethin) 1.5 times the interquartile range. \textbf{Rows 1 -- 2}: Gaussian family with $(n, m, s) = (500, 10, 10), 
    (500, 10, 20)$. 
    \textbf{Row 3}: Binomial family with $(n, m, s) = (1000, 5, 5)$. 
    \textbf{Row 4}: Poisson family with  $(n, m, s) = (1000, 5, 3)$. Note that for the Binomial family, Setting 1 results in perfect separation of classes and is excluded from the study.
    }}
    \label{fig:gaus_comp}
\end{figure}

\section{Application to real-world data} \label{c5:sec:real_data}

This section presents the analysis of three real-world datasets\footnote{
R scripts used to produce these results are available at \url{https://github.com/mkomod/p3}
}, the first two being linear regression problems wherein $p \gg n$ for which GSVB and SSGL are applied, and the third a logistic regression problem where $p < n$, for which the logistic group LASSO is applied in addition to GSVB and SSGL. As before, $a_0 = 1$, $b_0 = M$, and GSVB was run with a value of $\lambda=1$. For SSGL, we let $\lambda_1 = 1$, and $\lambda_0$ was chosen via ten-fold cross-validation on the training set. 

\red{
Throughout, we evaluated the predictive performance on held-out test sets alongside the number of features selected by each method. Beyond this, given that GSVB provides a distribution over $\beta$, we evaluated the posterior predictive (PP) coverage. Notably, the PP distribution is given by,
$$
\tilde{p} (y^\ast | x^\ast, \D) = \int_{\R^p \times \R^+} p(y^\ast | \beta, \tau^2, x^\ast, \D) \ d\tilde{\Pi}(\beta, \tau^2)
$$
where $x^\ast \in \R^p$ is a feature vector. For the Gaussian family, this expression can be simplified by integrating out the variance $\tau^2$. Following \cite{Murphy2007}, substituting $\xi = \tau^2$ and recalling the independence of $\tau^2$ and $\beta$ in our variational family, we have,
\begin{align}
    \tilde{p} (y^\ast | x^\ast, \D) 
    &\ \propto
    \int_{\R^p} \int_{\R^+} \xi^{-(a' + 1/2) -1} \exp \left( - \frac{(y^\ast - x^{\ast \top} \beta)^2 + 2b'}{2 \xi} \right)
    \ d\xi \ d\tilde{\Pi}(\beta)
    \nonumber
    \\
    &\ \propto
    \int_{\R^p} \left(\frac{(y^\ast - x^{\ast \top} \beta)^2}{2b'} + 1 \right)^{-(a' + 1/2)} \ d\tilde{\Pi}(\beta)
    \nonumber
\end{align}
Recognizing that the expression within the integral has the same functional form as a generalized $t$-distribution, whose density is denoted as $t(x; \mu, \sigma^2, \nu) = \Gamma((\nu + 1)/2) / (\Gamma(\nu / 2)\\ \sqrt{\nu \pi} \sigma) (1+ (x-\mu)^2/(\nu \sigma^2))^{-(\nu+1)/2}$, yields,
\begin{equation} \label{eq:post_pred} 
    \tilde{p} (y^\ast | x^\ast, \D) 
    = \int_{\R^p} t(y^\ast; x^{\ast \top} \beta, b'/a', 2a') \ d\tilde{\Pi}(\beta)
\end{equation}
As \eqref{eq:post_pred} is intractable, we instead sample from $p(y^\ast | x^\ast, \D)$ by sampling $\beta \sim \tilde{\Pi}$, and then sampling $y^\ast$ from $Y^\ast = \mu + \sigma t_\nu$ where $\mu = x^{\ast \top} \beta, \sigma = \sqrt{b'/a'}$, and $\nu = 2a'$, where $t_\nu$ denotes a t-distribution with $\nu$ degrees of freedom. To sample from $\tilde{\Pi}$, we sample $z_k \overset{\text{ind}}{\sim} \text{Bernoulli}(\gamma_k)$ and then $\beta_{G_k} \overset{\text{ind}}{\sim} N(\mu_{G_k}, \Sigma_{G_k})$ if $z_k=1$; otherwise, we set $\beta_{G_k} = 0_{m_k}$.

Although not considered for the logistic regression dataset, a similar procedure is available for the Binomial and Poisson models. This is given by sampling $\beta \sim \tilde{\Pi}$ and then sampling $y^* | x^*, \D$ from the respective distribution with mean given by $g({x^*}^\top \beta)$ where $x^* \in \R^p$ and $g$ is the link function. 
}


\newcommand{\mousesec}{Genetic determinants of low-density lipoprotein in mice}
\newcommand{\mousesecshort}{Genetic determinants of LDL-C in mice}

\subsection{\mousesec} \label{sec:mice}

The first dataset we analyze is from the Mouse Genome Database \citep{Blake2021} and consists of $p=10,\!346$ single nucleotide polymorphisms (SNPs)
collected from $n=1,\!637$ laboratory mice. Notably, each SNP, $x_{ij}$ for $j=1,\dots, p$ and $i=1,\dots,n$, takes a value of 0, 1, or 2. This value indicates how many copies of the risk allele are present. Alongside the genotype data, a phenotype (response) is also collected. The phenotype we consider is low-density lipoprotein cholesterol (LDL-C), a continuous value, which has been shown to be a major risk factor for conditions like coronary artery disease, heart attacks, and strokes \citep{Silverman2016a}.

To pre-process the original dataset, SNPs with a rare allele frequency, given by $\text{RAF}_j= \frac{1}{2n} \sum_{i=1}^n \left( \mathbb{I}(x_{ij} = 1) + 2 \mathbb{I}(x_{ij} = 2) \right)$ for $ j = 1, \dots, p$, below the first quartile were discarded. Further, due to the high collinearity, covariates with $| \corr(x_j, x_k) | \geq 0.97$ for $k > j$, $j=1, \dots, p-1$, were removed. After pre-processing, $3,\!341$ SNPs remained. These were used to construct groups by coding each $x_{ij} : \{0, 1, 2\} \mapsto \{ (0, 0), (0, 1), (1, 1)\}$, in turn giving groups of size $m=2$.
.

To evaluate the methods, ten-fold cross-validation is used. Specifically, methods are fit to the validation set, and the test set is used for evaluation. This is done by computing the: (i) MSE between the true and predicted value of the response, given by $\frac{1}{\tilde{n}}\| y - X \widehat{\beta} \|$ where $\tilde{n}$ is the size of the test set, (ii) posterior predictive coverage, which measures the proportion of times the response is included in the 95\% PP set, and (iii) the average size of the 95\% PP set. In addition,
(iv) the number of groups selected, given by $\sum_{k=1}^M \mathbb{I}(\gamma_k > 0.5)$, is also reported.

The performance of each method is presented in \Cref{tab:mgd_results}. The results highlight that GSVB obtains parsimonious models with a comparable MSE to SSGL. In addition, GSVB provides uncertainty quantification with a posterior predictive coverage of $0.945$ and $0.941$ for GSVB--D and GSVB--B, respectively. Finally, when the methods are fit to the full dataset, the SNPs selected by both methods are:
\texttt{rs13476241}, \texttt{rs13477968},
and
\texttt{rs13483814} for GSVB--D, with GSVB--B selecting
\texttt{rs13477939} 
in addition. Notably, \texttt{rs13477968} corresponds to the \textit{Ago3} gene, which has been shown to play a role in activating LDL receptors \citep{Madrigano2008}.

Finally, in genetic studies, $x^{\ast \top} \beta$ is referred to as the polygenic risk score and is commonly used to evaluate genetic risk. While frequentist or MAP methods offer only point estimates, \Cref{fig:mouse-pgrs} demonstrates that GSVB can offer a distribution for this score, highlighting the uncertainty associated with the estimate.

\vspace{1em}
\begin{table}[htp]
    \centering
    \makebox[\textwidth][c]{
    \resizebox{0.7\textwidth}{!}{ %
    \begin{tabular}{l c c c c }
	\toprule
    Method & MSE & Num. selected groups & PP. coverage & PP. length \\
    \midrule
    GSVB-D & 0.012 (0.002) & 3.50 (0.527) & 0.945 (0.027) & 0.425 (0.005) \\
    GSVB-B & 0.012 (0.002) & 4.00 (0.667) & 0.941 (0.029) & 0.421 (0.005) \\
    SSGL & 0.012 (0.002)   & 57.40 (5.481) & - & - \\
    \bottomrule
    \end{tabular}%
    }}
    \caption
    [Comparison of group sparse methods on genetic data]
    {\textit{\mousesecshort}, evaluation of methods on the held out genotype data. Reported is the 10-fold cross-validated MSE, the number of selected groups, the posterior predictive coverage, and mean posterior predictive interval length.
    }
    \label{tab:mgd_results}
\end{table}

\begin{figure}[htp]
    \centering
    \includegraphics[width=1.0\textwidth]{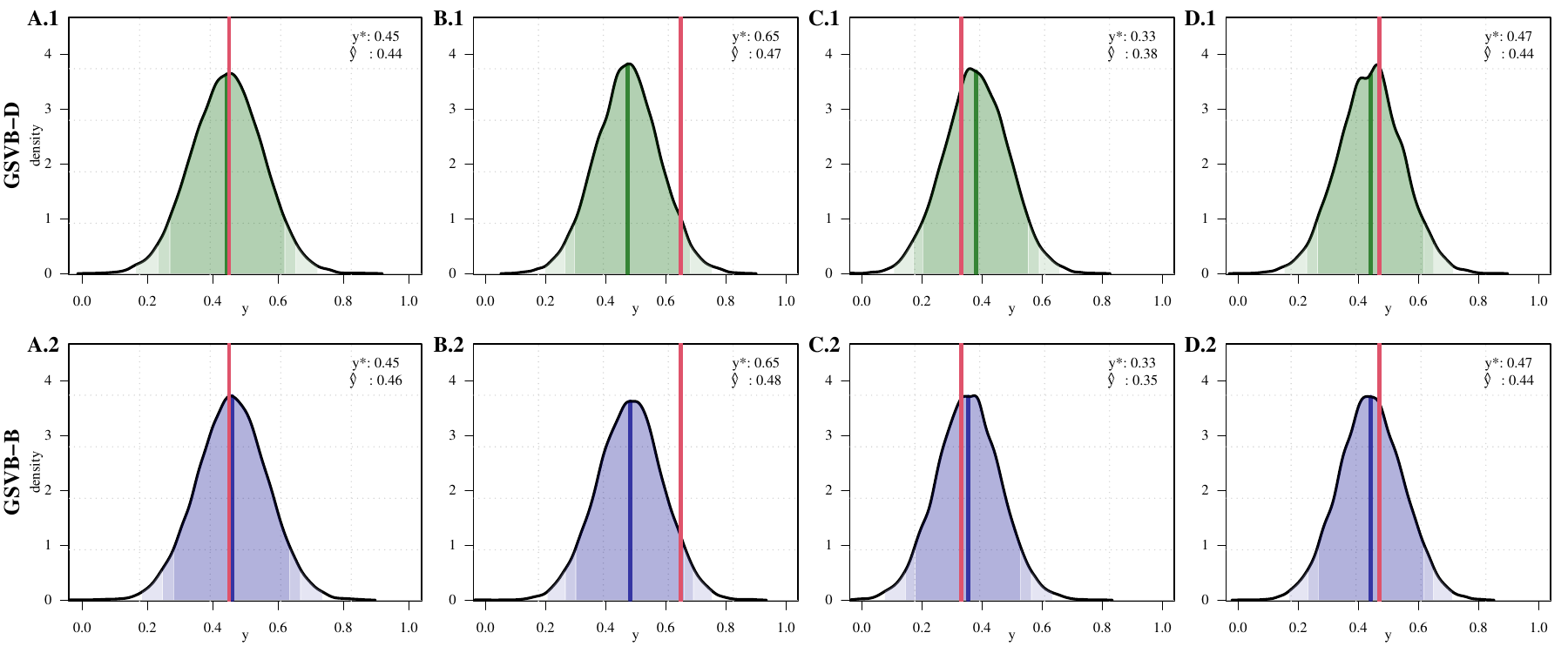}
    {\footnotesize
    \begin{tabular}{l c}
       \textbf{GSVB-D}:
       \hspace{0.4em} PP mean: \hspace{-0.4em} \thickgreen  
       \hspace{1em} Mass\%: \hspace{0.1em}
    	  \colorSquareB{p90!30} 90\% \hspace{0.3em}
    	  \colorSquareB{p95!30} 95\% \hspace{0.3em}
    	  \colorSquareB{p99!30} 99\%
       &  
       \multirow{2}{*}{\hspace{0.2em} \thickredline  \hspace{0.2em} Observed response 
       }
       \\
        \textbf{GSVB-B:}
        \hspace{0.4em} PP mean: \hspace{-0.4em} \thickblue
        \hspace{1em} Mass\%: \hspace{0.1em}
        \colorSquareB{p90B!30} 90\% \hspace{0.3em}
        \colorSquareB{p95B!30} 95\% \hspace{0.3em}
        \colorSquareB{p99B!30} 99\%
         & 
    \end{tabular}
    \vspace{-0.5em}
    } %
    \begin{minipage}{\textwidth}
    \caption
    [Posterior predictive distributions of variational group sparse methods on SNP data]
    {\textit{\mousesecshort}, variational posterior predictive distribution for GSVB-D (panels A.1 - D.1) and GSVB-B (panels A.2 - D.2) constructed for four observations in the held out test set (fold 1). Notably, the red line (\thickredline) represents the measured LDL-C level, with the adjacent text giving this value ($y^*$) and the posterior mean ($\widehat{y}$). The shading of the variational posterior indicates where $90\%$, $95\%$ and $99\%$ of the mass is contained.}
    \label{fig:mouse-pgrs}
    \end{minipage}
\end{figure}

\subsection{Bardet-Biedl syndrome gene expression study}
\label{sec:rat}

Here, we examine microarray data consisting of gene expression measurements for 120 laboratory rats\footnote{available at \url{ftp://ftp.ncbi.nlm.nih.gov/geo/series/GSE5nnn/GSE5680/matrix/}}. Originally, the dataset was studied by \cite{Scheetz2006a} and has subsequently been used to demonstrate the performance of several group selection algorithms \citep{Huang2010, Breheny2015, Bai2020}. Briefly, the dataset consists of normalized microarray data harvested from the eye tissue of 12-week-old male rats. The outcome of interest is the expression of TRIM32, a gene known to cause Bardet-Biedl syndrome, a genetic disease affecting multiple organs, including the retina.

To pre-process the original dataset, which consists of $31,\!099$ probe sets (the predictors), we follow \cite{Breheny2015} and select the $5,\!000$ probe sets that exhibit the largest variation in expression on the log scale. Further, following \cite{Breheny2015} and \cite{Bai2020}, a non-parametric additive model is used, wherein,
$
    y_i = \mu + \sum_{j=1}^p f_j(x_{ij}) + \epsilon_i
$
with $\epsilon_i \sim N(0, \tau^2)$ and $f_j : \R \rightarrow \R$. Here, $y_i \in \R$ corresponds to the expression of TRIM32 for the $i$th observation, and $x_{ij}$ the expression of the $j$th probe set for the $i$th observation. To approximate $f_j$, a three-term natural cubic spline basis expansion is used. The resulting processing gave a group regression problem with $n=120$ and $M=1,\!000$ groups of size $m=3$. Denoting by $\phi_{j, k}$ the $k$th basis for $f_j$, we have
\begin{equation} \label{eq:npam-basis}
    y_i = \mu + \sum_{j=1}^M \sum_{k=1}^m \phi_{j, k}(x_{ij}) + \epsilon_i\;.
\end{equation}

\red{
As before, the performance of the methods is evaluated using ten-fold cross-validation. Overall, GSVB performed excellently; in particular, GSVB--D obtained a comparable 10-fold cross-validated MSE to SSGL, however, with a much smaller model size, meaning the models produced are parsimonious while being comparably predictive (\Cref{tab:rgx}). Furthermore, the PP coverage is in line with the nominal level, particularly for GSVB--D, which has a larger PP coverage than GSVB--B, while also having a smaller PP interval length. 
 }

\blue{
As seen within \Cref{tab:rgx}, GSVB selects on average one group of probes per fold. Among the 10 cross-validation folds, the most commonly selected is the group of probes assigned to gene \textit{KLHL24}. This gene was also reported in the study of \cite{Bai2020} that analysed the same dataset with SSGL. While the proteins encoded by both \textit{TRIM32} and \textit{KLHL24} genes are ubiquitin ligases, to the best of our knowledge no biological link has been made between these two genes in the literature. } 



\vspace{2em}
\begin{table}[htp]
    \centering
    \makebox[\textwidth][c]{
    \resizebox{0.8\textwidth}{!}{ %
    \begin{tabular}{l c c c c}
        \toprule
        Method & MSE & Num. selected groups & PP. coverage & PP. length \\ 
        \midrule
        GSVB-D & 0.012 (0.009) & 1.00 (0.000) & 0.958 (0.071) & 0.465 (0.024) \\
        GSVB-B & 0.012 (0.009) & 1.00 (0.000) & 0.942 (0.069) & 0.532 (0.041) \\
        SSGL & 0.012 (0.011) & 27.80 (3.190) & -- & -- \\
        \bottomrule
    \end{tabular}%
    }}
    \vspace{0.5em}
    \caption[Comparison of group sparse methods on gene expression data]{\textit{Bardet-Biedl Syndrome Gene Expression Study}, evaluation metrics computed on the held-out test data. Average (sd.) over the 10 folds for the metrics are provided: mean squared error, number of selected groups, posterior predictive coverage, and mean posterior predictive interval length.} 
    \label{tab:rgx}
\end{table}

\subsection{MEMset splice site detection}
\label{sec:memset}

The following application is based on the prediction of short read DNA motifs, a task which plays an important role in many areas of computational biology. For example, in gene finding, wherein algorithms such as GENIE \citep{Burge1997} rely on the prediction of splice sites (regions between coding and non-coding DNA segments).

The MEMset donor data has been used to build predictive models for splice sites and consists of a large training and test set\footnote{available at \url{http://hollywood.mit.edu/burgelab/maxent/ssdata/}}. The original training set contains 8,415 true and 179,438 false donor sites, and the test set 4,208 true and 89,717 false donor sites. The predictors are given by 7 factors with 4 levels each (A, T, C, G); a more detailed description is given in \cite{Yeo2004}. Initially analyzed by \cite{Yeo2004}, the data has subsequently been used by \cite{Meier2008} to evaluate the logistic group LASSO. 

To create a predictive model, we follow \cite{Meier2008} and consider all $3$rd order interactions of the 7 factors, which gives a total of $M=64$ groups and $p=1,156$ predictors. To balance the sets, we randomly sub-sampled the training set without replacement, creating a training set of 1,500 true ($Y=1$) and 1,500 false donor sites. Regarding the test set, a balanced set of 4,208 true and 4,208 false donor sites was created. 

In addition to GSVB, we fit SSGL and group LASSO (used in the original analysis), for which we performed 5-fold cross-validation to tune the hyperparameters. To assess the different methods, we use the test data, dividing it into 10 folds, and reporting the: (i) precision $\left( \frac{\text{TP}}{\text{TP + FP}} \right)$, (ii) recall $\left( \frac{\text{TP}}{\text{TP + FN}} \right)$, (iii) F-score $\left( \frac{\text{TP}}{\text{TP + 0.5FP + 0.5FN}} \right)$,  (iv) AUC, and (v) the maximum correlation coefficient between true class membership and the predicted membership over a range of thresholds, $\rho_{\max}$, as used in \cite{Yeo2004} and \cite{Meier2008}. In addition, we report (vi) the number of selected groups, which for the group LASSO is given by the number which have a non-zero $\ell_{2,1}$ norm. 

\begin{table}[htp]
    \centering
\makebox[\textwidth][c]{
\resizebox{1.0\textwidth}{!}{ %
   \begin{tabular}{l c c c c c c}
	\toprule
	Method & Precision & Recall & F-score & AUC & $\rho_{\max}$ & Num. selected groups \\ 
    \midrule
 GSVB--D &
0.915 (0.013) & 0.951 (0.011) & 0.933 (0.009) & 0.975 (0.004) & 0.870 (0.016) &  10 \\
GSVB--B &
0.916 (0.011) & 0.958 (0.013) & 0.936 (0.008) & 0.975 (0.004) & 0.875 (0.016) &  12 \\
SSGL &
0.921 (0.012) & 0.950 (0.011) & 0.935 (0.009) & 0.977 (0.004) & 0.879 (0.015) &  48 \\
Group LASSO &
0.924 (0.013) & 0.952 (0.010) & 0.938 (0.010) & 0.977 (0.004) & 0.882 (0.016) &  60 \\
	\bottomrule
    \end{tabular}
}
}
    \caption[Comparison of group sparse methods on splice site detection data]
    {\textit{MEMset splice site detection}, evaluation metrics computed on the held-out test data. The test set was split into 10 folds, and the mean (sd.) of the metrics are computed.}
    \label{tab:memset}
\end{table}

\begin{figure}[htp]
    \centering
    \makebox[\textwidth][c]{
	   \includegraphics[width=1.1\textwidth]{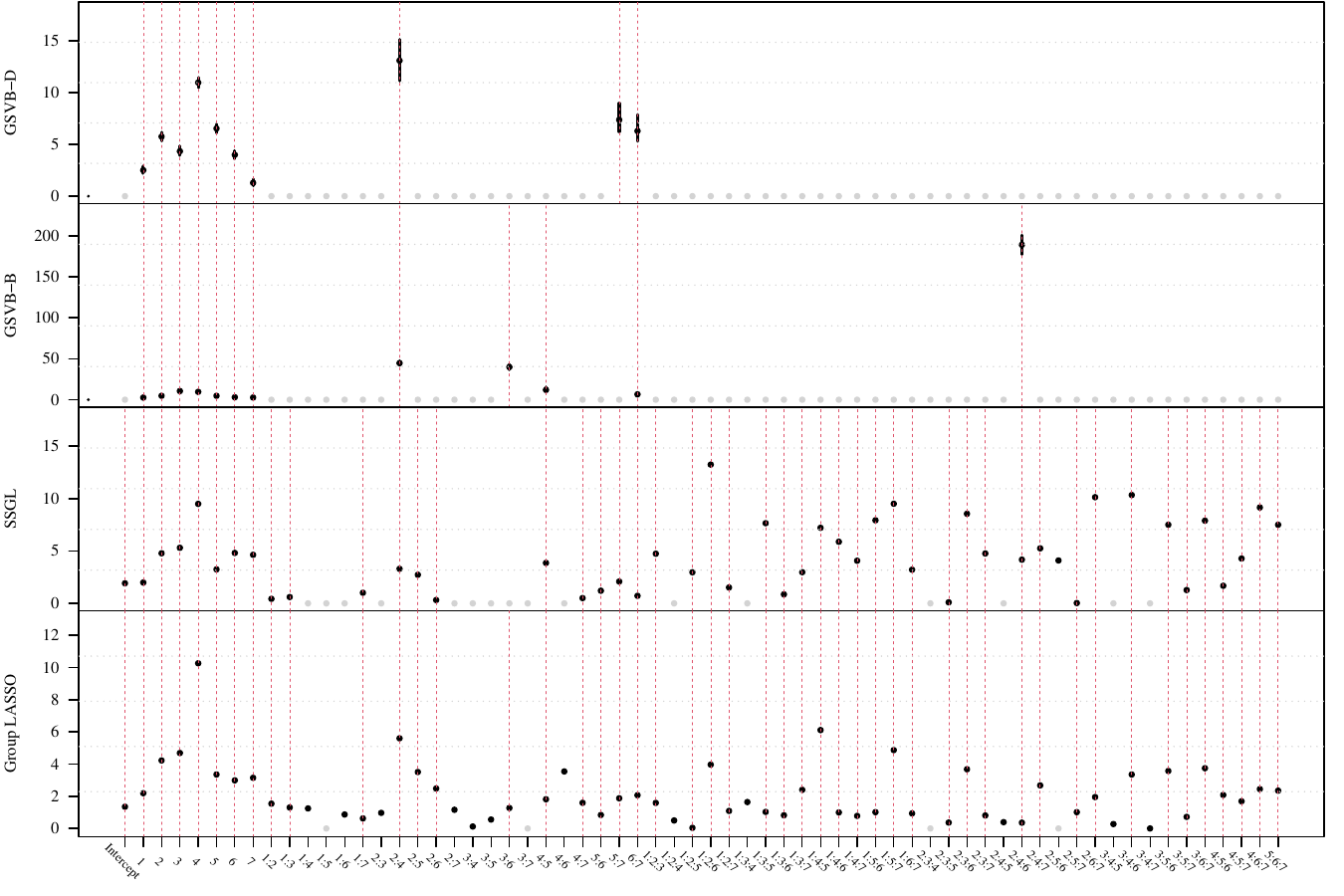}
    } %
    {\footnotesize
	\redline  \hspace{0.2em} Groups selected by more than one method \hspace{1em}
	  \greydot \hspace{0.2em}   Groups for which $\| \beta_{G_k} \| = 0$
   \hspace{0.5em}

	  \blackdot \hspace{0.2em}  Groups for which $\| \widehat{\beta}_{G_k} \| \neq 0$
   \hspace{1em}
	  \blackline \hspace{0.2em}  95\% credible intervals for $ \| \beta_{G_k} \|$ where available 
   \hspace{0.5em}
    }
    \caption[Selected groups by group sparse methods on splice site detection data]
    {\textit{MEMset splice site detection}, comparison of $\ell_2$-norms for each group of coefficients $\widehat{\beta}_{G_k}$ for $k=1,\dots,64$. The red dotted line 
    shows groups which have been selected by more than one method. The points indicate $\| \widehat{\beta}_{G_k} \|$ where $\widehat{\beta}_{G_k}$ is the posterior mean for GSVB, the MAP for SSGL, and the MLE estimate under the group LASSO. The points are colored in black ({\raisebox{1pt}{\footnotesize $\bullet$}}) when they are non-zero and grey 
    ({\raisebox{.6pt} {\color{lightgray} \small $\bullet$}}) otherwise. Finally, when available, the 95\% credible set for $\| \beta_{G_k} \| $ is given by the solid black line.
    }
    \label{fig:memset_norms}
\end{figure}

Overall, the methods performed comparably, with GSVB--D obtaining the largest recall and the group LASSO the largest precision, F-score, AUC, and $\rho_{\max}$ by a small margin (\Cref{tab:memset}). However, GSVB returned models that were of smaller size and therefore more parsimonious than those of SSGL and the group LASSO. This is further highlighted in \Cref{fig:memset_norms}, which showcases the fact that the models obtained by GSVB are far simpler, selecting groups 1--7 as well as the interactions between groups 2:4 and 6:7 for both methods, with GSVB--D selecting 5:7 and GSVB--B selecting 2:6, 4:6, and 2:4:6 in addition to these. Further, \Cref{fig:memset_norms} showcases the fact that uncertainty is available about the norms of the groups. 
\section{Discussion}

In this paper, we have introduced GSVB, a scalable method for group-sparse general linear regression. We have shown how fast coordinate ascent variational inference algorithms can be constructed and used to compute the variational posterior. We have provided theoretical guarantees for the proposed VB approach in the setting of grouped sparse linear regression by extending the theoretical work of \cite{CSV15} and \cite{RS22} for group sparse settings. 

Through extensive numerical studies, we have demonstrated that GSVB provides state-of-the-art performance, offering a computationally inexpensive substitute to MCMC, whilst also performing comparably or better than MAP methods. 
Furthermore, through our analysis of real-world datasets, we have highlighted the practical utility of our method. Demonstrating that GSVB provides parsimonious models with competitive predictive performance, and, as demonstrated in \Cref{sec:mice}, selects variables with established biological significance.


\blue{One limitation of the GSVB method lies in its sensitivity to the initialisation of the variational Bayes parameters. In particular, if a purely random initialization is used, the CAVI algorithm may fail to converge. This limitation reflects a broader challenge shared by many gradient-based optimization methods. We proposed a heuristic initialisation strategy that consistently yielded robust and satisfactory performance in our simulation studies and we recommend that practitioners adopt this initialisation scheme. However, further examination of different initialization and update strategies could be an insightful avenue for future work which might lead to further improvements in runtime and performance}

\red{
Furthermore, different to MAP and frequentist methods, GSVB provides scalable uncertainty quantification, which serves as a powerful tool in several application areas. Nevertheless, VB methods are known to underestimate the posterior variance. Methods have been proposed to correct this issue, see, for example \cite{Giordano2018}, however to our knowledge these methods are yet to be extended to the high-dimensional setting, which may prove an interesting avenue for future research.
} 

\blue{Finally, the proposed method can be extended to the sparse-group setting, enabling a bi-level selection of variables within selected groups, as it has been done in the setting of sparse group lasso \citep{Breheny2009, Simon2013}. Bayesian bi-level selection has been explored in the literature by \cite{Xu2015} who proposed spike-and-slab priors for sparse group lasso. GSVB can be potentially adapted for bi-selection by extending the prior and variational family to include a spike-and-slab prior at the coefficient level within each group. For example, instead of using a multivariate Laplace distribution in \eqref{eq:prior}, one could use a spike-and-slab mixture instead. }




\bibliography{refs.bib}

\newpage
\appendix
\clearpage

\pagenumbering{arabic}
\setcounter{algorithm}{0}
\setcounter{page}{1}
\setcounter{table}{0}
\setcounter{figure}{0}
\numberwithin{equation}{section}


\section{Co-ordinate ascent algorithms} \label{appendix:gsvb_derivations}

Derivations are presented for the variational family $\Q'$, noting that $\Q \subset \Q'$, hence the update equations under $\Q$ follow directly from those under $\Q'$. Recall, the family is given as,
\begin{equation}
    \Q' =
    \left\{ Q'(\mu, \Sigma, \gamma) = 
	\bigotimes_{k=1}^M Q_k'(\mu_{G_k}, \Sigma_{G_k}, \gamma_k) :=
	\bigotimes_{k=1}^M 
	\left[ 
	    \gamma_k\ N\left(\mu_{G_k}, \Sigma_{G_k} \right) + 
	    (1-\gamma_k) \delta_0
	\right] 
    \right\} 
\end{equation}
where $\Sigma \in \R^{p \times p}$ is a covariance matrix for which $\Sigma_{ij} = 0$, for $i \in G_k, j \in G_l, k \neq l$ (i.e. there is independence between groups) and $\Sigma_{G_k} = (\Sigma_{ij})_{i, j \in G_k} \in \R^{m_k \times m_k}$ denotes the covariance matrix of the $k$th group.

\subsection{Gaussian Family}

Under the Gaussian family, $Y_i \overset{\text{iid}}{\sim} N(x_i^\top \beta, \tau^2)$, and the $\log$-likelihood is given by,
\begin{equation*}
    \ell(\D; \beta, \tau^2) = - \frac{n}{2}\log(2\pi\tau^2) - \frac{1}{2\tau^2} \| y - X \beta \|^2.
\end{equation*}
Recall, we have chosen to model $\tau^2$ by an inverse-Gamma prior 
which has density,
\begin{equation*}
    \frac{b^a}{\Gamma(a)} x^{-a - 1} \exp \left(\frac{-b}{x} \right)
\end{equation*}
where $a, b >0$ and in turn we extend $\Q'$ to $\Q'_\tau = \Q' \times \{ \IG(a', b') : a', b > 0\}$.

Under this variational family, the expectation of the negative log-likelihood is given by,
{\allowdisplaybreaks
\begin{align}
    \E_{Q'_\tau} 
    & \left[ 
	- \ell(\D; \beta) 
    \right]  
= 
    \E_{Q'_\tau} 
    \left[ 
	\frac{n}{2} \log(2\pi \tau^2) 
	+ \frac{1}{2\tau^2} \| y - X \beta \|^2
    \right] 
    \nonumber \\
= &\
    \E_{Q'_\tau} 
    \left[ 
	\frac{n}{2} \log(2\pi \tau^2) +
	\frac{1}{2\tau^2} \left( 
	    \|y \|^2 + \| X \beta \|^2 - 2 \langle y, X\beta \rangle 
	\right)
    \right] 
    \nonumber \\
= &\
    \frac{n}{2}(\log(2\pi) + \log(b') - \kappa(a'))
    \nonumber \\
+ &\
    \frac{a'}{2b'}
    \Bigg(
	\| y \|^2 
	+
	\left(
	\sum_{i=1}^p \sum_{j=1}^{p} 
	    (X^\top X)_{ij} 
	    \E_{Q'_\tau} \left[ \beta_i \beta_j \right] 
	\right)
	 - 2 \sum_{k=1}^M \gamma_k \langle y, X_{G_k} \mu_{G_k} \rangle
    \Bigg)
\end{align}
} %
where the expectation 
\begin{equation*}
    \E_{Q'_\tau} \left[ \beta_i \beta_j \right] = \begin{cases}
	\gamma_k \left( \Sigma_{ij} + \mu_{i} \mu_{j} \right) & \quad i,j \in G_k \\
	\gamma_k \gamma_h \mu_{i}\mu_{j}       & \quad i \in G_k, j \in G_h, h \neq k
    \end{cases}
\end{equation*}

\section{Gibbs Sampler} \label{appendix:gibbs}

We present a Gibbs sampler for the Gaussian family of models, noting that the samplers for the Binomial and Poisson family use the same principles. We begin by considering a slight alteration of the prior given in \eqref{eq:prior}. Formally, 
\begin{equation}
\begin{aligned}
    \beta_{G_k} \overset{\text{ind}}{\sim} &\ \Psi(\beta_{G_k}; \lambda) \\
z_k | \theta_k \overset{\text{ind}}{\sim} &\ \text{Bernoulli}(\theta_k) \\
    \theta_k \overset{\text{iid}}{\sim} &\ \text{Beta}(a_0, b_0) \\
\end{aligned}
\end{equation}
for $k=1,\dots,M$ and $\tau^2 = \xi \overset{\text{iid}}{\sim} \Gamma^{-1}(\xi; a, b)$. Under this prior writing the likelihood as,
\begin{equation}
    p(\D | \beta, z, \xi) = \prod_{i=1}^n \phi \left(y_i; f \left( \sum_{k = 1}^M z_k \langle x_{G_k}, \beta_{G_k} \rangle \right), \xi \right)
\end{equation}
where $\phi(\cdot; \mu, \sigma^2)$ is the density of the Normal distribution with mean $\mu$ and variance $\sigma^2$, and $f$ is the link function (which in this case is the identity), yields the posterior
\begin{equation} \label{eq:posterior_mcmc} 
    p(\beta, z, \xi | \D) \propto p(\D | \beta, z, \xi) \pi(\xi) \prod_{k=1}^M \pi(\beta_{G_k}) \pi(z_k | \theta_k) \pi(\theta_k).
\end{equation}
Notably, the posterior is equivalent to our previous formulation.

To sample from \eqref{eq:posterior_mcmc}, we construct a Gibbs sampler as outlined in \Cref{alg:mcmc_sampler_gaus}. Ignoring the superscript for clarity, the distribution $ \theta_k | \D, \beta, z, \theta_{-k}, \xi $ is conditionally independent of $ \D, \beta, z_{-k}, \xi$ and $\theta_{-k}$. Therefore, $\theta_k$ is sampled from $\theta_k | z_k$, which has a $\text{Beta}(a_0 + z_k, b_0 + 1 - z_k)$ distribution. Regarding $z_j^{(i)}$, the conditional density
\begin{align} \label{eq:mcmc_zj} 
    p(z_k | \D, \beta, z_{-k}, \theta, \xi) \propto &\ p(\D | \beta, z_{-k}, z_k, \theta, \xi) \pi(z_k | \beta, z_{-k}, \theta, \xi). \nonumber \\
    = &\ p(\D; \beta, z, \xi) \pi(z_k | \theta_k). 
\end{align}
As $z_k$ is discrete, evaluating the RHS of \eqref{eq:mcmc_zj} for $z_k = 0 $ and $z_k = 1 $, gives the unnormalised conditional probabilities. Summing gives the normalisation constant and thus we can sample $ z_k $ from a Bernoulli distribution with parameter 
\begin{equation}
    p_k = \frac
	{p(z_k = 1 | \D, \beta, z_{-k}, \theta, \xi)}
	{p(z_k = 0 | \D, \beta, z_{-k}, \theta, \xi) + p(z_k = 1 | \D, \beta, z_{-k}, \theta, \xi)}.
\end{equation}
Finally, regarding $ \beta_{G_k, j}^{(i)} $ we use a Metropolis-Hastings within Gibbs step, wherein a proposal $ \beta_{G_k, j}^{(i)} $ is sampled from a random-walk proposition kernel $ K $, where in our implementation $K(x | \vartheta, \varepsilon) = N(x; \vartheta, \sqrt{2} \left( 10^{1 - \varepsilon} \right)^{1/2})$. The proposal is then accepted with probability $A$ or rejected with probability $1-A$, in which case $ \beta_{G_k, j}^{(i)} \leftarrow \beta_{G_k, j}^{(i-1)} $, where $A$ is given by,
\begin{equation*}
    A = \min \left(1, \frac
	{p(\D; \beta_{G_k^c}, \beta_{G_k, -j}, \beta_{G_k, j}^{(i)}, z, \xi) \pi(\beta_{G_k, j}^{(i)} | \beta_{G_k, -j}) } 
	{p(\D; \beta_{G_k^c}, \beta_{G_k, -j}, \beta_{G_k, j}^{(i-1)}, z, \xi) \pi(\beta_{G_k, j}^{(i-1)} | \beta_{G_k, -j}) } 
	\frac
	    {K(\beta_j^{(i-1)} |\beta_j^{(i)}, z_k^{(i)})}
	    {K(\beta_j^{(i)} |\beta_j^{(i-1)}, z_k^{(i-1)})}
    \right)
\end{equation*}

\begin{algorithm}[ht]
    \caption{MCMC sampler for the Gaussian family GSpSL regression}
    \label{alg:mcmc_sampler_gaus}
    \begin{algorithmic}[0]
	\vspace{.3em}
	\State Initialize $\beta^{(0)}, z^{(0)}, \theta^{(0)}, \xi^{(0)}$
	\NoDo \For {$i = 1,\dots, N $}
	    \NoDo \For {$k = 1,\dots, M $}
	    \State $\theta^{(i)}_{k} \overset{\text{iid.}}{\sim} \text{Beta}(a_0 + z_k^{(i-1)}, b_0 + 1 - z_k^{(i-1)})$
	    \EndFor
	    \NoDo \For {$k = 1,\dots, M $}
		\State $z^{(i)}_{k} \overset{\text{ind.}}{\sim} \text{Bernoulli}(p_k)$
	    \EndFor
	    \NoDo \For {$k = 1,\dots, M $}
		\NoDo \For {$j = 1,\dots, m_k $}
		\State $\beta_{G_k, j}^{(i)} \sim p(\beta_{G_k, j}^{(i)} | \D, z^{(i)}, \beta^{(i)}_{G_{1:k-1}}, \beta_{G_{k, 1:j-1}}^{(i)}, \beta_{G_{k, j+1:m_k}}^{(i-1)}, \beta_{G_{k+1:M}}^{(i-1)}, \xi^{(i-1)})$
		\EndFor
	    \EndFor
	    \State Sample $\xi^{(i)} \overset{\text{iid.}}{\sim} \Gamma^{-1}(a + 0.5n, b + 0.5 \| y - \sum_{k=1}^M z_k X_{G_k} \beta_{G_k}^{(i)} \| ^2 )$
	\EndFor
	\State \Return $ \{ \beta^{(i)}, z^{(i)}, \theta^{(i)}, \xi^{(i)} \}_{i=1}^N$.
    \end{algorithmic}
\end{algorithm}

\section{Numerical Study}

\subsection{Runtime}

\begin{table}[htp]
    \centering
\makebox[\textwidth][c]{
\resizebox{0.90\textwidth}{!}{ %
\renewcommand{\arraystretch}{1.0}
    \begin{tabular}{c l l l l l}
\toprule
 & Method & Setting 1 & Setting 2 & Setting 3 & Setting 4 \\
\midrule
\multirow{3}{1.5cm}{Gaus. m=10, s=10}
& GSVB--D &
51.1s (46.8s, 56.1s)    & 50.0s (46.0s, 59.2s)    & 55.2s (49.8s, 5m 47s)   & 5m 42s (2m 57s, 15m 33s) \\
& GSVB--B &
51.8s (47.2s, 56.7s)    & 51.6s (46.5s, 1m 2s)    & 58.3s (49.9s, 5m 51s)   & 5m 48s (2m 1s, 15m 28s) \\
& SSGL &
1m 3s (59.5s, 1m 13s)   & 50.8s (41.6s, 59.5s)    & 1m 1s (58.1s, 1m 16s)   & 1m 24s (1m 11s, 2m 3s)  \\
\midrule
\multirow{3}{1.5cm}{Gaus. m=10, s=20}
& GSVB--D &
2m 53s (2m 42s, 2m 2s)  & 2m 51s (2m 41s, 2m 5s)  & 7m 34s (2m 48s, 8m 59s) & 15m 39s (5m 42s, 1h 17m) \\
& GSVB--B &
2m 53s (2m 42s, 2m 3s)  & 2m 53s (2m 43s, 2m 9s)  & 7m 42s (2m 49s, 8m 10s) & 18m 43s (5m 1s, 1h 32m) \\
& SSGL &
2m 42s (1m 23s, 2m 3s)  & 1m 16s (59.2s, 2m 40s)  & 2m 6s (2m 37s, 2m 30s)  & 2m 16s (2m 31s, 6m 32s) \\
\midrule
\multirow{3}{1.5cm}{Binom. m=5, s=5}
& GSVB--D &
\multicolumn{1}{c}{--} & 3m 38s (2m 50s, 4m 45s) & 2m 24s (2m 31s, 4m 9s)  &  2m 47s (1m 14s, 4m 31s)\\
& GSVB--B &
\multicolumn{1}{c}{--} & 3m 23s (2m 27s, 5m 49s) & 3m 9s (2m 13s, 5m 3s)   &  2m 15s (2m 37s, 4m 14s)\\
& SSGL &
\multicolumn{1}{c}{--} & 2m 54s (2m 44s, 3m 42s) & 3m 33s (2m 57s, 3m 26s) &  2m 40s (2m 31s, 2m 3s)\\
\midrule
\multirow{3}{1.5cm}{Pois. m=5, s=3}
& GSVB--D &
16.3s (9.8s, 21.5s)      & 11.2s (8.2s, 20.8s)     & 16.0s (10.4s, 43.0s)      & 12.9s (8.6s, 48.7s)    \\
& GSVB--B &
2m 26s (2m 59s, 3m 20s)  & 2m 55s (1m 29s, 3m 20s) & 3m 53s (2m 49s, 9m 25s)   & 3m 33s (2m 4s, 9m 11s) \\
& SSGL &
4m 11s (3m 49s, 8m 38s) & 4m 34s (2m 11s, 13m 4s) & 2m 29s (2m 56s, 16m 33s) & 2m 32s (1m 25s, 3m 2s)\\
\bottomrule
    \end{tabular}
    }
    }
    \caption[Runtime comparison of GSVB and SSGL]{Median (5\%, 95\% quartile) runtimes for numerical experiments presented in \Cref{fig:gaus_comp}. 
    }
    \label{tab:runtime_comp}
\end{table}

\subsection{Evaluation of initialization schemes and update ordering} \label{c5:sec:initialization}

\red{
Here the impact of the initialization and update order on the performance of GSVB is investigated. To do so, data is simulated for $(n, p, m, s) = (200, 1000, 5, 10)$, $(200, 1000, 10, 10)$ and $(100, 200, 1, 20)$, under the Gaussian family for settings 1 -- 4. 
Three different initialization schemes: (i) the group LASSO solution with small $\lambda$, (ii) random initialization where $\beta_j \sim N(0, 0.5^2)$, (iii) the ridge solution were considered. Additionally, three different update orders were considered: (i) sequential, (ii) random, (iii) magnitude where the groups are updated in order of decreasing $\| \mu_{G_k} \|$. 

In \Cref{fig:sen_init}, where the update order is fixed to the magnitude ordering (the best scheme amongst those considered), we see that group LASSO initialization consistently outperformed the other initialization schemes, yielding the lowest $\ell_2$-error and highest AUC across the different settings. We notice that random initialization may not be suitable as GSVB failed to converge in some settings. Additionally, in \Cref{fig:sen_order}, where we fix the initialization to the group LASSO solution and vary the update order, we see that the magnitude-based ordering yielded the best results, whereas the random update ordering scheme performed the worst. Overall, the results suggest that the magnitude ordering in combination with the group LASSO initialization is the most robust and should be used when applying the algorithm in practice. 
}

\begin{figure}[htp]
    \centering
    \vspace*{\fill}
    \makebox[\textwidth][c]{
	\includegraphics[width=1.02\textwidth]{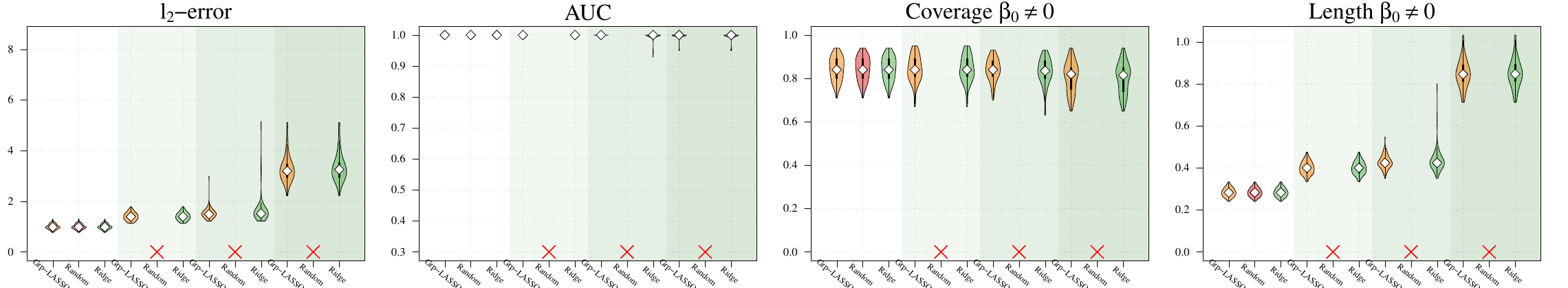}
    } %
    \makebox[\textwidth][c]{
	\includegraphics[width=1.02\textwidth]{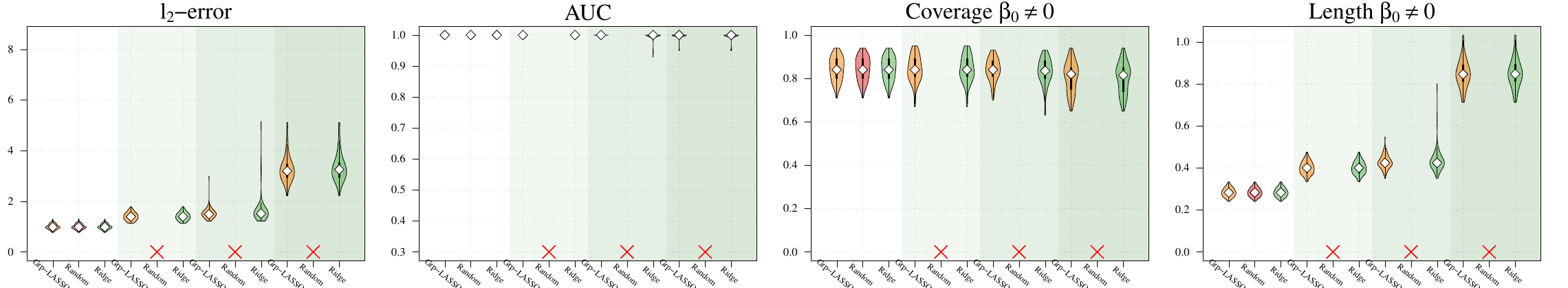}
    } %
    \makebox[\textwidth][c]{
	\includegraphics[width=1.02\textwidth]{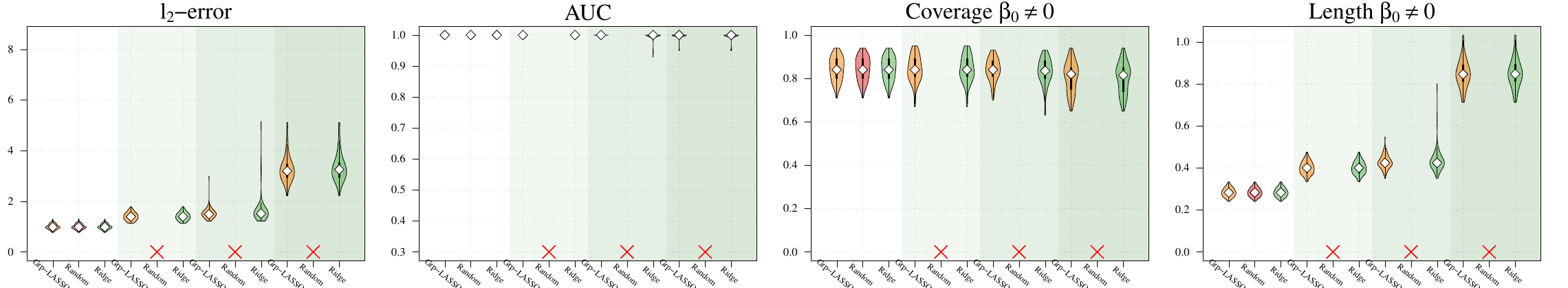}
    } %
    {\footnotesize
    \vspace{0.5em}
	\colorSquareB{setting1!15}  Setting 1 \hspace{0.5em}
	\colorSquareB{setting2!15}  Setting 2 \hspace{0.5em}
	\colorSquareB{setting3!15}  Setting 3 \hspace{0.5em}
	\colorSquareB{setting4!15}  Setting 4 \hspace{0.5em}
	
    \vspace{-0.5em}
	{\color{dnf} \large $\times$} 50+ runs did not converge \hspace{0.5em}
	\colorSquareB{mcmc!50} Group LASSO \hspace{0.5em}
	\colorSquareB{gsvb-b!50} Random \hspace{0.5em}
	\colorSquareB{gsvb-d!50}   Ridge
    }
    \caption[Comparison of initialization schemes for GSVB]{\red{Performance evaluation of different initialization schemes for the Gaussian family using the magnitude update ordering across 100 runs. \textbf{Row 1}: $(n, p, m, s) = (200, 1000, 5, 10)$, \textbf{Row 2:} $(n,p,m,s)=(200,1000,10,10)$, and \textbf{Row 3:} $(n,p,s,m)=(100,200,1,20)$.  
    For each method, the white diamond (\whitediamond) indicates the median of the metric, the thick black line (\blacklinethick) the interquartile range, and the black line (\blacklinethin) 1.5 times the interquartile range.
    }}
    \label{fig:sen_init}
    \vspace*{\fill}
\end{figure}

\begin{figure}[htp]
    \centering
    \makebox[\textwidth][c]{
	\includegraphics[width=1.02\textwidth]{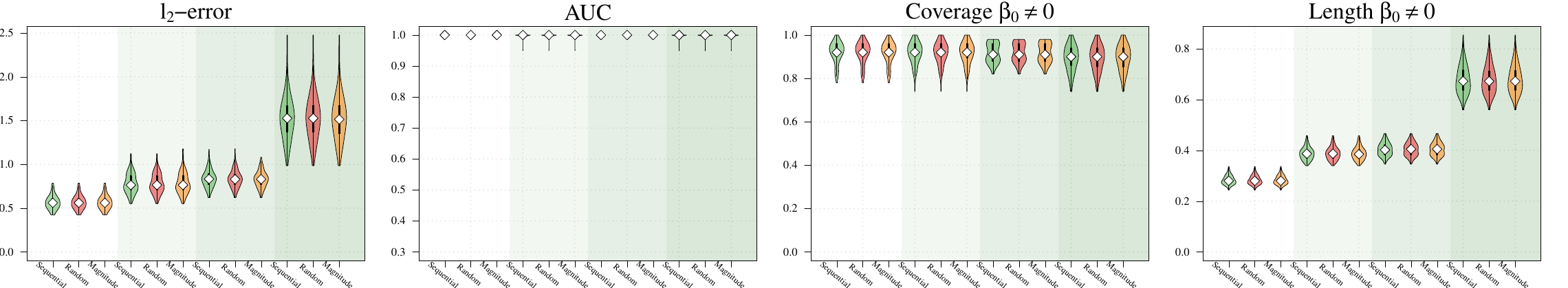}
    } %
    \makebox[\textwidth][c]{
	\includegraphics[width=1.02\textwidth]{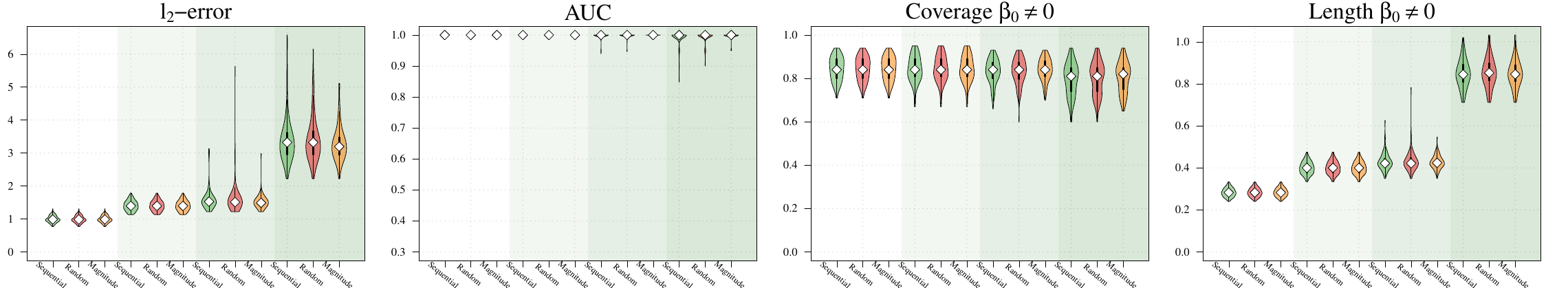}
    } %
    \makebox[\textwidth][c]{
	\includegraphics[width=1.02\textwidth]{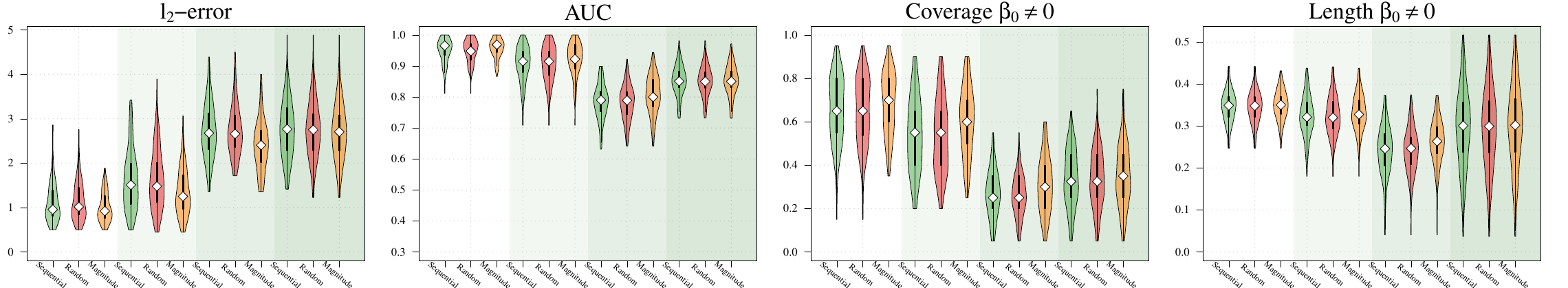}
    } %
    {\footnotesize
    \vspace{0.5em}
	\colorSquareB{setting1!15}  Setting 1 \hspace{0.5em}
	\colorSquareB{setting2!15}  Setting 2 \hspace{0.5em}
	\colorSquareB{setting3!15}  Setting 3 \hspace{0.5em}
	\colorSquareB{setting4!15}  Setting 4 \hspace{0.5em}
	
    \vspace{-0.5em}
	\colorSquareB{gsvb-d!50} Sequential \hspace{0.5em}
	\colorSquareB{gsvb-b!50} Random \hspace{0.5em}
	\colorSquareB{mcmc!50}   Magnitude
    }
    \caption[Comparison of parameter update ordering schemes for GSVB]{\red{Performance evaluation of different ordering schemes for the Gaussian family using the group LASSO as initialization across 100 runs. \textbf{Row 1}: $(n, p, m, s) = (200, 1000, 5, 10)$, \textbf{Row 2:} $(n,p,m,s)=(200,1000,10,10)$, and \textbf{Row 3:} $(n,p,s,m)=(100,200,1,20)$.  For each method, the white diamond (\whitediamond) indicates the median of the metric, the thick black line (\blacklinethick) the interquartile range, and the black line (\blacklinethin) 1.5 times the interquartile range.
    }}
    \label{fig:sen_order}
\end{figure}

\subsection{Effect of coefficient magnitudes}

\red{
The effect of the magnitude of the coefficients on the performance of GSVB is examined. For this we consider the Gaussian family with $(n, p, m, s) = (200, 1000, 5, 10)$ and vary the magnitude of the coefficients to be $0.05 - 2.0$, sampling the non-zero coefficients from $U(-\beta_{\max}, \beta_{\max})$. The results, presented in \Cref{fig:gaus_bmax} show that the performance of GSVB can be sensitive to the magnitude of the coefficients, with the performance of GSVB-D being more sensitive than GSVB-B, and the performance of both methods improving as the magnitude of the coefficients increases. 
}

\begin{figure}[H]
    \centering
    \makebox[\textwidth][c]{
	\includegraphics[width=1.0\textwidth]{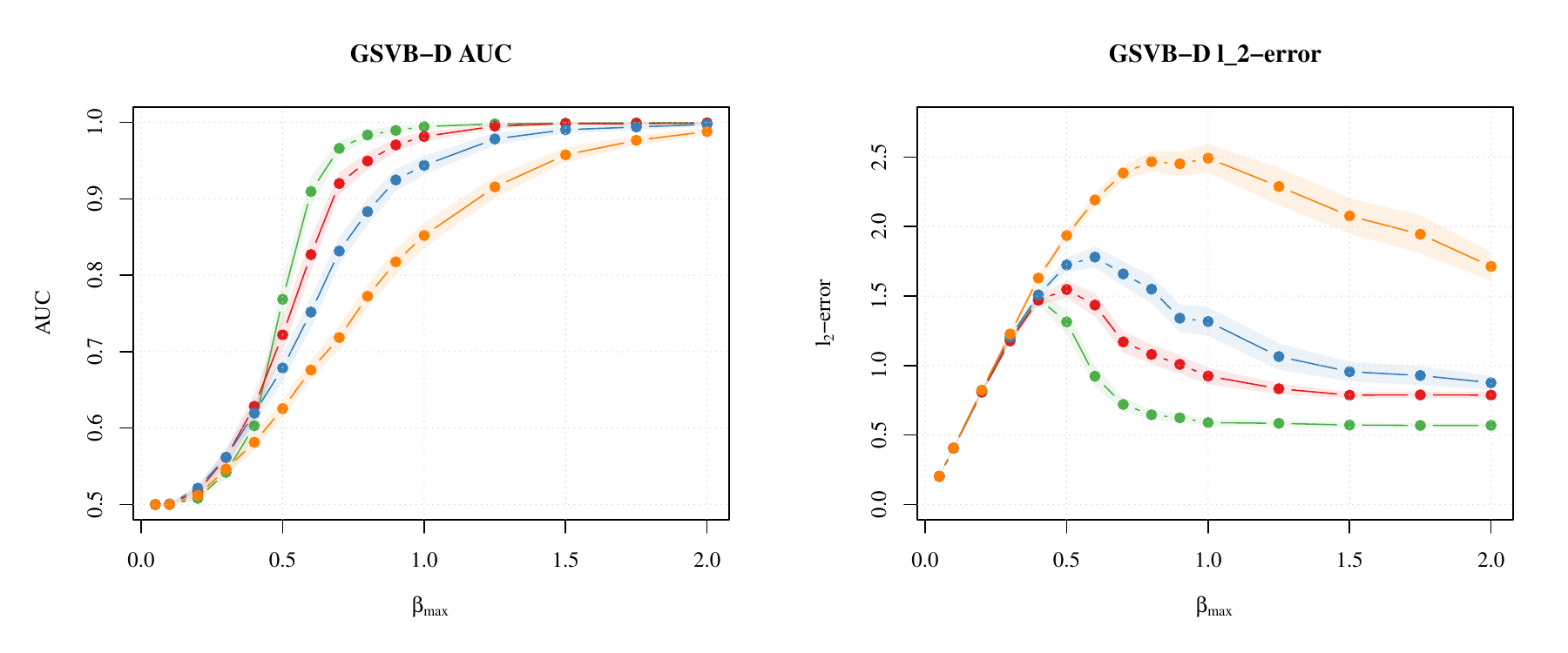}
    } %
    \vspace{-0.5em}
    \makebox[\textwidth][c]{
    \includegraphics[width=1.0\textwidth]{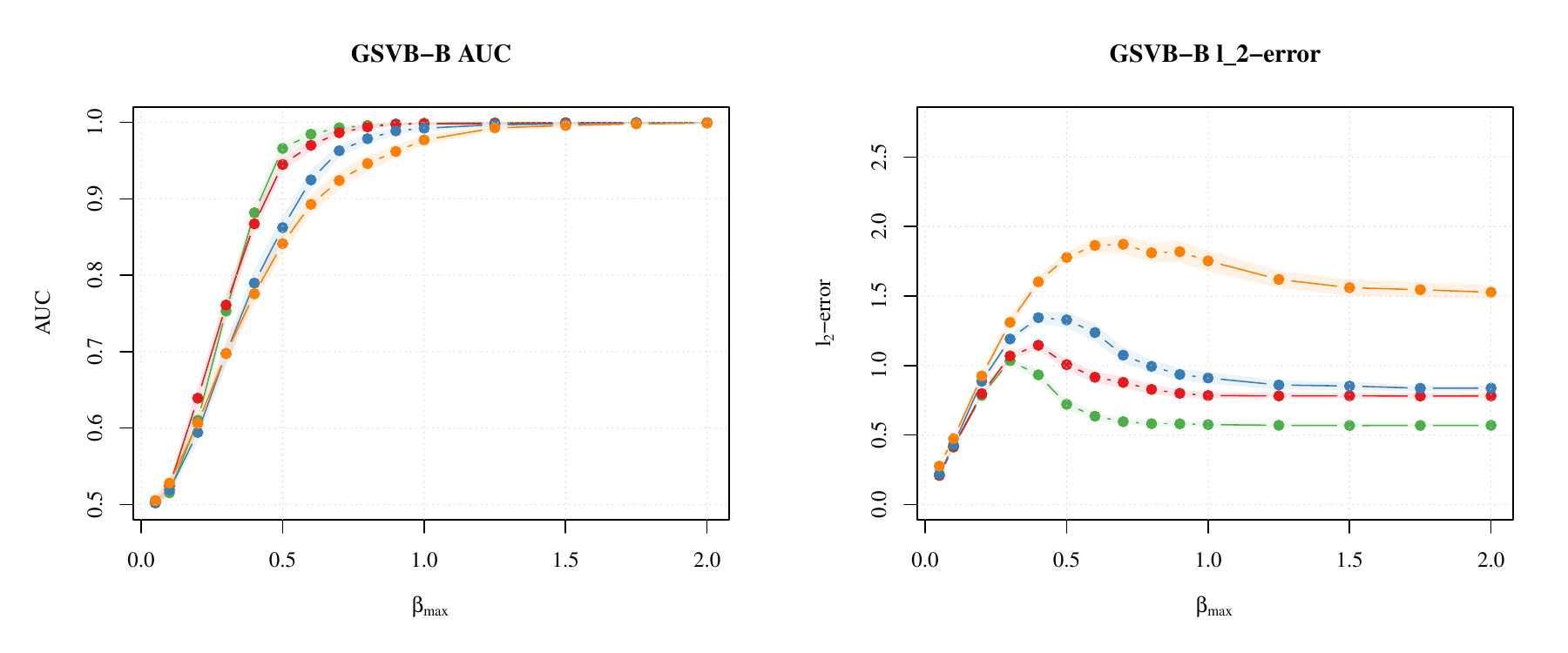}
    } %
    {\footnotesize
	\thickcolorline{gsvb-d}  Setting 1 \hspace{0.5em}
	\thickcolorline{gsvb-b}  Setting 2 \hspace{0.5em}
	\thickcolorline{ssgl}  Setting 3 \hspace{0.5em}
	\thickcolorline{mcmc}  Setting 4 \hspace{0.5em}
    }
    \caption{Performance of GSVB-D (\textbf{row 1}) and GSVB-B (\textbf{row 2}) for Settings 1--4 with $(n, p, m, s)=(200, 1,\!000, 5, 10)$ across 100 runs. Here the non-zero coefficients are sampled $U(-\beta_{\max}, \beta_{\max})$, with 
    $\beta_{\max}$ being varied from 0.05 to 2.0. The thick line presents the mean and the shaded area shows the $95\%$ interval of the metric.
    }
    \label{fig:gaus_bmax}
\end{figure}

\subsection{Performance with large groups}

The performance of GSVB with large groups is examined. For this we fix $(n, p, s) = (1,\!000, 5,\!000, 10)$ and vary the group size to be $m=10,20,25,50,100$. The results, presented in \Cref{fig:gaus_large_groups} highlight that the performance of the method is excellent for groups of size $m=10,20,25$, and begins to suffer when groups are of size $m=50$ and $100$, particularly in settings 3 and 4 wherein GSVB--B did not run.

\begin{figure}[htp]
    \centering
    \makebox[\textwidth][c]{
	\includegraphics[width=0.51\textwidth]{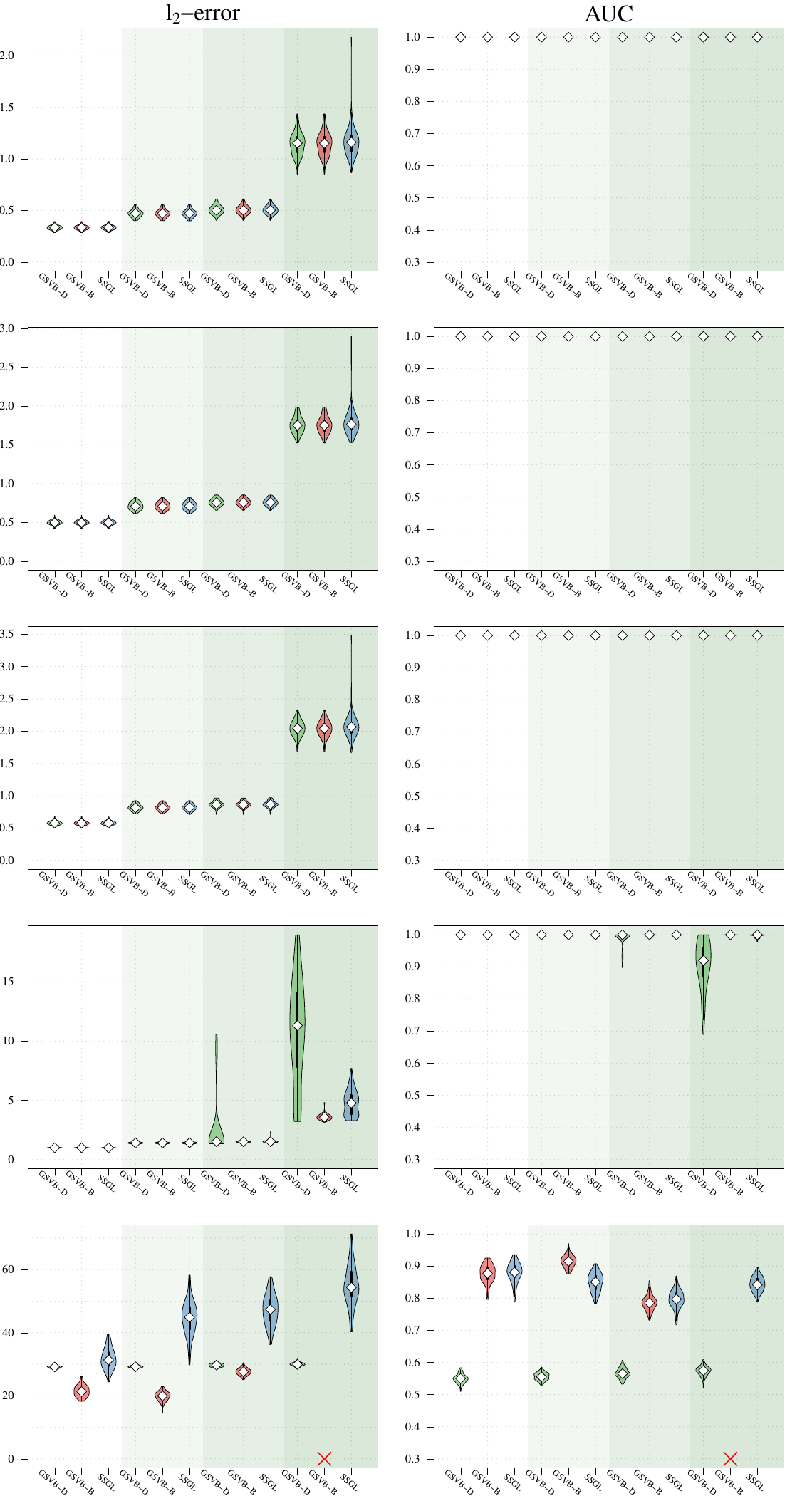}
	\includegraphics[width=0.51\textwidth]{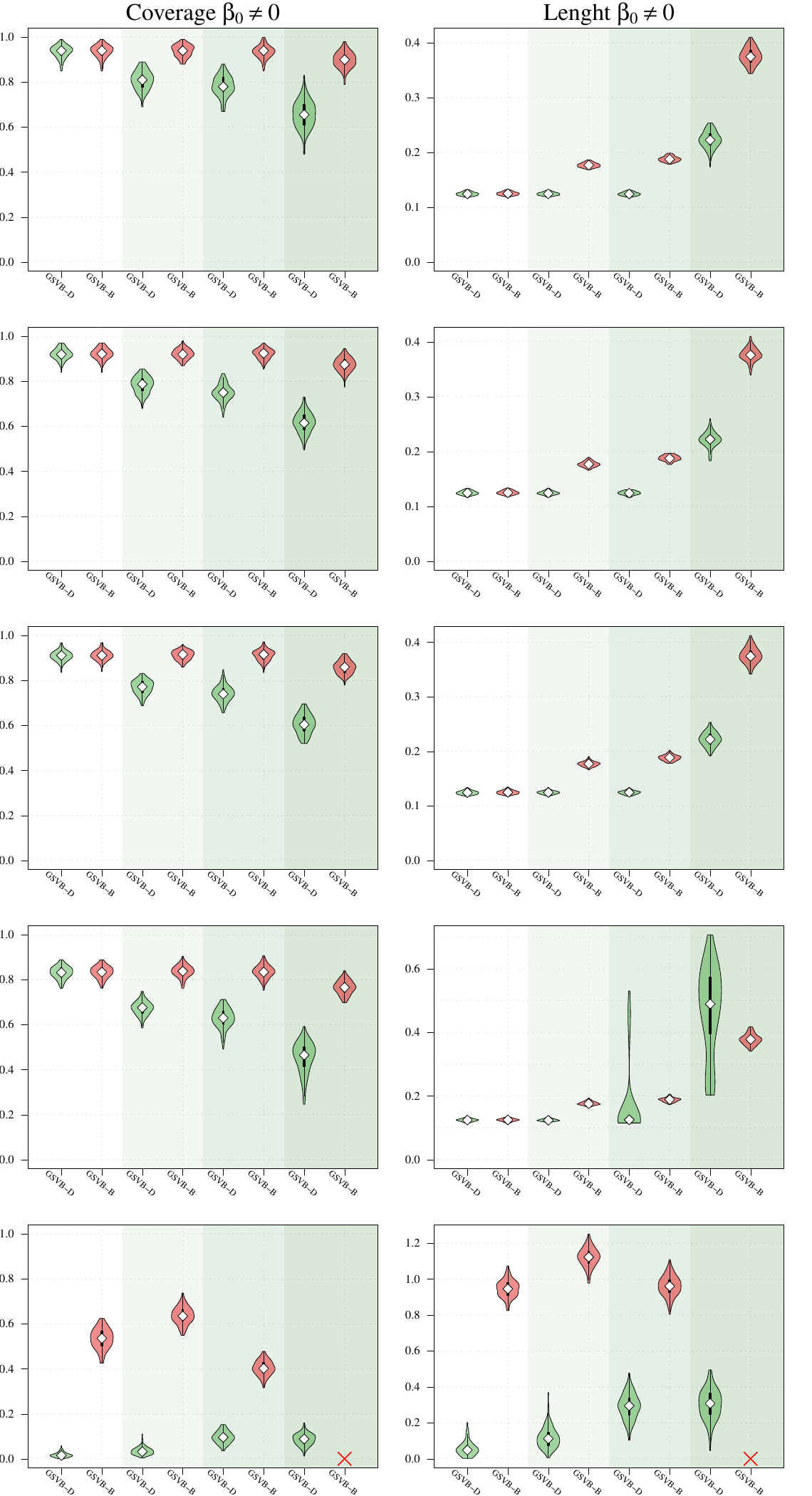}
    } %
    {\footnotesize
	\colorSquareB{setting1!15} Setting 1 \hspace{0.5em}
	\colorSquareB{setting2!15}  Setting 2 \hspace{0.5em}
	\colorSquareB{setting3!15}  Setting 3 \hspace{0.5em}
	\colorSquareB{setting4!15}  Setting 4 \hspace{0.5em}
	
	\vspace{-0.5em}
	  {\color{dnf} \large $\times$} 95+ runs did not converge \hspace{0.5em}
	\colorSquareB{gsvb-d!50} GSVB-D \hspace{0.5em}
	\colorSquareB{gsvb-b!50} GSVB-B \hspace{0.5em}
	\colorSquareB{ssgl!50} SSGL
    }
    \caption{Performance evaluation of GSVB and SSGL for Settings 1--4 with $(n, p, s)=(1\!,000, 5\!,000, 10)$ across 100 runs. For each method the white diamond (\whitediamond) indicates the median of the metric, the thick black line (\blacklinethick) the interquartile range, and the black line (\blacklinethin) 1.5 times the interquartile range.  
    \textbf{Rows 1--5}: Gaussian family with increasing group sizes of $m = 10, 20, 25, 50, 100$.
    }
    \label{fig:gaus_large_groups}
\end{figure}

\section{Proofs of asymptotic results}\label{sec:prooofs}

\subsection{A general class of model selection priors and an overview of the proof}

Our theoretical results apply to a wider class of model selection priors than the group spike and slab prior \eqref{eq:prior} underlying our variational approximation. In this section, we thus consider this more general class of model selection priors \cite{CSV15,NJG20,RS22} defined hierarchically via:
\begin{equation}
\begin{split}
s \sim \pi_M(s) & \\
S||S|=s \sim \text{Unif}_{M,s} & \\
\beta_{G_k} \sim^{ind} \begin{cases}
\Psi_{\lambda,m_k}(\beta_{G_k}), ~~ & G_k \in S,\\
\delta_0, & G_k\not\in S,
\end{cases}
\end{split}
\end{equation}
where $\pi_M$ is a prior on $\{0,1,\dots,M\}$, $\text{Unif}_{M,s}$ is the uniform distribution on subsets $S \subseteq \{1,\dots,M\}$ of size $s$ and
\begin{equation}\label{eq:Psi}
\Psi_{\lambda,m_k}(\beta_{G_k})  = \Delta_{m_k}  \lambda^{m_k} \exp \left( - \lambda \|\beta_{G_k}\|_2 \right)
\end{equation}
is a density on $\R^{m_k}$ with $\Delta_{m_k} = 2^{-m_k} \pi^{(1-m_k)/2} \Gamma((m_k+1)/2)^{-1}$. This can be concisely written as
\begin{align}\label{eq:prior_concise}
(S,\beta) \mapsto \pi_M(|S|) \frac{1}{ {M\choose |S|} } \delta_0(S^c) \prod_{G_k \in S} \Psi_{\lambda,m_k} (\beta_{G_k}).
\end{align}
Following \cite{CSV15,NJG20,RS22}, we assume as usual that
\begin{align}\label{eq:prior_condition}
A_1 M^{-A_3} \pi_M(s-1) \leq \pi_M(s) \leq A_2 M^{-A_4} \pi_{M}(s-1), \qquad s = 1,\dots,M.
\end{align}
We further recall the assumption \eqref{eq:lambda} made on the scale parameter:
\begin{equation}\label{eq:lambda}
\underline{\lambda} \leq \lambda \leq 2\bar{\lambda}, \qquad \qquad \underline{\lambda}=\frac{\|X\|}{M^{1/\mm}}, \qquad \qquad \bar{\lambda}= 3 \|X\| \sqrt{\log M}.
\end{equation}
The group spike and slab prior fits within this framework by taking $\pi_M = \text{Bin}(M,\tfrac{a_0}{a_0+b_0})$, and hence the following results immediately imply Theorems \ref{thm:contraction} and \ref{thm:dimension}. Recall the parameter set
$$\mathcal{B}_{\rho_n,s_n} =\{\beta_0 \in \R^p:\phi(S_{\beta_0}) \geq c_0, ~ |S_{\beta_0}|\leq s_n, ~ \widetilde{\phi}(\rho_n |S_{\beta_0}|) \geq c_0 \}$$
defined in \eqref{eq:Bn}.

\begin{theorem}[Contraction]\label{thm:contraction_general}
Suppose that Assumption \ref{ass:theory} holds, the prior satisfies \eqref{eq:prior_condition}-\eqref{eq:lambda} and $s_n$ satisfies $\mm \log s_n \leq K' \log M$ for some $K'>0$. Then the variational posterior $\tilde{\Pi}$ based on either the variational family $\mathcal{Q}$ in \eqref{eq:family_1} or $\mathcal{Q}'$ in \eqref{eq:family_2} satisfies, with $s_0 = |S_{\beta_0}|$,
$$\sup_{ \beta_0\in \mathcal{B}_{\rho_n,s_n} } E_{\beta_0} \tilde{\Pi} \left( \beta: \|X(\beta-\beta_0)\|_2 \geq  \frac{H_0 \rho_n^{1/2} \sqrt{s_0\log M}}{\bar{\phi}(\rho_n s_0)} \right) \to 0$$
$$\sup_{ \beta_0\in \mathcal{B}_{\rho_n,s_n} } E_{\beta_0} \widetilde{\Pi} \left( \beta: \|\beta-\beta_0\|_{2,1} \geq  \frac{H_0 \rho_n s_0 \sqrt{\log M}}{\|X\| \bar{\phi}(\rho_n s_0)^2} \right) \to 0, $$
$$\sup_{ \beta_0\in \mathcal{B}_{\rho_n,s_n} } E_{\beta_0} \widetilde{\Pi} \left( \beta: \|\beta-\beta_0\|_2 \geq  \frac{H_0\rho_n^{1/2} \sqrt{s_0\log M}}{ \|X\| \widetilde{\phi}(\rho_n s_0)^2}  \right) \to 0 , $$
for any $\rho_n \to \infty$ (arbitrarily slowly), $\mathcal{B}_{\rho_n,s_n}$ defined in \eqref{eq:Bn} and where $H_0$ depends only on the prior.
\end{theorem}

\begin{theorem}[Dimension]\label{thm:dimension_general}
Suppose that Assumption \ref{ass:theory} holds, the prior satisfies \eqref{eq:prior_condition}-\eqref{eq:lambda} and $s_n$ satisfies $\mm \log s_n \leq K' \log M$ for some $K'>0$. Then the variational posterior $\tilde{\Pi}$ based on either the variational family $\mathcal{Q}$ in \eqref{eq:family_1} or $\mathcal{Q}'$ in \eqref{eq:family_2} satisfies
$$\sup_{ \beta_0\in \mathcal{B}_{\rho_n,s_n} } E_{\beta_0} \tilde{\Pi} \left( \beta: |S_\beta| \geq \rho_n |S_{\beta_0}|  \right) \to 0$$
for any $\rho_n \to \infty$ (arbitrarily slowly) and $\mathcal{B}_{\rho_n,s_n}$ defined in \eqref{eq:Bn}.
\end{theorem}

To prove Theorems \ref{thm:contraction_general} and \ref{thm:dimension_general}, we use the following result which relates the VB probability of sets having exponentially small probability under the true posterior.

\begin{lemma}[Theorem 5 of \cite{RS22}]\label{lem:post_to_VB}
Let ${B}_n$ be a subset of the parameter space, $A_n$ be events and $Q_n$ be distributions for $\beta$. If there exists $C>0$ and $\delta_n>0$ such that
\begin{align*}
E_{\beta_0} \Pi(\beta \in {B}_n^c|Y) 1_{A_n} \leq Ce^{-\delta_n},
\end{align*}
then
$$E_{\beta_0} Q_n(\beta \in {B}_n^c) 1_{A_n} \leq \frac{2}{\delta_n} \left[ E_{\beta_0} \KL (Q_n\|\Pi(\cdot|Y))1_{A_n} + Ce^{-\delta/2} \right].$$
\end{lemma}

{\color{red}Applying Lemma \ref{lem:post_to_VB} with $Q_n = \tilde{\Pi}$ the VB posterior,} the proof thus reduces to finding events $A_n$ with $P_{\beta_0}(A_n) \to 1$ on which:
\begin{enumerate}
\item The true posterior places only exponentially small probability outside $B_n$, that is $\Pi (B_n^c|Y) \leq Ce^{-\delta_n}$ for some rate $\delta_n \to \infty$,
\item The $\KL$-divergence between the VB posterior $\tilde{\Pi}$ and the full posterior is $o(\delta_n)$.
\end{enumerate}
In our setting, we shall take $\delta_n = C s_0 \log M$. Full posterior results are dealt with in Section \ref{sec:contraction_proofs}, the $\KL$-divergence in Section \ref{sec:KL_proofs} and the proofs of Theorems \ref{thm:contraction_general} and \ref{thm:dimension_general} are completed in Section \ref{sec:proof_completion}.

\subsection{Asymptotic theory for the full posterior}\label{sec:contraction_proofs}

We now establish contraction rates for the full computationally expensive posterior distribution, keeping track of the exponential tail bounds needed to apply Lemma \ref{lem:post_to_VB}. While the proofs in this section largely follow those in Castillo et al. \cite{CSV15}, the precise arguments adapting these results to the group sparse setting are rather technical and hence we provide them for convenience.  We first establish a Gaussian tail bound in terms of the group structure.
\begin{lemma}\label{lem:T0}
The event
\begin{equation}\label{eq:T0}
\T_0 = \left\{\max_{k=1,\dots,M} \|X_{G_k}^T(Y-X\beta_0)\|_2 \leq 3\|X\| \sqrt{\log M} \right\}
\end{equation}
satisfies $\sup_{\beta_0 \in \R^p} P_{\beta_0}(\T_0^c) \leq 2/M$.
\end{lemma}

\begin{proof}
Since $Y-X\beta_0 = \eps \sim N_n(0,I_n)$ under $P_{\beta_0}$, applying a union bound over the group structure yields $P_{\beta_0} (\max_k \|X_{G_k}^T (Y-X\beta_0)\|_2 > t) \leq \sum_{k=1}^M P_{\beta_0} ( \|X_{G_k}^T \eps\|_2 >t)$. Recall that for a multivariate normal $W \sim N_m (0,\Sigma)$, we have $P(\|W\|_2 - E\|W\|_2 \geq x) \leq 2 \exp(-x^2 /(2E\|W\|_2^2))$ by Corollary 3 of \cite{PS86}. Since $X_{G_k}^T \eps \sim N_{m_k} (0,X_{G_k}^T X_{G_k})$, this implies
$$(E\|X_{G_k}^T \eps\|_2)^2 \leq E\|X_{G_k}^T \eps\|_2^2 = \Tr(X_{G_k}^T X_{G_k} )  \leq  \|X\|^2,$$
and hence $P( \|X_{G_k}^T \eps\|_2 \geq  \|X\| +x) \leq 2 \exp(-x^2/(2 \|X\|^2))$ for any $x>0$. Substituting this bound with $x= 2 \|X\| \sqrt{\log M}$ into the above union bound thus yields
$$P_{\beta_0} (\max_k \|X_{G_k}^T (Y-X\beta_0)\|_\infty > 3 \|X\| \sqrt{\log M}) \leq 2M e^{-2\log M} = 2M^{-1}.$$
\end{proof}

Let $\ell_n(\beta)$ denote the log-likelihood of the $N_n(X\beta,I_n)$-distribution. For any $\beta \in \R^p$, we have log-likelihood ratio
\begin{equation}\label{eq:ll}
\Lambda_{\beta}(Y) = e^{\ell_n(\beta)-\ell_n(\beta_0)} = e^{-\frac{1}{2}\|X(\beta-\beta_0)\|_2^2 + (Y-X\beta_0)^TX(\beta-\beta_0)}.
\end{equation}
The next result establishes an almost sure lower bound on the denominator of the Bayes formula. It follows Lemma 2 of \cite{CSV15}, but must be adapted to account for the uneven prior normalizing factors coming from the group structure.

\begin{lemma}\label{lem:evidence_lb}
Suppose Assumption \ref{ass:theory} holds and that $\beta_0 \in \R^p$ satisfies $\mm \log s_0 \leq K' \log M$ for $K'>0$. Then it holds that with $P_{\beta_0}$-probability one,
\begin{align*}
\int \Lambda_\beta(Y) d\Pi(\beta) &\geq C \pi_m(s_0)  s_0!  e^{-\lambda \|\beta_0\|_{2,1}} e^{-c s_0 \log M},
\end{align*}
where $C,c>0$ depend only on $K,K'$.
\end{lemma}

The mild condition $\mm \log s_0 \leq K' \log M$, which will be assumed throughout, relates the true sparsity with the maximal group size. As in Assumption \ref{ass:theory}, the constant $K'$ is used to provide uniformity, but this can be ignored at first reading.

\begin{proof}
The bound trivially holds true for $s_0=0$, hence we assume $s_0\geq 1$. Write $p_0:= \sum_{G_k \in S_0} m_k$ to be maximal number of non-zero coefficients in $\beta_0$. In an abuse of notation, we shall sometimes interchangeably use $\beta_{S_0}$ for both the vector in $\R^{p_0}$ and the vector in $\R^p$ with the entries in $S_0^c$ set to zero. Using the form \eqref{eq:prior_concise} of the prior,
\begin{align*}
\int \Lambda_\beta(Y) d\Pi(\beta) &\geq \frac{\pi_M(s_0)}{{M \choose s_0} } \int_{\R^{p_0}} \Lambda_\beta(Y) \prod_{G_k \in S_0} \Psi_{\lambda,m_k}(\beta_{G_k}) d\beta_{G_k} .
\end{align*}
By the change of variable $b_{S_0} = \beta_{S_0}-\beta_{0,S_0}$ and the form of the log-likelihood \eqref{eq:ll}, the last display is lower bounded by
\begin{align*}
 \frac{\pi_M(s_0)}{{M \choose s_0} } e^{-\lambda \|\beta_0\|_{2,1}} \int_{\R^{p_0}} e^{-\frac{1}{2}\|X_{S_0} b_{S_0}\|_2^2 + (Y-X\beta_0)^T X_{S_0} b_{S_0} } \prod_{G_k \in S_0} \Psi_{\lambda,m_k}(b_{G_k}) db_{G_k} .
\end{align*}
Define the measure $\mu$ on $\R^{p_0}$ by $d\mu(b_{S_0}) = e^{-\frac{1}{2}\|X_{S_0}b_{S_0}\|_2^2} \prod_{G_k \in S_0} \Psi_{\lambda,m_k}(b_{G_k}) db_{G_k}$. Let $\bar{\mu} = \mu/\mu(\R^{p_0})$ denote the normalized probability measure with corresponding expectation $E_{\bar{\mu}}$. Defining $Z(b_{S_0}) = (Y-X\beta_0)^T X_{S_0} b_{S_0}$, Jensen's inequality implies $E_{\bar{\mu}}e^Z \geq e^{E_{\bar{\mu}}Z} = 1$, since $E_{\bar{\mu}}Z=0$ as $\bar{\mu}$ is a symmetric probability distribution about zero. Thus the last display is lower bounded by $ \pi_m(s_0)  e^{-\lambda \|\beta_0\|_{2,1}}\mu(\R^{p_0})/{M \choose s_0}$, which equals
\begin{align}\label{eq:lb1}
\frac{\pi_M(s_0)}{{M \choose s_0} } e^{-\lambda \|\beta_0\|_{2,1}} \int_{\R^{p_0}} e^{-\frac{1}{2}\|X_{S_0} b_{S_0}\|_2^2  } \prod_{G_k \in S_0} \Psi_{\lambda,m_k}(b_{G_k}) db_{G_k} .
\end{align}
Using the group structure, $\|X b\|_2 = \| \sum_{k=1}^M X_{\cdot G_k} b_{G_k}\|_2 \leq \sum_{k=1}^M \|X_{\cdot G_k}\|_2 \|b_{G_k}\|_2 \leq \|X\| \|b\|_{1}$ since $\|b\|_2 \leq \|b\|_1$. Using the form \eqref{eq:Psi} of the density $\Psi_{\lambda,m_k}$ and recalling that
$$\int_{\|\beta_S\|_1\leq r} (\lambda/2)^{|S|} e^{-\lambda \|\beta_S\|_1} d\beta_S \geq e^{-\lambda r} (\lambda r)^{|S|}/|S|!$$
by (6.2) in \cite{CSV15}, the integral in the last display is bounded below by
\begin{equation}\label{eq:integral}
\begin{split}
& e^{-1/2} \int_{\|X\| \|b_{S_0}\|_1 \leq 1} \prod_{G_k \in S_0} \Delta_{m_k} \lambda^{m_k} e^{-\lambda \|\beta_{G_k}\|_1} d\beta_{G_k} \\
& \qquad \geq e^{-1/2} e^{-\lambda /\|X\|} \frac{ \lambda ^{p_0}}{\|X\|^{p_0} p_0! } \prod_{G_k \in S_0} \Delta_{m_k} 2^{m_k}.
\end{split}
\end{equation}
Deviating from \cite{CSV15}, we must now take careful account of the normalizing constants $\Delta_{m_k}$.

Recall the form of the normalized constants $\Delta_{m_k} = 2^{-m_k} \pi^{(1-m_k)/2} \Gamma((m_k+1)/2)^{-1}$.
The non-asymptotic upper bound in Stirling's approximation for the Gamma function gives for $z\geq 2$:
$\Gamma(z) \leq \sqrt{2\pi(z-1)} \left(\frac{z-1}{e} \right)^{z-1} e^{\frac{1}{12(z-1)}}.$
Taking $z = (m_k+1)/2 \geq 2$ for $m_k\geq 3$, 
\begin{align}\label{eq:Gamma_lb}
\Gamma((m_k+1)/2) \leq 2 \sqrt{\pi(m_k-1)} \left( \frac{m_k-1}{2e}\right)^{(m_k-1)/2} \leq  2e^{-1/2} \sqrt{\pi} m_k^{m_k/2} \frac{1}{(2e)^{(m_k-1)/2}},
\end{align}
where we have used that $(\frac{m_k-1}{m_k})^{m_k/2} \leq \lim_{x\to\infty}  (1-\frac{1}{x})^{x/2} =  e^{-1/2} $ since the function in the limit is strictly increasing on $(1,\infty)$. One can directly verify that the upper bound \eqref{eq:Gamma_lb} also holds for $m_k=1,2$. Using \eqref{eq:Gamma_lb},
\begin{align*}
\prod_{G_k \in S_0} \Delta_{m_k} 2^{m_k} = \prod_{G_k \in S_0} \frac{\pi^{(1-m_k)/2} }{\Gamma((m_k+1)/2)} & \geq \frac{1}{(\sqrt{2}e^{-1})^{s_0} (2\pi e)^{p_0/2}} \prod_{G_k \in S_0} m_k^{-m_k/2} 
& \geq c^{p_0} \ \mm^{-p_0/2}
\end{align*}
for some universal constant $c>0$. Using this last display, we lower bound \eqref{eq:integral} by a constant multiple of
\begin{align*}
e^{-\lambda /\|X\|} \left(\frac{ \lambda}{\|X\|}\right)^{p_0} \frac{1}{p_0!} c^{p_0} \mm^{-p_0/2}.
\end{align*}
Using \eqref{eq:lambda}, if $\lambda /\|X\| \leq 1/2$, then $e^{-\lambda /\|X\|} (\lambda /\|X\|)^{p_0} \geq e^{-1/2} M^{-p_0/\mm} \geq e^{-1/2} e^{-s_0 \log M}$, while if $\lambda /\|X\| \geq 1/2$, then $e^{-\lambda /\|X\|} (\lambda /\|X\|)^{p_0} \geq e^{-6 \sqrt{\log M}} 2^{-p_0} \geq e^{-C(K) s_0\log M}$ 
since $\mm \leq K\log M$ by assumption. Since also $\mm^{-p_0/2}/p_0! \geq e^{-2p_0 \log p_0} \geq e^{-2\mm s_0 \log (\mm s_0)}$, the last display is lower bounded by $e^{-Cs_0 \log M}$ under the lemma's hypotheses. The result then follows by substituting this lower bound for the integral in \eqref{eq:lb1} and using that ${M\choose s_0} \leq M^{s_0}/s_0! = e^{s_0\log M}/s_0!$.
\end{proof}

The next result follows Theorem 10 of \cite{CSV15}.

\begin{lemma}[Dimension]\label{lem:dimension}
Suppose that Assumption \ref{ass:theory} holds and the prior satisfies \eqref{eq:prior_condition} and \eqref{eq:lambda}. Further assume $M>0$ is large enough that $\log M\geq \max \{ \frac{4\mm \log 4}{A_2}, \frac{4\log A_2}{A_2}\}$ and $M \geq (4^{\mm+1/2} A_2)^{1/A_4}$. Then for any $\beta_0 \in \R^p$ such that $\mm \log s_0 \leq K' \log M$ and any $L>0$,
\begin{align*}
& E_{\beta_0} \Pi \left( \left. \beta: |S_\beta| \geq  (L+1)s_0 \right| Y \right)1_{\T_0} \\
& \qquad \quad \leq C(K,K') \exp \left\{ \left( c(K,K',A_2) + \frac{144}{\phi(S_0)^2} -\frac{LA_4}{2} \right) s_0 \log M\right\},
\end{align*}
where $s_0 = |S_{\beta_0}|$ and $\T_0$ is the event \eqref{eq:T0}.
\end{lemma}

\begin{proof}
Using Bayes formula with likelihood ratio \eqref{eq:ll} and Lemma \ref{lem:evidence_lb}, for any measurable set $B \subset \R^p$,
\begin{equation}\label{eq:Bayes}
\begin{split}
\Pi(B|Y)1_{\T_0} &=  1_{\T_0} \frac{\int_B \Lambda_\beta(Y) d\Pi(\beta)}{\int \Lambda_\beta(Y)d\Pi(\beta)}  \\
&\leq C \frac{e^{\lambda \|\beta_0\|_{2,1} + cs_0 \log M}}{ s_0! \pi_M(s_0) }   \int_B 1_{\T_0} e^{-\frac{1}{2}\|X(\beta-\beta_0)\|_2^2 + (Y-X\beta_0)^TX(\beta-\beta_0)} d\Pi(\beta) 
\end{split}
\end{equation}
for $\mathcal{T}_0$ the event defined in \eqref{eq:T0}. Applying Cauchy-Schwarz on the event $\T_0$ gives
\begin{equation}\label{eq:cross_term}
|(Y-X\beta_0)^T X(\beta-\beta_0)| \leq  \max_{k=1,\dots,M} \|X_{G_k}^T(Y-X\beta_0)\|_2 \sum_{k=1}^M \|\beta_{G_k}-\beta_{0,G_k}\|_2 \leq \bar{\lambda} \|\beta-\beta_0\|_{2,1}.
\end{equation}
Therefore, since $(Y-X\beta)^TX(\beta-\beta_0) \sim N(0,\|X(\beta-\beta_0)\|_2^2)$ under $P_{\beta_0}$, on $\T_0$, the integrand in the second last display is bounded by
\begin{align*}
e^{-\frac{1}{2}\|X(\beta-\beta_0)\|_2^2} E_{\beta_0} [1_{\T_0} e^{ (1-\frac{\lambda}{2\bar{\lambda}})(Y-X\beta_0)^TX(\beta-\beta_0)}] e^{\frac{\lambda}{2} \|\beta-\beta_0\|_{2,1} } & \leq e^{-\frac{1}{2}[1-(1-\lambda/(2\bar{\lambda}))^2 ]\|X(\beta-\beta_0)\|_2^2} e^{\frac{\lambda}{2} \|\beta-\beta_0\|_{2,1} }
\end{align*}
and hence 
\begin{align*}
E_{\beta_0} \Pi(B|Y) 1_{\T_0}  \leq C \frac{e^{\lambda \|\beta_0\|_{2,1} + cs_0 \log M}}{ s_0! \pi_M(s_0) }  \int_B e^{-\frac{1}{2}[1-(1-\lambda/(2\bar{\lambda}))^2 ]\|X(\beta-\beta_0)\|_2^2} e^{\frac{\lambda}{2} \|\beta-\beta_0\|_{2,1} } d\Pi(\beta).
\end{align*}
Arguing exactly as on p. 2007-8 in \cite{CSV15},
$$\|\beta_0\|_{2,1} + \frac{1}{2}\|\beta-\beta_0\|_{2,1} \leq  \frac{1}{8\bar{\lambda}} \|X(\beta-\beta_0)\|_2^2 + \frac{8s_0\bar{\lambda}}{\|X\|^2 \phi(S_0)^2} - \frac{1}{4}\|\beta-\beta_0\|_{2,1} + \|\beta\|_{2,1}$$
and hence
\begin{align}\label{eq:dim_eq1}
E_{\beta_0} \Pi(B|Y)1_{\T_0}  \leq C \frac{e^{ cs_0 \log M}}{ s_0! \pi_M(s_0) } e^{\frac{8s_0\bar{\lambda}\lambda }{\|X\|^2 \phi(S_0)^2}} \int_B  e^{-\frac{\lambda}{4} \|\beta-\beta_0\|_{2,1} + \lambda \|\beta\|_{2,1} } d\Pi(\beta) .
\end{align}
Setting now $B = \{\beta:|S_\beta|>R\}$ with $R \geq s_0$, the integral in the last display is bounded by
\begin{align*}
& \sum_{S:|S|>R} \frac{\pi_M(|S|)}{{M\choose |S|} } \int e^{-\frac{\lambda}{4} \|\beta-\beta_0\|_{2,1}} \prod_{G_k \in S} \Delta_{m_k} \lambda^{m_k} d\beta_{G_k} = \sum_{S:|S|>R}^\infty \frac{\pi_M(|S|)}{{M\choose |S|} } \prod_{G_k \in S} 4^{m_k}.
\end{align*}
Using the prior condition \eqref{eq:prior_condition}, this is then bounded by
\begin{align*}
\sum_{s=R+1}^M \pi_M(s) 4^{\mm s} \leq \pi_M(s_0) 4^{\mm s_0} \left( \frac{4^{\mm} A_2}{M^{A_4}} \right)^{R+1-s_0} \sum_{j=0}^\infty  \left( \frac{4^{\mm} A_2}{M^{A_4}} \right)^{j}.
\end{align*}
Since $4^{\mm} A_2/M^{A_4}\leq 1/2$ by the lemma hypothesis, the last sum is bounded by 2 and hence the first term in \eqref{eq:dim_eq1} is bounded by
\begin{align*}
2C\exp \left\{ cs_0 \log M  + \frac{8s_0\lambda\bar{\lambda}}{\|X\|^2 \phi(S_0)^2} + \mm s_0 \log 4 + (R+1-s_0) \log (4^{\mm} A_2/ M^{A_4}) \right\},
\end{align*}
where $C,c>0$ depend only on $K,K'>0$. Taking $R = (L+1) s_0 -1$, using the lemma hypotheses and that $\lambda \leq 2\bar{\lambda} =  6 \|X\| \sqrt{\log M}$, the last display is bounded by
\begin{align*}
2C \exp \left\{ \left( c + \frac{144}{\phi(S_0)^2} + \frac{ (L+1) \mm  \log 4}{\log M} +\frac{L\log A_2}{\log M} - LA_4  \right) s_0 \log M  \right\}.
\end{align*}
For $M>0$ large enough that $\log M\geq \max \{ \frac{4\mm \log 4}{A_2}, \frac{2\log A_2}{A_2}\}$, this is then bounded by
\begin{align*}
2C \exp \left\{ \left( c + \frac{144}{\phi(S_0)^2} + \frac{ A_2}{4}  - LA_4/2  \right) s_0 \log M  \right\}.
\end{align*}
\end{proof}

We next obtain a contraction rate for the full posterior underlying the variational approximation. The proof follows that of Theorem 3 of \cite{CSV15}, modified for the group setting.

\begin{lemma}[Contraction]\label{lem:contract_full}
Suppose that Assumption \ref{ass:theory} holds and the prior satisfies \eqref{eq:prior_condition} and \eqref{eq:lambda}. Further assume $M>0$ is large enough that $\log M\geq \max \{ \frac{4\mm \log 4}{A_2}, \frac{4\log A_2}{A_2}\}$ and $M \geq (4^{\mm+1/2} A_2)^{1/A_4}$. Then there exists a constant $H_0=H_0(A_1,A_3,A_4)>0$ such that for any $\beta_0 \in \R^p$ such that $\mm \log s_0 \leq K' \log M$ and any $L>0$,
\begin{align*}
& E_{\beta_0} \Pi \left( \left. \beta: \|X(\beta-\beta_0)\|_2 \geq  \frac{H_0 \sqrt{(L+2)s_0\log M}}{\bar{\phi}((L+2)s_0)} \right| Y \right)1_{\T_0} \\
& \qquad \leq C(K,K') \exp \left\{ \left( c(K,K',A_2) + \frac{144}{\phi(S_0)^2} -\frac{LA_4}{2} \right) s_0 \log M\right\},
\end{align*}
where $s_0 = |S_{\beta_0}|$ and $\T_0$ is the event \eqref{eq:T0}. Moreover, both
$$E_{\beta_0} \Pi \left( \left. \beta: \|\beta-\beta_0\|_{2,1} \geq  \frac{H_0(L+2)s_0 \sqrt{\log M}}{\|X\| \bar{\phi}((L+2)s_0)^2} \right| Y \right)1_{\T_0} , $$
$$E_{\beta_0} \Pi \left( \left. \beta: \|\beta-\beta_0\|_2 \geq  \frac{H_0 \sqrt{(L+2)s_0\log M}}{ \|X\| \widetilde{\phi}((L+2)s_0)^2} \right| Y \right)1_{\T_0} , $$
satisfy the same inequality.
\end{lemma}

\begin{proof}
Consider the set $E = E_L =  \{\beta \in \R^p: |S_\beta| \leq (L+1)s_0\}$, which satisfies 
\begin{equation}\label{eq:Ec}
E_{\beta_0}\Pi(E^c|Y) 1_{\T_0} \leq C(K,K') \exp \left\{ \left( c(K,K',A_2) + \frac{144}{\phi(S_0)^2} -\frac{LA_4}{2} \right) s_0 \log M\right\}
\end{equation}
by Lemma \ref{lem:dimension}. Recall that $\lambda \|\beta_0\|_{2,1} \leq 2\bar{\lambda} \|\beta-\beta_0\|_{2,1} + \lambda \|\beta\|_{2,1}$ by \eqref{eq:lambda}. Using this, for any set $B \subseteq E$ and $\T_0 = \{\max_k \|X_{G_k}^T(Y-X\beta_0)\|_2 \leq \bar{\lambda} \}$, \eqref{eq:Bayes} and \eqref{eq:cross_term} give
\begin{align*}
\Pi(B|Y)1_{\T_0} \leq C \frac{e^{cs_0 \log M}}{ s_0! \pi_M(s_0) }   \int_B 1_{\T_0} e^{-\frac{1}{2}\|X(\beta-\beta_0)\|_2^2 + 3\bar{\lambda} \|\beta-\beta_0\|_{2,1} + \lambda \|\beta\|_{2,1}} d\Pi(\beta).
\end{align*}
Note that for any $\beta \in E$, $|S_{\beta-\beta_0}| \leq |S_\beta| +s_0 \leq (L+2)s_0 =: D_L s_0$. Using Definition \ref{def:unif_compat} of the uniform compatibility $\bar{\phi}(s)$,
\begin{align*}
(4-1) \bar{\lambda} \|\beta-\beta_0\|_{2,1} & \leq \frac{4\bar{\lambda} \|X(\beta-\beta_0)\|_2 |S_{\beta-\beta_0}|^{1/2}}{\|X\| \bar{\phi}(|S_{\beta-\beta_0}|)}- \bar{\lambda} \|\beta-\beta_0\|_{2,1}  \\
& \leq \frac{1}{8} \|X(\beta-\beta_0)\|_2^2 + \frac{32 \bar{\lambda}^2 D_Ls_0}{\|X\|^2 \bar{\phi}(D_Ls_0)^2} - \bar{\lambda} \|\beta-\beta_0\|_{2,1}.
\end{align*}
Thus for any $B \subseteq E$,
\begin{align*}
\Pi(B|Y)1_{\T_0} \leq C \frac{e^{cs_0 \log M}}{ s_0! \pi_M(s_0) } e^{\frac{32 \bar{\lambda}^2 D_Ls_0}{\|X\|^2 \bar{\phi}(D_Ls_0)^2}}  \int_B 1_{\T_0} e^{-\frac{3}{8}\|X(\beta-\beta_0)\|_2^2 - \bar{\lambda} \|\beta-\beta_0\|_{2,1} + \lambda \|\beta\|_{2,1}} d\Pi(\beta).
\end{align*}

Note that by \eqref{eq:prior_condition}, $\pi_M(s_0) \geq (A_1 M^{-A_3})^{s_0} \pi_M(0)$. For $B = \{ \beta \in E: \|X(\beta-\beta_0)\|_2 \geq R\}$, the last display implies
\begin{align*}
\Pi(B|Y)1_{\T_0} \leq C \frac{e^{cs_0 \log M}}{ s_0! A_1^{s_0} \pi_M(0)} M^{A_3 s_0}  e^{\frac{32 \bar{\lambda}^2 D_Ls_0}{\|X\|^2 \bar{\phi}(D_Ls_0)^2}}  e^{-3R^2/8} \int_B e^{ - \bar{\lambda} \|\beta-\beta_0\|_{2,1} + \lambda \|\beta\|_{2,1}} d\Pi(\beta).
\end{align*}
The last integral is upper bounded by
\begin{align*}
& \sum_{S:|S|\leq (L+1)s_0} \frac{\pi_M(|S|)}{{M \choose |S|}}\int e^{-\bar{\lambda} \|\beta-\beta_0\|_{2,1}} \prod_{G_k \in S} \Delta_{m_k} \lambda^{m_k} d\beta_{G_k} 
= \sum_{s=0}^{(L+1)s_0} \pi_M(s) \prod_{G_k \in S} (\lambda /\bar{\lambda})^{m_k}.
\end{align*}
Using \eqref{eq:prior_condition} and \eqref{eq:lambda}, the last display is bounded by $\pi_M(0) \sum_{s=0}^\infty (A_2M^{-A_4})^s 2^{\mm s} \leq \pi_M(0) \sum_{s=0}^\infty 4^{-s} \leq 2\pi_M(0)$, since $2^{\mm} A_2 M^{-A_4} \leq 2^{-\mm+1} \leq 1/4$ by the lemma hypothesis. Setting $R^2 = H_0^2 D_L s_0 \log M/\bar{\phi}(D_Ls_0)^2$ for a constant $H_0>0$, the second last display is bounded by
\begin{align*}
& C(K,K') \exp \left\{ \left( c(K,K') + A_3 - \frac{\log A_1}{\log M} + \frac{288 D_L}{\bar{\phi}(D_Ls_0)^2} \right) s_0 \log M  - \frac{3}{8}R^2 \right\} \\
& \quad \leq C \exp \left\{ - \left[ \frac{3H_0^2}{8} -288- c - A_3 - \frac{|\log A_1|}{\log M} \right] \frac{D_L s_0 \log M}{\bar{\phi}(D_Ls_0)^2} \right\},
\end{align*}
where we have used that $\bar{\phi}(D_Ls_0) \leq \bar{\phi}(1) \leq 1$. Taking $H_0$ large enough depending on $A_1,A_3,A_4$ and using again that $\bar{\phi}(D_Ls_0)\leq 1$, the last display can be made smaller than \eqref{eq:Ec}, which is thus the dominant term in the tail probability bound.

For the $\|\cdot\|_{2,1}$-loss, using Definition \ref{def:unif_compat} of the uniform compatibility, for any $\beta \in E$,
\begin{align*}
\bar{\lambda} \|\beta-\beta_0\|_{2,1} & \leq \frac{\bar{\lambda} \|X(\beta-\beta_0)\|_2|S_{\beta-\beta_0}|^{1/2}}{\|X\| \bar{\phi}(|S_{\beta-\beta_0}|)}  
 \leq \frac{1}{2} \frac{\bar{\lambda}^2 D_L s_0}{\|X\|^2 \bar{\phi}(D_Ls_0)^2} + \frac{1}{2}\|X(\beta-\beta_0)\|_2^2.
\end{align*}
The result then follows from the first assertion for the prediction loss $\|X(\beta-\beta_0)\|_2$.

For the $\|\cdot\|_2$-loss, note that $\|X(\beta-\beta_0)\|_2 \geq \widetilde{\phi}(D_Ls_0)\|X\| \|\beta-\beta_0\|_2$ for any $\beta \in E$. Since $\widetilde\phi(s) \leq \bar{\phi}(s)$ for all $s$ by Lemma 1 of \cite{CSV15} (which extends to the group setting), the result follows.
\end{proof}

\subsection{$\KL$-divergence between the variational and true posterior}\label{sec:KL_proofs}

We next bound the KL-divergence between the variational family and the true posterior on the following event:
\begin{equation}\label{eq:T1}
\T_1 = \T_1(\Gamma,\eps,\kappa) = \T_0 \cap \{ \Pi(\beta: |S_\beta| > \Gamma|Y) \leq 1/4\} \cap \{ \Pi(\beta: \|\beta-\beta_0\|_2 > \eps|Y) \leq e^{-\kappa}\},
\end{equation}
where $\Gamma,\kappa,\eps>0$. The proof largely follows Section B.2 of \cite{RS22}, again modified to the group sparse setting, which we produce in full for completeness.

\begin{lemma}\label{lem:KL}
If $4e^{1+\Gamma \log M-\kappa} \leq 1$, then there exists an element $Q \in \mathcal{Q} \subset \mathcal{Q}'$ of the variational families such that
\begin{align*}
\KL (Q||\Pi(\cdot|Y)) 1_{\T_1} & \leq \Gamma \left( \log M +  \mm \log \frac{1}{\widetilde{\phi}(\Gamma)} \right) + \frac{\lambda \Gamma}{\widetilde{\phi}(\Gamma)^2} \left( 3s_0^{1/2} \eps  +  \frac{3\sqrt{\log M}}{\|X\|} + \frac{\mm^{1/2}}{\|X\|} \right) \\
& \qquad + \log (4e).
\end{align*}
If Assumption \ref{ass:theory} also holds, then
\begin{align*}
\KL (Q||\Pi(\cdot|Y)) 1_{\T_1} & \leq \Gamma (K+1) \log M \log \frac{1}{\widetilde{\phi}(\Gamma)} + \frac{\lambda \Gamma}{\widetilde{\phi}(\Gamma)^2} \left( 3s_0^{1/2} \eps  +  \frac{(3+\sqrt{K})\sqrt{\log M}}{\|X\|}  \right) \\
& \quad + \log (4e).
\end{align*}
\end{lemma}

\begin{proof}
The full posterior takes the form
\begin{equation}\label{eq:post}
\Pi(\cdot|Y) = \sum_{S: S\subseteq \{1,\dots,M\}} \hat{q}_S \Pi_S(\cdot|Y) \otimes \delta_{S^c},
\end{equation}
with weights $0\leq \hat{q}_S \leq 1$ satisfying $\sum_S \hat{q}_S = 1$ and where $\Pi_S(\cdot|Y)$ denotes the posterior for $\beta_S\in \R^{\sum_{G_k \in S} m_k}$ in the restricted model $Y = X_S\beta_S + Z$. Arguing exactly as in the proof of Lemma B.2 of \cite{RS22}, one has that on $\T_1$ and for $4e^{1+\Gamma \log M - \kappa}\leq 1$, there exists a set $\tilde{S} \subseteq \{G_1,\dots,G_M\}$ such that
\begin{equation}\label{eq:S}
|\tilde{S}| \leq \Gamma, \qquad \|\beta_{0,S^c}\|_2 \leq \eps, \qquad \hat{q}_{\tilde{S}} \geq (2e)^{-1} M^{-\Gamma}.
\end{equation}

Writing $\tilde{p} = \sum_{G_k \in \tilde{S}}m_k$, consider the element of the variational family
\begin{equation}\label{eq:tQ}
\tilde{Q} = \left( \bigotimes_{G_k \in \tilde{S}} N_{m_k}(\mu_{G_k},D_{G_k}) \right) \otimes \left( \bigotimes_{G_k \in \tilde{S}^c} \delta_{G_k} \right) =: N_{\tilde{p}}(\mu_{\tilde{S}},D_{\tilde{S}}) \otimes \delta_{\tilde{S}^c}
\end{equation}
where $D_{G_k}$ are diagonal matrices to be defined below. This is the distribution that assigns mass one to the model $\tilde{S}$ and then fits an independent normal distribution with diagonal covariance on each of the groups in $\tilde{S}$. Since $\tilde{Q}$ is only absolutely continuous with respect to the $\hat{q}_{\tilde{S}} \Pi_{\tilde{S}}(\cdot|Y)\otimes \delta_{\tilde{S}^c}$ component of the posterior \eqref{eq:post},
\begin{equation}\label{eq:KL1}
\begin{split}
\inf_{Q \in \mathcal{Q}} \KL(Q\| \Pi(\cdot|Y)) \leq \KL (\tilde{Q}\|\Pi(\cdot|Y) &= E_{\beta \sim \tilde{Q}} \left[ \log \frac{dN_{\tilde{p}}(\mu_{\tilde{S}},D_{\tilde{S}}) \otimes \delta_{\tilde{S}^c} }{ \hat{q}_{\tilde{S}} d\Pi_{\tilde{S}}(\cdot|Y) \otimes \delta_{\tilde{S}^c} }(\beta) \right] \\
& = \log \frac{1}{\hat{q}_{\tilde{S}}} +  \KL (N_{\tilde{p}}(\mu_{\tilde{S}},D_{\tilde{S}}) \| \Pi_{\tilde{S}}(\cdot|Y) ).
\end{split}
\end{equation}
On $\T_1$, the first term is bounded by $\log (2eM^{\Gamma})$ by \eqref{eq:S}, so that it remains to bound the second term in \eqref{eq:KL1}.

Define
\begin{equation}\label{eq:vb_parameters}
\mu_{\tilde{S}} = (X_{\tilde{S}}^T X_{\tilde{S}})^{-1} X_{\tilde{S}}^T Y, \qquad D_{\tilde{S}} = \text{diag}((D_{G_k})_{G_k \in \tilde{S}}), \qquad \Sigma_{\tilde{S}} = (X_{\tilde{S}}^T X_{\tilde{S}})^{-1},
\end{equation}
where $D_{G_k} \in \R^{m_k \times m_k}$, $G_k \in \tilde{S}$, is the diagonal matrix with entries
\begin{equation*}
(D_{G_k})_{ii} =\frac{1}{(X_{G_k}^T X_{G_k})_{ii}}, \qquad \qquad i = 1,\dots,m_k.
\end{equation*}
Writing $E_{\mu_{\tilde{S}},D_{\tilde{S}}}$ for the expectation under the distribution $\beta_{\tilde{S}} \sim N_{\tilde{p}}(\mu_{\tilde{S}},D_{\tilde{S}})$,
\begin{equation}\label{eq:KL_split}
\begin{split}
\KL (N_{\tilde{p}}(\mu_{\tilde{S}},D_{\tilde{S}}) \| \Pi_{\tilde{S}}(\cdot|Y) ) & = E_{\mu_{\tilde{S}},D_{\tilde{S}}} \left[ \log \frac{dN_{\tilde{p}}(\mu_{\tilde{S}},D_{\tilde{S}})}{dN_{\tilde{p}}(\mu_{\tilde{S}},\Sigma_{\tilde{S}})} + \log \frac{dN_{\tilde{p}}(\mu_{\tilde{S}},\Sigma_{\tilde{S}})}{d\Pi_{\tilde{S}}(\cdot|Y)}   \right] \\
&  =:(I) + (II).
\end{split}
\end{equation}
We next deal with each term separately.

\textit{Term (I) in \eqref{eq:KL_split}.} Using the formula for the Kullback-Leibler divergence between two multivariate Gaussian distributions,
$$(I) = \KL (N_{\tilde{p}}(\mu_{\tilde{S}},D_{\tilde{S}})\|N_{\tilde{p}}(\mu_{\tilde{S}},\Sigma_{\tilde{S}})) = \tfrac{1}{2}\left( \Tr (\Sigma_{\tilde{S}}^{-1} D_{\tilde{S}}) - \tilde{p} + \log (|\Sigma_{\tilde{S}}|/|D_{\tilde{S}}|) \right),$$
where $|A|$ denotes the determinant of a square matrix $A$. Further define the matrix
$$V_{\tilde{S}} = \text{diag}(((X_{G_k}^T X_{G_k})^{-1})_{G_k \in \tilde{S}}),$$
which is a block-diagonalization of $\Sigma_{\tilde{S}} = (X_{\tilde{S}}^T X_{\tilde{S}})^{-1}$. Let $\tilde{S} = \{G_{k_1},\dots,G_{k_s}\}$ for $s = |\tilde{S}|$. By considering multiplication along the block structure, $\Sigma_{\tilde{S}}^{-1} V_{\tilde{S}}$ has $(i,j)^{th}$-block equal to $(X_{G_{k_i}}^T X_{G_{k_j}}) (X_{G_{k_j}}^T X_{G_{k_j}})^{-1}$, $i,j=1,\dots,s$. Furthermore, $V_{\tilde{S}}^{-1} D_{\tilde{S}}$ is a block-diagonal matrix with $(i,i)^{th}$-block $(X_{G_{k_i}}^T X_{G_{k_i}}) D_{G_{k_i}}$, $i=1,\dots,s$. Thus, by considering the block-diagonal terms,
\begin{align*}
\Tr(\Sigma_{\tilde{S}}^{-1} D_{\tilde{S}}) = \Tr(\Sigma_{\tilde{S}}^{-1} V_{\tilde{S}} V_{\tilde{S}}^{-1} D_{\tilde{S}}) &= \Tr \left( \text{diag} ((X_{G_{k_i}}^T X_{G_k} D_{G_k})_{G_k \in \tilde{S}}) \right) \\
&= \sum_{G_k \in \tilde{S}} \Tr \left( X_{G_{k_i}}^T X_{G_{k_i}} D_{G_{k_i}} \right) = \sum_{G_k \in \tilde{S}} m_k = \tilde{p},
\end{align*}
and hence $(I) = \tfrac{1}{2} \log (|\Sigma_{\tilde{S}}||D_{\tilde{S}}^{-1}|)$. Using \eqref{eq:vb_parameters},
$$|D_{\tilde{S}}^{-1}| = \prod_{G_k \in \tilde{S}} \prod_{i\in G_k} (X_{G_k}^T X_{G_k})_{ii} \leq \prod_{G_k \in \tilde{S}} \|X\|^{2m_k} = \|X\|^{2\tilde{p}}.$$
Let $\lambda_{max}(A)$ and $\lambda_{min}(A)$ denote the largest and smallest eigenvalues, respectively, of a matrix $A$. Arguing as in equation (B.12) in \cite{RS22}, for any $S \subseteq \{G_1,\dots,G_M\}$,
\begin{equation}\label{eq:lambda_min}
\lambda_{min}(X_S^T X_S) \geq \|X\|^2 \widetilde{\phi}(|S|)^2.
\end{equation}
Therefore,
$$|\Sigma_{\tilde{S}}| = 1/|X_{\tilde{S}}^T X_{\tilde{S}}| \leq 1/\lambda_{\min}(X_{\tilde{S}}^T X_{\tilde{S}})^{\tilde{p}} \leq 1/(\|X\| \widetilde{\phi}(|\tilde{S}|))^{2\tilde{p}},$$
and hence $(I) = \tfrac{1}{2}\log (|\Sigma_{\tilde{S}}||D_{\tilde{S}}^{-1}|) \leq \tilde{p} \log (1/\widetilde{\phi}(|\tilde{S}|)) \leq \mm \Gamma \log (1/\widetilde{\phi}(\Gamma))$ using \eqref{eq:S}.

\textit{Term (II) in \eqref{eq:KL_split}}.
One can check that the $N_{\tilde{p}}(\mu_{\tilde{S}},\Sigma_{\tilde{S}})$ distribution has density function proportional to $e^{-\frac{1}{2}\|Y-X_{\tilde{S}}\beta_{\tilde{S}}\|_2^2}$ for $\beta_{\tilde{S}} \in \R^{\tilde{p}}$. Therefore,
\begin{align*}
(II) & = E_{\mu_{\tilde{S}},D_{\tilde{S}}} \left[ \log \frac{D_\Pi e^{-\frac{1}{2}\|Y-X_{\tilde{S}} \beta_{\tilde{S}}\|_2^2 - \lambda \|\beta_{0,\tilde{S}}\|_{2,1} }}{D_N e^{-\frac{1}{2}\|Y-X_{\tilde{S}} \beta_{\tilde{S}}\|_2^2 - \lambda \|\beta_{\tilde{S}}\|_{2,1}}} \right] \\
& = \log (D_\Pi/D_N) + \lambda E_{\mu_{\tilde{S}},D_{\tilde{S}}} (\|\beta_{\tilde{S}}\|_{2,1} - \|\beta_{0,\tilde{S}}\|_{2,1}),
\end{align*}
where $D_\Pi = \int_{\R^{\tilde{p}}} e^{-\frac{1}{2}\|Y-X_{\tilde{S}} \beta_{\tilde{S}}\|_2^2 - \lambda \|\beta_{\tilde{S}}\|_{2,1}} d\beta_{\tilde{S}}$ and
$D_N = \int_{\R^{\tilde{p}}} e^{-\frac{1}{2}\|Y-X_{\tilde{S}} \beta_{\tilde{S}}\|_2^2 - \lambda \|\beta_{0,\tilde{S}}\|_{2,1}} d\beta_{\tilde{S}}$
are the normalizing constants for the densities. Arguing exactly as in the proof of Lemma B.2 of \cite{RS22}, one can show that on the event $\T_1$ is holds that $\log (D_\Pi/D_N) \leq 2\lambda \Gamma^{1/2} \eps + \log 2$ if $4e^{1+\Gamma \log M - \kappa}\leq 1$. On $\T_1$, the second term in the last display is bounded by
\begin{equation}\label{eq:decomp1}
\begin{split}
\lambda E_{\mu_{\tilde{S}},D_{\tilde{S}}} \|\beta_{\tilde{S}} - \beta_{0,\tilde{S}}\|_{2,1} & \leq \lambda |\tilde{S}|^{1/2} E_{\mu_{\tilde{S}},D_{\tilde{S}}}  \|\beta_{\tilde{S}} - \beta_{0,\tilde{S}}\|_2 \\
&\leq \lambda \Gamma^{1/2} \left(  \|\mu_{\tilde{S}} - \beta_{0,\tilde{S}}\|_{2} + E_{0,D_{\tilde{S}}} \|\beta_{\tilde{S}}\|_{2} \right)
\end{split}
\end{equation}
using Cauchy-Schwarz.

Under $P_{\beta_0}$, so that $Y =^d X\beta_0 + \eps$, and using \eqref{eq:vb_parameters},
\begin{align*}
\|\mu_{\tilde{S}} - \beta_{0,\tilde{S}}\|_{2} \leq \| (X_{\tilde{S}}^T X_{\tilde{S}})^{-1} X_{\tilde{S}}^T X_{\tilde{S}^c} \beta_{0,\tilde{S}^c}\|_2 + \| (X_{\tilde{S}}^T X_{\tilde{S}})^{-1}X_{\tilde{S}}^T \eps \|_2.
\end{align*}
Arguing for this term as in Lemma B.2 of \cite{RS22}, one can show that the first term is bounded by $\Gamma^{1/2} s_0^{1/2} \eps/\widetilde{\phi}(|\tilde{S}|)^2$ on $\T_1$. Using that the $\ell_2$-operator norm of $(X_{\tilde{S}}^T X_{\tilde{S}})^{-1}$ is bounded by $1/(\|X\| \widetilde{\phi}(|\tilde{S}|))^2$ by \eqref{eq:lambda_min}, the second term in the last display is bounded by $\|X_{\tilde{S}}^T \eps\|_2/(\|X\|^2 \widetilde{\phi}(|\tilde{S}|)^2)$. But on $\T_1 \subset \T_0$,
$$\|X_{\tilde{S}}^T \eps\|_2^2 = \sum_{G_k \in \tilde{S}} \|X_{G_k}^T(Y-X\beta_0)\|_2^2 \leq 9|\tilde{S}| \|X\|^2 \log M.$$
Combining the above bounds thus yields 
\begin{align*}
\|\mu_{\tilde{S}} - \beta_{0,\tilde{S}}\|_{2} \leq \frac{\Gamma^{1/2} s_0^{1/2} \eps}{\widetilde{\phi}(\Gamma)^2}  +  \frac{3 \Gamma^{1/2} \sqrt{\log M}}{\|X\| \widetilde{\phi}(|\tilde{S}|)^2}
\end{align*}
on $\T_1$, thereby controlling the first term in \eqref{eq:decomp1}.

It remains only to bound the second term in \eqref{eq:decomp1}. Let $e_i$ denote the $i^{th}$ unit vector in $\R^{m_k}$, $i=1,\dots,m_k$, and let $\bar{e}_i$ denote its extension to $\R^{\tilde{p}}$ with unit entry in the $i^{th}$ coordinate of group $G_k$. Then $(X_{G_k}^T X_{G_k})_{ii} = \|X_{G_k}e_i\|_2^2 = \|X \bar{e}_i\|_2^2 \geq \|X\|^2 \widetilde{\phi}(1)^2$ for $i=1,\dots,m_k$. Therefore,
\begin{align*}
E_{0,D_{\tilde{S}}} \|\beta_{\tilde{S}}\|_{2}^2 = \Tr(D_{\tilde{S}}) = \sum_{G_k \in \tilde{S}} \sum_{i \in G_k} \frac{1}{(X_{\tilde{S}}^T X_{\tilde{S}})_{ii}} \leq \frac{\tilde{p}}{\|X\|^2 \widetilde{\phi}(1)^2}.
\end{align*}
Putting together of all the above bounds yields
$$(II) \leq 2\lambda \Gamma^{1/2} \eps + \log 2 + \frac{\lambda \Gamma}{\widetilde{\phi}(\Gamma)^2} \left( s_0^{1/2} \eps  +  \frac{3\sqrt{\log M}}{\|X\|} + \frac{\mm^{1/2}}{\|X\|} \right),$$
using that $\tilde{p} \leq \mm |\tilde{S}|$, $|\tilde{S}| \leq \Gamma$ and $\widetilde{\phi}(|\tilde{S}|) \leq \widetilde{\phi}(1) \leq 1$ for $\tilde{S} \neq \emptyset$.

Substituting the bounds for $(I)$ and $(II)$ just derived into \eqref{eq:KL1} and \eqref{eq:KL_split} then gives the first result. The second result follows from using that $\mm \leq K \log M$ under Assumption \ref{ass:theory}.
\end{proof}

\subsection{Proofs of Theorems \ref{thm:contraction_general} and \ref{thm:dimension_general}}\label{sec:proof_completion}

To complete the proofs, we apply Lemma \ref{lem:post_to_VB} with the event $\T_1$ defined in \eqref{eq:T1}
for suitably chosen constants $\Gamma,\eps,\kappa>0$.

\begin{lemma}\label{lem:T1}
Suppose that Assumption \ref{ass:theory} holds, the prior satisfies \eqref{eq:prior_condition}-\eqref{eq:lambda} and $s_n$ satisfies $\mm \log s_n \leq K' \log M$ for some $K'>0$. Then
$$\inf_{\substack{\beta_0\in \R^p: \phi(S_{\beta_0}) \geq c_0\\ |S_{\beta_0}|\leq s_n }} P_{\beta_0} (\T_1(\Gamma_{s_0},\eps_{s_0},\kappa_{s_0})) \to 1$$
as $n\to\infty$, where 
\begin{equation}\label{eq:T1_parameters}
\Gamma_{s_0} =  C_\Gamma s_0, \qquad 
\eps_{s_0} = C_\eps \frac{\sqrt{s_0 \log M}}{\|X\| \widetilde{\phi}(C_\phi s_0)}, \qquad
\kappa_{s_0} = (C_\Gamma s_0 + 1)\log M,
\end{equation}
and the constants have dependence $C_\Gamma = C_\Gamma(K,K',A_2,A_4,c_0)$, $C_\eps = C_\eps(K,K',A_1-A_4,c_0)$ and $C_\phi = C_\phi(K,K',A_2,A_4,c_0)$. Moreover, $4e^{1+\Gamma_{s_0} \log M-\kappa_{s_0}} \leq 1$ for $M>0$ large enough.
\end{lemma}

\begin{proof}
We consider each of the sets in $\T_1$ individually. In what follows, $C = C(K,K')$ and $c=c(K,K',A_2)$ will be constants that may change line-by-line but which will not depend on other parameters.

Applying Lemma \ref{lem:dimension} with $L= \frac{2}{A_4} (1+\log C + c + 144/c_0^2)$, where $C,c$ are the constants in that lemma, we have $E_{\beta_0} \Pi(\beta:|S_\beta| > (L+1) s_0 |Y)1_{\T_0} \leq e^{-s_0 \log M}\leq 1/4$ for $M$ large enough. Setting $C_\Gamma = L+1$, the second event in $\T_1$ is thus a subset of $\T_0$ for $\Gamma = \Gamma_{s_0}$ and all $\beta_0$ in the infimum in the lemma.

For the third event in $\T_1$, we apply Lemma \ref{lem:contract_full} with $L= \frac{2}{A_4}(1+\log C + c + 144/c_0^2 + C_\Gamma)$, where now $C,c$ are the constants in Lemma \ref{lem:contract_full}, to obtain
$$E_{\beta_0} \Pi \left( \left. \beta: \|\beta-\beta_0\|_2 \geq  \frac{H_0 \sqrt{(L+2)s_0\log M}}{ \|X\| \widetilde{\phi}((L+2)s_0)^2} \right| Y \right)1_{\T_0} \leq e^{-\kappa_{s_0}}.$$
Setting $C_\eps = H_0 \sqrt{L+2}$ and $C_\phi = L+2$ shows that the third event in $\T_1$ is also contained in $\T_0$ for these choices and all $\beta_0$ in the infimum in the lemma. It thus suffices to control the probability of $\T_0$, with tends to one uniformly in $\beta_0 \in \R^p$ by Lemma \ref{lem:T0}. Lastly, note that $4e^{1+\Gamma_{s_0} \log M-\kappa_{s_0}} = 4e^{1-\log M} \leq 1$ for $M$ large enough.
\end{proof}

\begin{proof}[Proof of Theorem \ref{thm:contraction_general}]
Write $\mathcal{B}_n = \mathcal{B}_{\rho_n,s_n} =\{\beta_0 \in \R^p:\phi(S_{\beta_0}) \geq c_0, ~ |S_{\beta_0}|\leq s_n, ~ \widetilde{\phi}(\rho_n |S_{\beta_0}|) \geq c_0 \}$ for the parameter set and let 
\begin{align*}
& \Omega_n= \left\{  \beta: \|X(\beta-\beta_0)\|_2 \leq  \frac{H_0 \rho_n^{1/2} \sqrt{s_0\log M}}{\bar{\phi}(\rho_n s_0)} \right\}
\end{align*}
for $H_0 = H_0(A_1,A_3,A_4)>0$ the constant in Lemma \ref{lem:contract_full}. Let $\T_1 = \T_1 = \T_1(\Gamma_{s_0},\eps_{s_0},\kappa_{s_0})$ be the event \eqref{eq:T1} with parameters \eqref{eq:T1_parameters}, so that Lemma \ref{lem:T1} gives $E_{\beta_0}\tilde{\Pi}(\Omega_n^c) = E_{\beta_0} \tilde{\Pi}(\Omega_n^c)1_{\T_1} + o(1)$, uniformly over $\beta_0 \in \mathcal{B}_n$.

Applying Lemma \ref{lem:contract_full} with $L+2 = \rho_n \to \infty$ gives that for $n$ large enough, $E_{\beta_0} \Pi(\Omega_n^c|Y)1_{\T_0} \leq Ce^{-c\rho_n s_0 \log M}$, where $C,c$ depend only on $K,K',A_2,A_4,c_0$. We can then use Lemma \ref{lem:post_to_VB} with event $A_n = \T_1 \subset \T_0$ and $\delta_n = c\rho_n s_0 \log M$ to obtain
\begin{align*}
E_{\beta_0} \tilde{\Pi} (\Omega_n^c) 1_{\T_1}  & \leq \frac{2}{c\rho_n s_0 \log M} \left[ \KL(\tilde{\Pi} \|\Pi(\cdot|Y))1_{\T_1} + C^{-c\rho_n s_0 \log M} \right].
\end{align*}
Since $\rho_n s_0 \log M \to \infty$ and $\tilde{\Pi}$ is by definition the KL-minimizer of the variational family to the posterior, we obtain that for any $Q \in \mathcal{Q}$,
\begin{align*}
E_{\beta_0} \tilde{\Pi} (\Omega_n^c)   & \leq \frac{2}{c\rho_n s_0 \log M} \KL(Q \|\Pi(\cdot|Y))1_{\T_1} + o(1),
\end{align*}
where the $o(1)$ term is uniform over $\beta_0 \in \mathcal{B}_n$. But by Lemma \ref{lem:KL}, there exists an element $Q \in \mathcal{Q}$ of the variational family such that
$$\KL (Q\| \Pi(\cdot|Y))1_{\T_1} \leq C \left( \log \frac{1}{\widetilde{\phi}(C_\Gamma s_0)} + \frac{1}{\widetilde{\phi}(C_\Gamma s_0)^2 \widetilde{\phi}(C_\phi s_0)}  \right) s_0 \log M,$$
where $C$ is uniform over $\beta_0 \in \mathcal{B}_n$ and $C_\Gamma,C_\phi$ have dependences specified in Lemma \ref{lem:T1}. But since $C_\Gamma,C_\phi$ are bounded under the hypothesis of the present theorem and $\rho_n \to \infty$, we have $\rho_n \geq C_\Gamma,C_\phi$ for $n$ large enough and hence $\widetilde{\phi}(C_\Gamma s_0),\widetilde{\phi}(C_\phi s_0)\geq \widetilde{\phi}(\rho_n s_0) \geq c_0$. Using this and the last two displays, we thus have
$$\sup_{\beta_0 \in \mathcal{B}_n} E_{\beta_0} \tilde{\Pi} (\Omega_n^c) \leq C/\rho_n + o(1) = o(1)$$
as $n\to\infty$, thereby establishing the first assertion. The second two assertions follow similarly for using the corresponding results in Lemma \ref{lem:contract_full}.
\end{proof}

\begin{proof}[Proof of Theorem \ref{thm:dimension_general}]
The proof follows similarly to the proof of Theorem \ref{thm:contraction_general}, using Lemma \ref{lem:dimension} to control the probability of the posterior set instead of Lemma \ref{lem:contract_full}.
\end{proof}

\end{document}